\documentclass[runningheads]{llncs}
\usepackage[T1]{fontenc}
\usepackage[english]{babel}

\usepackage{etoolbox}
\newtoggle{arxiv}
\newcommand{\ifarxivelse}[2]{\iftoggle{arxiv}{#1}{#2}}
\toggletrue{arxiv} 

\usepackage{amsmath}
\usepackage{graphicx}
\usepackage[colorlinks=true, allcolors=blue]{hyperref}
\usepackage{amssymb}
\usepackage{stmaryrd}
\usepackage{todonotes}
\usepackage{mathtools}
\usepackage{mdframed}
\usepackage{booktabs}
\usepackage{tikz}
\usepackage{csquotes}
\usetikzlibrary{automata,positioning}
\usepackage[altpo]{backnaur}
\usepackage[shortlabels]{enumitem}
\usepackage{multicol}
\usepackage{adjustbox}
\usepackage{wrapfig}
\usepackage{nicefrac}
\usepackage{siunitx}
\usepackage{colortbl}
\newlist{compactitem}{itemize}{3}
\setlist[compactitem]{topsep=2pt,partopsep=0pt,itemsep=0.5pt,parsep=0pt}
\setlist[compactitem,1]{label=\textbullet}
\setlist[compactitem,2]{label=---}
\setlist[compactitem,3]{label=*}
\newenvironment{tightcenter}{%
	\setlength\topsep{2pt}
	\setlength\parskip{2pt}
	\begin{center}
	}{%
	\end{center}
}

\usepackage[vlined,linesnumbered]{algorithm2e}
\SetAlCapSkip{2ex plus 0.5ex minus 1ex}
\SetEndCharOfAlgoLine{}
\SetArgSty{textrm}
\SetKwInOut{Input}{Input}\SetKwInOut{Output}{Output}
\definecolor{commentcolor}{RGB}{60,114,26}

\SetCommentSty{commentstyle}
\SetKw{Break}{break}
\SetKw{Continue}{continue}
\SetKwFor{ForEach}{foreach}{do}{}
\SetKwRepeat{DoWhile}{repeat}{while}
\SetKwRepeat{DoUntil}{repeat}{until}

\usepackage{cleveref} 

\newcounter{mycounter} 

\usepackage{ntheorem}
\newmdtheoremenv[%
    outerlinewidth=2,%
    roundcorner=10pt,%
    leftmargin=0,%
    rightmargin=0,%
    backgroundcolor=black!10,%
    outerlinecolor=blue!70!black,
    innertopmargin=5pt,%
    innerleftmargin=5pt,%
    innerrightmargin=5pt,%
    splittopskip=\topskip,%
    skipabove=5pt,%
    skipbelow=5pt,
]{leitmotif}{Guiding Principle}

\crefname{leitmotif}{guiding principle}{guiding principles}
\Crefname{leitmotif}{Guiding Principle}{Guiding Principles}

\newcommand{\genericcomment}[3]{\todo[color=#1,size=\tiny,fancyline,author=#2]{#3}\xspace}
\renewcommand{\genericcomment}[3]{} 

\newcommand{\ms}[1]{\genericcomment{red!35}{MS}{#1}\xspace}

\newcommand{\mw}[1]{\genericcomment{yellow!35}{MW}{#1}\xspace}

\newcommand{\tw}[1]{\genericcomment{blue!35}{TW}{#1}\xspace}

\DeclareMathOperator*{\argmax}{arg\,max}
\DeclareMathOperator*{\argmin}{arg\,min}

\newcommand{\nats}{\mathbb{N}}
\newcommand{\exnats}{{\overline{\nats}}} 
\newcommand{\floats}{\mathbb{F}}

\newcommand{\tool}[1]{\textsf{#1}\xspace}
\newcommand{\storm}{\tool{Storm}}
\newcommand{\mcsta}{\tool{mcsta}}
\newcommand{\prism}{\tool{PRISM}}
\newcommand{\isabelle}{\tool{Isabelle/HOL}}

\newcommand{\set}{X} 
\newcommand{\setalt}{Y} 
\newcommand{\po}{\preceq} 

\newcommand{\pwrel}[1]{\mathrel{\ddot{#1}}} 
\newcommand{\lub}[1]{\sup{#1}}
\newcommand{\glb}[1]{\inf{#1}}
\newcommand{\lfp}[1]{\operatorname{lfp}{#1}} 
\newcommand{\gfp}[1]{\operatorname{gfp}{#1}} 
\newcommand{\fp}[1]{\operatorname{fp}{#1}} 
\newcommand{\monFunGeneric}{\mathcal{F}}

\newcommand{\mdp}{\mathcal{M}}
\newcommand{\mdptuple}{(\states,\act,\prmdp)}
\newcommand{\states}{S}
\newcommand{\act}{{Act}}

\newcommand{\target}{T}
\newcommand{\prmdp}{P}
\newcommand{\pr}{\mathbb{P}} 
\newcommand{\opt}{\operatorname{opt}}

\newcommand{\prmax}{\pr^{\max}}
\newcommand{\prmin}{\pr^{\min}}
\newcommand{\propt}{\pr^{\opt}}
\newcommand{\prind}[1]{\pr^{#1}}
\newcommand{\reach}{\lozenge}
\newcommand{\strat}{\sigma}
\newcommand{\pre}{\mathit{Pre}}
\newcommand{\post}{\mathit{Post}}

\newcommand{\bellman}[1]{\mathcal{B}^{#1}} 
\newcommand{\modbellman}[1]{\tilde{\mathcal{B}}^{#1}} 
\newcommand{\bellmanmax}{\bellman{\max}}
\newcommand{\modbellmanmax}{\modbellman{\max}}
\newcommand{\bellmanmin}{\bellman{\min}}
\newcommand{\modbellmanmin}{\modbellman{\min}}
\newcommand{\bellmanopt}{\bellman{\opt}}
\newcommand{\bellmanind}[1]{\bellman{#1}}
\newcommand{\modbellmanind}[1]{\modbellman{#1}}
\newcommand{\vals}{x} 
\newcommand{\valsb}{y} 

\newcommand{\distop}[2]{\mathcal{D}^{#1}_{#2}} 
\newcommand{\distopvals}[2]{\mathcal{D}^{#1}_{\vals#2\uparrow}} 
\newcommand{\distoprvals}[2]{\mathcal{D}^{#1}_{\vals\downarrow}} 
\newcommand{\rank}{r} 

\newcommand{\distopmod}[1]{\overline{\mathcal{D}}^{#1}} 

\newcommand{\Eopt}{\mathbb{E}^{\opt}}
\newcommand{\E}{\mathbb{E}}
\newcommand{\Estrat}{\mathbb{E}^{\strat}}
\newcommand{\Emax}{\mathbb{E}^{\max}}
\newcommand{\Emin}{\mathbb{E}^{\min}}
\newcommand{\starinf}{* = \infty}
\newcommand{\starrho}{* = \rho}
\newcommand{\RgeqZero}{\mathbb{R}_{\geq 0}}
\newcommand{\RgeqZeroInf}{\overline{\mathbb{R}}_{\geq 0}}
\newcommand{\RgeqZeroInfStates}{\overline{\mathbb{R}}_{\geq 0}^{\states}}

\newcommand{\rew}{\textnormal{rew}}
\newcommand{\newtarget}{\target'}

\newcommand{\newtargetmin}{\target^{\min}}
\newcommand{\newtargetmax}{\target^{\max}}
\newcommand{\newtargetstrat}{\target^{\strat}}

\newcommand{\bellmanr}[1]{\mathcal{E}^{#1}} 
\newcommand{\modbellmanr}[1]{\tilde{\mathcal{E}}^{#1}}
\newcommand{\bellmanrmax}{\bellmanr{\max}}
\newcommand{\modbellmanrmax}{\modbellmanr{\max}}
\newcommand{\bellmanrmin}{\bellmanr{\min}}
\newcommand{\modbellmanrmin}{\modbellmanr{\min}}
\newcommand{\bellmanropt}{\bellmanr{\opt}}
\newcommand{\bellmanrind}[1]{\bellmanr{#1}}
\newcommand{\modbellmanrind}[1]{\modbellmanr{#1}}

\newcommand{\modbellmanrrhomin}{\bellmanrmintarget{\newtargetmin}}

\newcommand{\nopt}{\overline{\opt}}
\newcommand{\nmin}{\overline{\min}}
\newcommand{\nmax}{\overline{\max}}

\newcommand{\pos}{\textnormal{Pos}}

\newcommand{\bellmanropttarget}[1]{\bellmanr{\opt}_{#1}}
\newcommand{\bellmanrindtarget}[2]{\bellmanr{#1}_{#2}}
\newcommand{\bellmanrmintarget}[1]{\bellmanr{\min}_{#1}}
\newcommand{\bellmanrmaxtarget}[1]{\bellmanr{\max}_{#1}}
\newcommand{\distopmodtarget}[2]{\overline{\mathcal{D}}^{#1}_{#2}}
\newcommand{\newtargetmininput}{\target^{\min}_i}
\newcommand{\newtargetmaxinput}{\target^{\max}_i}
\newcommand{\newtargetstratinput}{\target^{\strat}_i}
\newcommand{\distoprrhovals}[2]{\mathcal{D}^{#1}_{\vals\uparrow}}

\newcommand{\indact}[1]{\act_{#1}^{\uparrow}}
\newcommand{\indactrew}[1]{\act_{#1}^{\uparrow}}
\newcommand{\coindactrew}[1]{\act_{#1}^{\downarrow}}

\newcommand{\colMinUpper}{blue!75!black}
\newcommand{\colMinLower}{green!40!black}
\newcommand{\colMaxUpper}{red!75!black}
\newcommand{\colMaxLower}{orange!75!black}

\usepackage{pgfplots}
\pgfplotsset{compat=newest}

\newtoggle{showplots}
\toggletrue{showplots}

\definecolor{plotred}{RGB}{255,0,0}
\definecolor{plotgreen}{RGB}{0,255,0}
\definecolor{plotblue}{RGB}{0,0,255}
\definecolor{plotyellow}{RGB}{230,230,0}
\definecolor{plotcyan}{RGB}{0,255,255}
\definecolor{plotorange}{RGB}{255,127,0}
\definecolor{plotpink}{RGB}{255,0,255}
\definecolor{plotlightgray}{RGB}{192,192,192}
\definecolor{plotdarkgray}{RGB}{128,128,128}
\definecolor{plotdarkred}{RGB}{128,0,0}
\definecolor{plotgreenyellow}{RGB}{128,128,0}
\definecolor{plotdarkgreen}{RGB}{0,128,0}
\definecolor{plotpurple}{RGB}{128,0,128}
\definecolor{plotteal}{RGB}{0,128,128}
\definecolor{plotdarkblue}{RGB}{0,0,128}
\definecolor{plotlightred}{RGB}{205,92,92}
\definecolor{plotlightblue}{RGB}{176,196,222}
\colorlet{color1}{plotred}
\colorlet{color2}{plotgreen}
\colorlet{color3}{plotblue}
\colorlet{color4}{plotyellow}
\colorlet{color5}{plotcyan}
\colorlet{color6}{plotorange}
\colorlet{color7}{plotpink}
\colorlet{color8}{plotlightgray}
\colorlet{color9}{plotdarkgray}
\colorlet{color10}{plotdarkred}
\colorlet{color11}{plotgreenyellow}
\colorlet{color12}{plotdarkgreen}
\colorlet{color13}{plotpurple}
\colorlet{color14}{plotteal}
\colorlet{color15}{plotdarkblue}
\colorlet{color16}{plotlightred}
\colorlet{color17}{plotlightblue}

\newcommand{\quantileplotxlabel}{Number of solved instances (out of \numbenchmarks)}
\newcommand{\quantileplotylabel}{Runtime (seconds)}
\newlength{\quantileplotwidth}
\newlength{\quantileplotheight}
\newcommand{\quantileplotlegendcols}{1}

\newcommand{\standardquantileplotlinestyle}{ultra thick}
\newcommand{\quantileplotlinestyle}{\standardquantileplotlinestyle}

\newcommand{\quantileplot}[8]{%
	\begin{tikzpicture}
	\begin{axis}[
	width=\quantileplotwidth,
	height=\quantileplotheight,
	xmin=#4,
	xmax=#5,
	ymin=#6,
	ymax=#7,
	ymajorgrids,
	ymode=log,
	axis x line=bottom,
	axis y line=left,
	unbounded coords=discard,filter discard warning=false, 
	xlabel=\quantileplotxlabel,
	ylabel=\quantileplotylabel,
	log ticks with fixed point, 
	yticklabel style={font=\scriptsize},
	scaled y ticks=false,
	xticklabel style={font=\scriptsize},
	legend columns=\quantileplotlegendcols,
	legend pos={#8},
	legend style={nodes={scale=0.75, transform shape},inner sep=1pt},
/pgfplots/legend image code/.code={\draw[mark repeat=2,mark phase=2,##1] plot coordinates {(0cm,0cm) (0.3cm,0cm)};}, 
	every axis plot/.append style={\quantileplotlinestyle},
	legend cell align={left}
	]
	\iftoggle{showplots}{
	\foreach \tool\color in {#2}{%
		\edef\loopbody{
			\noexpand\addplot[\color] table [x=n,y=\tool shifted, col sep=tab] {#1};
		}
		\loopbody
	}
	\legend{#3}}{\node[anchor=south west, align=center, red] {\huge NOT\\ COMPILED};}
	\end{axis}
	\end{tikzpicture}%
}

\newlength\scatterplotsize
\newcommand{\scatterplotAlg}[6]{%
	\begin{tikzpicture}
	\begin{axis}[
	width=\scatterplotsize,
	height=\scatterplotsize,
	axis equal image,
	xmin=1,
	ymin=1,
	ymax=1280,
	xmax=1280,
	xmode=log,
	ymode=log,
	axis x line=bottom,
	axis y line=left,
	xtick={2,4,8,16,32,64,128},
	xticklabels={2,4,8,16,,64},
	extra x ticks = {1,360,512,724},
	extra x tick labels = {,\raisebox{-5pt}{$\ge$360},inval,n/a},
	extra x tick style = {grid=major, xticklabel style = {font=\tiny,rotate=-90,anchor=west}},
	xticklabel style={font=\scriptsize,yshift=2pt},
	xlabel={#3},
	xlabel style={font=\scriptsize,yshift=15pt},
	ytick={2,4,8,16,32,64,128},
	yticklabels={2,4,8,16,32,64,128},
	extra y ticks = {1,360,512,724},
	extra y tick labels = {\scriptsize${\le}1$,\raisebox{-5pt}{$\ge$360},inval,n/a},
	extra y tick style = {grid = major, yticklabel style={font=\tiny}},
	yticklabel style={font=\scriptsize,xshift=2pt},
	ylabel={#5},
	ylabel style={font=\scriptsize,yshift=-12pt,xshift=-5pt},
	legend pos=north east,
	legend columns=3,
	legend style={nodes={scale=0.75, transform shape},inner sep=1.5pt, xshift=1mm, yshift=6mm},
	set layers,
	mark layer=axis background
	]
	
	\iftoggle{showplots}{\addplot[
	scatter,
	only marks,
	scatter/classes={
		P={mark=diamond*,plotorange,mark size=1.75},
		E={mark=asterisk,plotdarkgreen,mark size=1.75}
	},
	scatter src=explicit symbolic
	]%
	table [col sep=tab,x expr=#2,y expr=#4,meta=Prop] {#1};
	}{\node[anchor=south west, align=center, red] {\huge NOT\\ COMPILED};}
	\ifthenelse{\NOT\equal{#6}{false}}{\legend{reach. prob, exp. reward}}{}
	\addplot[no marks] coordinates {(0.01,0.01) (360,360) };
	\addplot[no marks, densely dotted] coordinates {(0.01,0.02) (150,360)};
	\addplot[no marks, densely dotted] coordinates {(0.02,0.01) (360,150)};
	\end{axis}
	\end{tikzpicture}
}

\newcommand{\scatterplotStates}[7]{%
\begin{tikzpicture}
	\begin{axis}[
	width=\scatterplotsize,
	height=\scatterplotsize,
	ymin=1,
	ymax=1280,
	xmode=log,
	ymode=log,
	axis x line=bottom,
	axis y line=left,
	xtick={1e1,1e3,1e5,1e7},
	xticklabel style={font=\scriptsize,yshift=2pt},
	xlabel=number of states,
	xlabel style={font=\scriptsize,yshift=1.6pt},
	ytick={2,4,8,16,32,64,128},
	yticklabels={2,4,8,16,32,64,128},
	extra y ticks = {1,360,512,724},
	extra y tick labels = {\scriptsize${\le}1$,\raisebox{-5pt}{$\ge$360},inval,n/a},
	extra y tick style = {grid = major, yticklabel style={font=\tiny}},
	yticklabel style={font=\scriptsize,xshift=2pt},
	ylabel={#7},
	ylabel style={font=\scriptsize,yshift=-12pt,xshift=-12pt},
	legend pos=north east,
	legend columns=3,
	legend style={nodes={scale=0.75, transform shape},inner sep=1.5pt, xshift=1mm, yshift=6mm},
	set layers,
	mark layer=axis background,
	table/col sep=tab
	]
	
\iftoggle{showplots}{
	\addplot[#2] table [y expr={\thisrow{#1} < 1 ? 1 : \thisrow{#1}},x=States]{plotdata/scatter.csv};
	\addplot[#5] table [y expr={\thisrow{#4} < 1 ? 1 : \thisrow{#4}},x=States]{plotdata/scatter.csv};
}{\node[anchor=south west, align=center, red] {\huge NOT\\ COMPILED};}
	\legend{#3, #6}
	\end{axis}
\end{tikzpicture}
}

\newcommand{\scatterplotStatesCert}[1]{%
\scatterplotStates{logs.cert.#1}{plotred, only marks, mark=+, mark size=1.75}{incl. parsing}
{logs.cert.#1veriftime}{plotblue, only marks, mark=x,, mark size=1.75}{excl. parsing}
{certification time}
}

\newcommand{\scatterplotStatesStorm}[4]{%
\scatterplotStates{logs.Storm.#1}{plotred, only marks, mark=+, mark size=1.75}{#2}
{logs.Storm.#3}{plotblue, only marks, mark=x,, mark size=1.75}{#4}
{generation time}
}

\usepackage{thm-restate} 

\AtBeginDocument{
	\setlength{\abovedisplayskip}{0.5\abovedisplayskip}
	\setlength{\abovedisplayshortskip}{0.65\abovedisplayshortskip}
	\setlength{\belowdisplayskip}{0.5\belowdisplayskip}
	\setlength{\belowdisplayshortskip}{0.65\belowdisplayshortskip}
	\setlength{\textfloatsep}{0.65\textfloatsep}
}

\makeatletter
\def\orcidID#1{\textsuperscript{\,\smash{\protect\raisebox{-1.25pt}{\href{http://orcid.org/#1}{\protect\includegraphics[scale=.8]{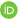}}}}}}
\makeatother

\begin{document}

\title{Fixed Point Certificates for Reachability and\\Expected Rewards in MDPs%
\thanks{This project has received funding from 
	the ERC CoG 863818 (ForM-SMArt), the Austrian Science Fund (FWF) 10.55776/COE12, 
	a KI-Starter grant from the Ministerium f\"ur Kultur und Wissenschaft NRW, 
	the DFG RTG 378803395 (ConVeY), 
	the EU's Horizon 2020 research and innovation programmes under the Marie Sklodowska-Curie grant agreement Nos.\ 101034413 (IST-BRIDGE) and 101008233 (MISSION), 
    and the DFG RTG 2236 (UnRAVeL). 
	Experiments were performed with computing resources granted by RWTH Aachen University under project rwth1632.}
}

\titlerunning{Certificates for Reachability and Expected Rewards in MDPs}

\author{Krishnendu Chatterjee\inst{1}\orcidID{0000-0002-4561-241X} \and
Tim Quatmann\inst{2}\orcidID{0000-0002-2843-5511} \and
Maximilian Sch\"affeler\inst{3}\orcidID{0000-0002-2612-2335} \and
Maximilian Weininger\inst{1}\orcidID{0000-0002-0163-2152} \and
Tobias Winkler\inst{2}\orcidID{0000-0003-1084-6408} \and
Daniel Zilken\inst{1,2}\orcidID{0009-0007-7221-2554}
}
\authorrunning{K. Chatterjee et al.}

\institute{%
	Institute of Science and Technology Austria, Klosterneuburg, Austria $\cdot$\\
	\email{\{Krishnendu.Chatterjee, Maximilian.Weininger\}@ist.ac.at}
	\and
	RWTH Aachen, Germany $\cdot$\\
	\email{\{tim.quatmann, tobias.winkler, daniel.zilken\}@cs.rwth-aachen.de}
	\and
	Technical University of Munich, Germany 
	$\cdot$ 
	\email{maximilian.schaeffeler@tum.de}
}

\maketitle

\begin{abstract}
The possibility of errors in human-engineered formal verification software, such as model checkers, poses a serious threat to the purpose of these tools.
An established approach to mitigate this problem are \emph{certificates}---lightweight, easy-to-check proofs of the verification results. 
In this paper, we develop novel certificates for model checking of \emph{Markov decision processes} (MDPs) with quantitative reachability and expected reward properties.
Our approach is conceptually simple and relies almost exclusively on elementary fixed point theory.
Our certificates work for \emph{arbitrary} finite MDPs and can be readily computed with little overhead using standard algorithms.
We formalize the soundness of our certificates in \isabelle and provide a \emph{formally verified certificate checker}.
Moreover, we augment existing algorithms in the probabilistic model checker \storm with the ability to produce certificates and demonstrate practical applicability by conducting the first formal certification of the reference results in the Quantitative Verification Benchmark Set.
\keywords{Probabilistic model checking \and Markov decision processes \and Certificates \and Reachability \and Expected rewards \and Proof assistant}
\end{abstract}


\section{Introduction}

\emph{Markov decision processes} (MDPs)~\cite{puterman,principles,Handbook} are \emph{the} model for sequential decision making in probabilistic environments.
Their many applications~\cite{white1993survey,QVBS} frequently require computing \emph{reachability probabilities} towards an \mbox{(un-)}desired system state, as well as the \emph{expected rewards} (or costs) accumulated until doing so.
\emph{MDP model checking} amounts to computing (approximations of) these quantities in a \emph{mathematically rigorous} way, with a formal guarantee of their correctness and precision.
Various mature MDP model checking tools such as \prism~\cite{prism}, \mcsta~\cite{modest}, and \storm~\cite{storm} exist.
\Cref{fig:exampleReach} shows an example MDP.

\emph{Who checks the model checker?}
The possibility of errors in complex, human-engineered formal verification tools is a delicate issue:
How \emph{formal} is a verification result produced by an \emph{informal}, i.e., unverified implementation?
We highlight four sources of errors: (i) classic implementation bugs, (ii) unintentionally unsound algorithms~\cite{DBLP:conf/cav/Baier0L0W17,HM18}, optimizations, and heuristics, (iii) numerical errors due to floating point arithmetic~\cite{arndFP}, and (iv) errors in third-party back end libraries or tools, e.g., commercial LP solvers~\cite{TACAS23}.


\emph{Certifying algorithms}~\cite{DBLP:journals/csr/McConnellMNS11}
are a paradigm for establishing trust in implementations.
A certifying algorithm produces a concise, easily verifiable proof---a \emph{certificate}---of its result.
The certificate can be checked independently, possibly even by an external, simpler program amenable to formal verification, or by a third party.
%
Formally verified certificate checkers are already employed in tool competitions on software verification~\cite{SVCOMP23} or SAT-solving~\cite{SATcomp}.
Existing proposals for certifying MDPs~\cite{DBLP:journals/jar/Holzl17,DBLP:conf/tacas/FunkeJB20,DBLP:phd/dnb/Jantsch22}, however, have some drawbacks (detailed further below) hindering wider adoption in the community 
and its competitions~\cite{qcomp19,qcomp20,qcomp23}.

\begin{tightcenter}
    \emph{The goal of this paper is to establish a new standard for certified MDP model checking, with a focus on applicability and extensibility.}
\end{tightcenter}

Our contributions towards this goal are as follows:
\begin{compactitem}
	\item We present \emph{fixed point certificates} for two-sided bounds on extremal reachability probabilities\ifarxivelse{}{ (\Cref{fig:certoverview})} and expected rewards\ifarxivelse{ (Tables~\ref{fig:certoverview}, \ref{fig:certoverviewrewards})}{}.
	Our certificates are sound and complete for \emph{arbitrary} finite MDPs without structural restrictions.
	\item We formalize the theory in \isabelle~\cite{DBLP:books/sp/NipkowPW02}, proving soundness of our certificates, and generate a \emph{formally verified certificate checker implementation}.
	\item We implement a certifying variant of~\cite{TACAS23} \emph{Interval Iteration}~\cite{DBLP:conf/cav/Baier0L0W17} with floating point arithmetic in \storm~\cite{storm}.
    Using this, we give certified reference results for the Quantitative Verification Benchmark Set~\cite{QVBS}.
\end{compactitem}

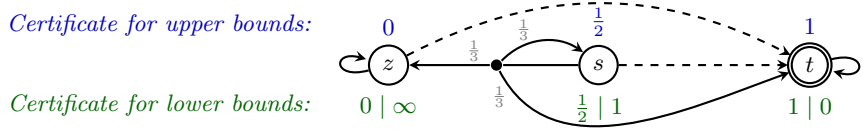
\begin{figure}[t]
    \centering
    \begin{tikzpicture}[every state/.style={minimum size=0},thick,>=stealth,node distance=22.5mm]
        \node[state,
            label=90:{\textcolor{\colMinUpper}{$\frac 1 2$}},%
            label=-90:{\textcolor{\colMinLower}{$\frac 1 2 \mid 1$}}%
        ] (s) {$s$};
        
        \node[state,accepting,right=of s,%
            label=90:{\textcolor{\colMinUpper}{$1$}},%
            label=-90:{\textcolor{\colMinLower}{$1 \mid 0$}}%
        ] (t) {$t$};
        
        \node[state,left=of s,%
            label=90:{\textcolor{\colMinUpper}{$0$}},%
            label=-90:{\textcolor{\colMinLower}{$0 \mid \infty$}}%
        ] (z) {$z$};

        \node[circle, fill,left=10mm of s,inner sep=1.5pt] (h) {};
        
        \node[state,white,left=25mm of z,%
            label=90:{\textcolor{\colMinUpper}{\textit{Certificate for upper bounds:}}},%
            label=-90:{\textcolor{\colMinLower}{\textit{Certificate for lower bounds:}}},%
        ] {$s$}; 

        \draw[] (s) edge (h);
        \draw[->] (h) edge[bend left=40] node[auto,gray,yshift=-3pt] {\tiny $\frac 1 3$} (s);
        \draw[->] (h) edge[out=-62.5,in=-160] node[near start,gray,xshift=-7mm,yshift=3mm] {\tiny $\frac 1 3$}(t);
        \draw[->] (h) edge (z) node[above,xshift=-3mm,yshift=-2pt,above,gray] {\tiny $\frac 1 3$};
        \draw[->,dashed] (s) edge (t);

        \draw[->] (t) edge[loop right] (t);
        \draw[->] (z) edge[loop left] (z);
        \draw[->,dashed] (z) edge[bend left=27.5] (t);
    \end{tikzpicture}%
    \caption{
        An MDP with states $\states = \{z,s,t\}$, two actions (distinguished by solid and dashed edges), uniform probabilities, and target set $T = \{t\}$.
        The annotations above and below each state are a certificate for \textcolor{black}{upper} and \textcolor{black}{lower} bounds on $\prmin(\reach \target)$, resp.
    }
    \label{fig:exampleReach}
\end{figure}


\emph{Extensibility} towards further properties is enabled by our simple, clean theory summarized as four \emph{guiding principles}: 
%
(GP\ref{leitmotif:knasterTarski})
We characterize the quantities of interest as a \emph{fixed point} of basic, easy-to-evaluate \emph{Bellman-type} operator~\cite{DBLP:journals/iandc/Bellman58}.
The fundamental certification mechanism is to use \emph{fixed point induction} for proving bounds on the least or greatest fixed point.
%
(GP\ref{leitmotif:rankingFunctions}) 
We certify \emph{qualitative} reachability properties using \emph{ranking functions} which are amenable to fixed point induction, too.
%
(GP\ref{leitmotif:uniqueFixedPoints}) 
As the basic Bellman operators frequently have undesired, spurious fixed points~\cite{atva14,HM18,lics23} (often related to \emph{end components}~\cite{dealfaro97}), we consider slight modifications requiring qualitative reachability information which we certify following GP\ref{leitmotif:rankingFunctions}.
%
(GP\ref{leitmotif:witnessStrat})
When GP\ref{leitmotif:uniqueFixedPoints} is insufficient or not applicable, we implicitly include a \emph{witness strategy} in our certificate. 

\emph{Technical challenges} still arise in concretely applying these guiding principles.
For instance, a key novelty of our paper is a ranking-function type certificate for \emph{not almost sure reachability} (\Cref{prop:certQualReachTwo}), which is surprisingly involved.

\begin{table}[t]
    \caption{
        Our reachability certificates.
        Sound and complete for arbitrary finite MDPs.
    }
    \label{fig:certoverview}
    \begin{adjustbox}{max width=\textwidth}
        \setlength{\tabcolsep}{4pt}%
        \renewcommand{\arraystretch}{1.05}%
        \begin{tabular}{l l l}        
            \toprule 
            \rowcolor{gray!4}\quad\textbf{Certificate} & \textbf{Condition(s)} & \textit{Explanation} \\ \midrule
            \rowcolor{\colMinUpper!4}\multicolumn{3}{l}{\textcolor{black}{\textbf{Upper} bounds on \textbf{minimal} reachability probabilities:} ~ $\forall s \in \states \colon \prmin_{s}(\reach\target) \leq \vals(s)$ \hspace*{2mm} \textcolor{gray}{[\Cref{thm:certupper}]}} \\ 
            \rowcolor{\colMinUpper!4}\quad$\vals \in [0,1]^\states$ & $\bellmanmin(\vals) \leq \vals$ & \textit{\textbf{min}-Bellman operator {\textbf{decreases}} value of {all} states}\\[2pt] 
            \rowcolor{\colMaxUpper!4}\multicolumn{3}{l}{\textcolor{black}{\textbf{Upper} bounds on \textbf{maximal} reachability probabilities:} ~ $\forall s \in \states \colon \prmax_{s}(\reach\target) \leq \vals(s)$ \hspace*{1mm} \textcolor{gray}{[\Cref{thm:certupper}]} } \\
            \rowcolor{\colMaxUpper!4}\quad$\vals \in [0,1]^\states$ & $\bellmanmax(\vals) \leq \vals$ & \textit{\textbf{max}-Bellman operator \textbf{decreases} value of {all} states}\\[2pt] 
            \rowcolor{\colMinLower!4}\multicolumn{3}{l}{\textcolor{black}{\textbf{Lower} bounds on \textbf{minimal} reachability probabilities:} ~  $\forall s \in \states \colon \prmin_{s}(\reach\target) \geq \vals(s)$ \hspace*{3mm} \textcolor{gray}{[\Cref{thm:certlower}]}} \\
            \rowcolor{\colMinLower!4}\quad$\vals \in [0,1]^\states$ & $\bellmanmin(\vals) \geq \vals$ & \textit{\textbf{min}-Bellman operator \textbf{increases} value of {all} states}\\
            \rowcolor{\colMinLower!4}\quad$\rank \in \exnats^\states$ & $\distop{\max}{}(\rank) \leq \rank$ & \textit{$\rank$ upper bounds \textbf{maximal} distances to $\target$} \\ 
            \rowcolor{\colMinLower!4}& $\vals(s) {>} 0 \implies \rank(s) {<} \infty$ & \textit{positive reachability necessitates finite distance} \\[2pt]
            \rowcolor{\colMaxLower!4}\multicolumn{3}{l}{\textcolor{black}{\textbf{Lower} bounds on \textbf{maximal} reachability probabilities:} ~ $\forall s \in \states \colon \prmax_{s}(\reach\target) \geq \vals(s)$ \hspace*{2mm} \textcolor{gray}{[\Cref{thm:certlower-nostrat}]}} \\
            \rowcolor{\colMaxLower!4}\quad$\vals \in [0,1]^\states$ & $\bellmanmax(\vals) \geq \vals$ & \textit{\textbf{max}-Bellman operator \textbf{increases} value of {all} states}\\
            \rowcolor{\colMaxLower!4}\quad$\rank \in \exnats^\states$ & $\distopvals{\min}{}(\rank) \leq \rank$ & \textit{$\rank$ upper bounds \textbf{min.} distances to $\target$ via \textbf{$x$-incr.\ actions}}\\ 
            \rowcolor{\colMaxLower!4}& $\vals(s) {>} 0 \implies \rank(s) {<} \infty$ & \textit{positive reachability necessitates finite distance} \\ 
            \bottomrule
        \end{tabular}
    \end{adjustbox}
\end{table}

\emph{Related work.}
Closest to our work are the previous proposals for certificates in MDP model checking:
\cite[Sec.~4]{DBLP:journals/jar/Holzl17} formally verifies a theory of certificates for reachability objectives, which is however limited to upper bounds on maximal and lower bounds on the minimal probabilities.
\cite{DBLP:conf/tacas/FunkeJB20} presents so-called \enquote{Farkas certificates} for reachability; however, it does not offer a formally verified certificate checker, is limited to MDPs without end components (ECs), and does not address certificate generation explicitly.
With similar limitations, \cite{DBLP:conf/qestformats/BaierCK24} provides Farkas certificates for \emph{multiple} reachability or mean payoff objectives, which can be computed via linear programming.
\cite{DBLP:phd/dnb/Jantsch22} suggests lifting the EC assumptions from~\cite{DBLP:conf/tacas/FunkeJB20} by certifying the full maximal EC decomposition.
In contrast, our certificates are more concise as they handle ECs using at most one ranking function.
Further,~\cite{DBLP:phd/dnb/Jantsch22} proposes certificates for expected rewards, but they require the target to be reached almost surely, an assumption we do not have to make.

\emph{Witnessing subsystems}~\cite{WJAKB14,DBLP:conf/tacas/FunkeJB20,SWITSS,DBLP:conf/qestformats/BaierCK24} are an alternative certification paradigm.
However, their verification requires more computational effort than the simple, state-wise operations needed for checking Farkas or fixed point certificates.
Still, they utilize similar ideas:
The backward reachable states in~\cite[Sec. 3.3.3]{WJAKB14} essentially use ranking functions, as do the constraints
in~\cite[Sec. 5.2.2]{DBLP:journals/jcss/JungesK0W21}.

The term \enquote{certifying algorithms} was coined in~\cite{DBLP:journals/siamcomp/KratschMMS06}.
Previous work on certificates for other verification problems includes~\cite{DBLP:conf/cav/Namjoshi01,DBLP:conf/fsttcs/PeledPZ01,DBLP:journals/tcs/KupfermanV05,DBLP:journals/informaticaSI/Debbi18,DBLP:conf/fossacs/KupfermanS21}.
Further, certificates were recently investigated in hardware verification~\cite{DBLP:conf/cav/YuBH20,DBLP:conf/fmcad/YuFBH22,DBLP:conf/fmcad/YuFBH23,DBLP:conf/ijcar/FroleyksYBH24} and approximate model counting~\cite{DBLP:conf/cav/TanYSMM24}.
Finally, we mention that \emph{Optimistic Value Iteration} (OVI)~\cite{DBLP:conf/cav/HartmannsK20,DBLP:conf/atva/AzeemEKSW22}, the supermartingales in~\cite{DBLP:conf/aaai/LechnerZCH22,DBLP:conf/cav/AbateGR24,DBLP:conf/cav/TakisakaZWL24}, the certificates for probabilistic pushdown systems from~\cite{DBLP:conf/tacas/WinklerK23,DBLP:conf/lics/WinklerK23}, and a recent strategy synthesis method for infinite MDPs~\cite{DBLP:journals/pacmpl/BatzBKW24} follow fixed point induction principles similar to GP\ref{leitmotif:knasterTarski}.


\emph{Paper outline.}
After the background on MDPs and fixed point theory (\Cref{sec:prelims}), we introduce ranking functions for qualitative reachability (\Cref{sec:qualreachandsafe}).
Building on this, we discuss quantitative reachability (\Cref{sec:certificates}, \Cref{fig:certoverview}) and expected rewards (\Cref{sec:ExpRew}, \ifarxivelse{\Cref{fig:certoverviewrewards}}{\cite[Tab. 2]{arxivversion}}).
We explain how to compute certificates (\Cref{sec:computing}) and report on experiments (\Cref{sec:expeval}).
Omitted pen-and-paper proofs are in \ifarxivelse{the appendix}{\cite[App.~C--E]{arxivversion}}.
\emph{All proofs regarding the soundness of the certificates, even standard results from the literature, are formalized in \isabelle.}

\section{Preliminaries}
\label{sec:prelims}

A \emph{Markov decision processes} (MDP) is a tuple $\mdp = (\states, \act, \prmdp)$ where $\states$ is a finite set of \emph{states}, $\act$ is a finite set of \emph{actions}, and $\prmdp \colon \states \times \act \times \states \to [0,1]$ is a \emph{transition probability function} with the property that $\sum_{s' \in \states} \prmdp(s,a,s') \in \{0,1\}$ for all $s \in \states$ and $a \in \act$.
For every $s \in \states$, the set of \emph{enabled} actions $a\in \act$ for which the above sum equals $1$ is written $\act(s)$.
It is required that $\act(s) \neq \emptyset$ for all $s \in \states$.
For $s \in \states$ and $a \in \act(s)$, we define the \emph{$a$-successors} of $s$ as $\post(s,a) = \{s' \in \states \mid \prmdp(s,a,s') > 0\}$.
Notice that our MDPs do not have a distinguished initial state.
See \Cref{fig:exampleReach} for an example MDP.

A (finite-state, discrete-time) \emph{Markov chain} (DTMC) is the special case of an MDP with $|\act(s)| =1$ for all $s \in \states$.
A (memoryless and deterministic) \emph{strategy}%
\footnote{%
    Aka.\ scheduler or policy.
    We do not define more general strategies as memoryless deterministic suffice for optimizing reachability probabilities and expected rewards.}
for an MDP $\mdp = (\states,\act,\prmdp)$ is a function $\strat \colon \states \to \act$ such that for all $s \in \states$ we have $\strat(s) \in \act(s)$.
We may apply $\strat$ to $\mdp$ to obtain the \emph{induced DTMC} $\mdp^\strat = (\states, \act, \prmdp^\strat)$ which, intuitively, only retains the actions chosen by $\strat$.
Formally, for all $s, s' \in \states$ we define $\prmdp^\strat(s,\sigma(s),s') = \prmdp(s,\strat(s),s')$, and $\prmdp^\strat(s,a,s') = 0$ for all $a \neq \strat(s)$.

\paragraph{Reachability and Expected Rewards.}%
\label{sec:exprewPrelims}%
Fix a DTMC $(\states,\act,\prmdp)$, a \emph{target} set $\target \subseteq\states$, and a \emph{reward function} $\rew \colon \states \to \RgeqZero$.
We define two random variables $\reach \target$ and $\rew^{\reach\target}$ taking values in $\{0,1\}$ and $\RgeqZeroInf$, respectively:
For $s_0s_1 ... \in \states^\omega$ an infinite path, we set $\reach \target(s_0s_1\ldots) = 1$ if and only if (iff) $\exists i \in \nats \colon s_i \in \target$.
Moreover:
\begin{align*}
    \rew^{\reach\target}(s_0s_1\ldots)
    ~=~
    \begin{cases}
        \sum_{k = 0}^{\min \{i \mid s_i \in \target\}} \rew(s_k) & \text{if } \exists i \in \nats \colon s_i \in \target \\
        * & \text{else}
    \end{cases}
\end{align*}
We consider both options $* = \infty$ and $* = \sum_{k =0}^\infty\rew(s_k)$~\cite{DBLP:journals/fmsd/ChenFKPS13}.
We focus on the former in the main body, as it is standard in the literature~\cite[Def.~10.71]{principles} and tool competitions~\cite{qcomp19,qcomp20}; we treat the latter in \ifarxivelse{\Cref{app:sec:exprewStarRho}}{\cite[App.~F]{arxivversion}}.
Intuitively, with $* = \infty$, $\rew^{\reach\target}$ assigns $\infty$ to paths that never reach $\target$.
The other paths receive the sum of rewards collected until reaching $\target$ for the first time.
Given a state $s \in \states$, we define the \emph{reachability probability} $\pr_{s}(\reach\target)$ from $s$ to $\target$ as the expected value (Lebesgue integral) of $\reach\target$ w.r.t.\ the probability measure $\pr_s$ on infinite paths of the DTMC with initial state fixed to $s$, see~\cite[Ch.~10]{principles} for the construction of $\pr_s$.
Similarly, the \emph{expected reward} $\E_s(\rew^{\reach\target})$ from $s$ to $\target$ is the expected value of $\rew^{\reach\target}$.
When $\rew$ is clear from context, we write $\E_s(\reach\target)$ instead of $\E_s(\rew^{\reach\target})$.

Finally, given an MDP $\mdp = (\states, \act, \prmdp)$, a state $s \in \states$, a target set $\target\subseteq\states$, a reward function $\rew \colon \states \to \RgeqZero$, and $\opt \in \{\min, \max\}$ we define the \emph{optimal reachability probability} $\propt_{s}(\reach \target) = \opt_{\strat} \prind{\strat}_{s}(\reach \target)$ and the \emph{optimal expected reward} $\Eopt_{s}(\rew^{\reach \target}) = \opt_{\strat} \Estrat_{s}(\rew^{\reach \target})$, where $\prind{\strat}_s(\reach \target)$ and $\Estrat_{s}(\rew^{\reach \target})$ are the reachability probabilities and expected rewards in the induced DTMC $\mdp^\strat$.

\paragraph{Fixed Point Theory.}%
\label{sec:prelims:fixedpoint}%
A \emph{partial order} on a set $\set$ is a binary relation $\po$ that is reflexive, transitive, and antisymmetric;
in this case, the tuple $(\set,\po)$ is called a \emph{poset}.
Given a poset $(\set, \po)$, we call $a \in \set$ an \emph{upper bound} on $\setalt \subseteq \set$ if for all $b \in \setalt$ we have $b \po a$.
If an upper bound $a$ on $\setalt$ is minimal among all upper bounds, it is the unique \emph{supremum} (or least upper bound) and denoted $\sup{\setalt}$.
\emph{Lower bounds} and \emph{infima} (or greatest lower bounds) are defined analogously.

The poset $(\set, \po)$ is a \emph{complete lattice} if $\sup \setalt$ and $\inf \setalt$ exist for every $\setalt \subseteq \set$.
Every complete lattice has a least and greatest element $\sup \emptyset$ and $\inf \emptyset$, respectively.
The following complete lattices are of interest in this paper:
\begin{compactitem}
    \item $(\exnats, \leq)$ where $\exnats = \nats \cup \{\infty\}$ are the \emph{extended natural numbers} and $\leq$ is the usual order on $\nats$ extended by $a \leq\infty$ for all $a \in \exnats$.
    Notice that for every $\setalt \subseteq \nats$, $\sup \setalt = \infty$  iff $\setalt$ is infinite.
    \item Similarly, $(\RgeqZeroInf, \leq)$, with $\RgeqZeroInf = \RgeqZero \cup \{\infty\}$ the \emph{extended non-negative reals}, is a complete lattice.
    For every $\setalt \subseteq \RgeqZero$, $\sup\setalt = \infty$ iff $\setalt$ is unbounded.
    \item $([0,1], \leq)$, the totally ordered set of real probabilities.
    \item If $(\set, \po)$ is an arbitrary complete lattice, then for all sets $\states$, $(\set^\states, \pwrel\po)$ is a complete lattice, where $\set^\states = \{f \mid f \colon \states \to \set\}$ is the set of functions from $\states$ to $\set$ and the partial order $\pwrel\po$ is defined as $f \pwrel\po g \iff \forall s \in \states \colon f(s) \po g(s)$.
    In the following, we overload notation and write $\po$ instead of $\pwrel\po$.
    For example, if $\states$ is the set of states of an MDP, then we can think of $([0,1]^\states, \leq)$ as the poset of \enquote{probability vectors} indexed by $\states$, partially ordered entry-wise.
\end{compactitem}

Let $(\set, \po)$ be a poset.
A function $\monFunGeneric \colon \set \to \set$ is called \emph{monotone} if $\forall a, b \in \set \colon a \po b \implies \monFunGeneric(a) \po \monFunGeneric(b)$.
The following is \emph{the} key tool of this paper:

\begin{theorem}[Knaster-Tarski]%
    \label{thm:knastertarski}%
    Let $(\set, \po)$ be a complete lattice and \mbox{$\monFunGeneric \colon \set \to \set$} be monotone.
    Then, the set of \emph{fixed points} $(\{a \in \set \mid \monFunGeneric(a) = a\}, \po)$ is also a complete lattice.
    In particular, $\monFunGeneric$ has a \emph{least} and a \emph{greatest} fixed point given by $\lfp{\monFunGeneric} = \glb{\{a \in \set \mid \monFunGeneric(a) \po a\}}$ and, dually, $\gfp{\monFunGeneric} = \lub{\{a \in \set \mid a \po \monFunGeneric(a)\}}$.
    As a consequence, the following \emph{fixed point induction rules} are sound:
    $\forall a \in \set \colon$
    \begin{compactitem}
        \item $\monFunGeneric(a) \po a \implies \lfp{\monFunGeneric} \po a$ \hfill\textnormal{(fixed point induction)}
        \item $a \po \monFunGeneric(a) \implies a \po \gfp{\monFunGeneric}$ \hfill\textnormal{(fixed point co-induction)}
    \end{compactitem}
\end{theorem}
Elements $a \in \set$ with $\monFunGeneric(a) \po a$ (or $a \po \monFunGeneric(a)$) are called \emph{\mbox{(co-)}inductive}.

\begin{leitmotif}[Fixed Point Induction]
    \label{leitmotif:knasterTarski}
    We apply the theorem of Knaster-Tarski to monotone operators of the form $\monFunGeneric \colon \set^\states \to \set^\states$, where $\set$ is a complete lattice and $\states$ is finite, to certify upper bounds on $\lfp{\monFunGeneric}$ and lower bounds on $\gfp{\monFunGeneric}$.
    We call such $\monFunGeneric$ \emph{Bellman-type operators}.
\end{leitmotif}

Throughout the rest of the paper, we fix an MDP $\mdp = \mdptuple$, a set of target states $\target \subseteq \states$ and, for \Cref{sec:ExpRew}, a reward function $\rew\colon \states \rightarrow \RgeqZero$.
Moreover, we let $\opt\in \{\min, \max\}$ and write $\nmin = \max$ and $\nmax = \min$.

\section{Certifying Qualitative Reachability}
\label{sec:qualreachandsafe}

Most of the certificates presented in the forthcoming \Cref{sec:certificates,sec:ExpRew} enclose a certificate for a \emph{qualitative} reachability property, e.g., $\propt(\reach\target) > 0$ or $\propt(\reach\target) < 1$.
Our approach to this is summarized as follows:
\begin{leitmotif}[Ranking Functions]
    \label{leitmotif:rankingFunctions}
    To certify qualitative reachability properties, we rely on \emph{ranking functions} formalized via appropriate operators capturing certain \emph{distances} in the MDP when viewed as a graph.
\end{leitmotif}

\begin{definition}[Distance Operator]
    \label{def:distop}%
    Let $\mdptuple$, $\target$, and $\opt$ be the fixed MDP, target set and optimization direction, respectively. (We omit these quantifications in the rest of the paper.)
    We define the following \emph{distance operator}:
    \begin{align*}
        \distop{\opt}{} \colon
        ~
        \exnats^\states \to \exnats^\states,
        ~
        \distop{\opt}{}(\rank)(s) =
        \begin{cases}
        	0 & \text{ if } s \in \target\\
            1 ~+~ \underset{a \in \act(s)}{\opt} ~ \underset{s' \in \post(s,a)}{\min}  \rank(s') & \text{ if } s \in \states\setminus\target
        \end{cases}
    \end{align*}
\end{definition}
$\distop{\opt}{}$ is a monotone Bellman-type operator on the complete lattice $(\exnats^\states, \leq)$ and thus has a least fixed point by \Cref{thm:knastertarski}.
In fact, we even have the following:
\begin{restatable}[Unique Fixed Point]{lemma}{distopuniquefp}%
    \label{lem:distopuniquefp}%
    $\distop{\opt}{}$ has a unique fixed point.
    \tw{proof: \Cref{app:distopuniquefp}}
\end{restatable}
\noindent Intuitively, if $\rank = \fp{\distop{\min}{}}$, then $\rank(s)$ represents the length of a shortest path from every state $s \in \states$ to $\target$, or $\rank(s) = \infty$ if $\target$ is not reachable from $s$.
For $\rank = \fp{\distop{\max}{}}$, $\rank(s)$ can be seen as the shortest path in the DTMC induced by a strategy that aims to avoid $\target$ or reach it as late as possible.
We formalize this intuition in \Cref{lem:distopequiv} (using the notation $\nmin = \max$ and $\nmax = \min$), and then in \Cref{prop:certQualReach} apply \Cref{leitmotif:knasterTarski} to certify positive reachability.
\begin{restatable}{lemma}{distopequiv}
	\label{lem:distopequiv}%
	Let $\rank = \fp{\distop{\opt}{}}$.
	Then for all $s \in \states$, $\rank(s) = \infty \iff \pr_s^{\nopt}(\lozenge \target) = 0$.
    \tw{proof: \Cref{app:rankandqualreach}}
\end{restatable}
\begin{restatable}[Certificates for $\propt(\reach\target) \!>\! 0$]{proposition}{certQualReach}%
	\label{prop:certQualReach}%
    A function $\rank \in \exnats^\states$ is called a \emph{valid certificate for positive $\opt$-reachability} if $\distop{\nopt}{}(\rank) \!\leq\! \rank$.
	If $\rank$ is valid, then $\forall s \!\in\! \states \colon \rank(s) \!<\! \infty \!\implies\! \pr_s^{\opt}(\lozenge \target) \!>\! 0$.
    \tw{proof in \Cref{app:certQualReach}}
\end{restatable}

\begin{example}%
    \label{ex:certPosReach}%
    Consider the MDP in \Cref{fig:exampleReach} on page~\pageref{fig:exampleReach}.
    The values on the bottom right of the states constitute a valid certificate $\rank$ for positive $\min$-reachability.
    To check that $\distop{\nmin}{}(\rank) = \distop{\max}{}(\rank)\leq \rank$ is indeed true, we verify the following:
    \[
        \distop{\max}{}(\rank)(s)
        ~=~
        1 + \max\Big\{\, \textcolor{gray}{\underbrace{\textcolor{black}{\min\{0,1,\infty\}}}_{\text{solid action}}} \,,\, \textcolor{gray}{\underbrace{\textcolor{black}{\phantom{\{}0\phantom{\}}}}_{\text{dashed action}}} \,\Big\}
        ~=~
        1 + 0
        ~\overset{\checkmark}{\leq}~
        1
        ~=~
        \rank(s) ~,
    \]
    and similar for $z$ and $t$.
    As $\rank(s), \rank(t) < \infty$, we conclude $\prmin_s(\lozenge \target), \prmin_t(\lozenge \target) > 0$.
\end{example}

\begin{remark}[Certificates for $\propt(\reach\target) \!=\! 0$]%
    \label{rem:certQualReachZero}%
    While we do not need this in our paper, it is instructive to
    notice that with \Cref{lem:distopuniquefp,lem:distopequiv} we can also certify \emph{zero reachability probability}:
    By Knaster-Tarski, any $\rank$ with $r \leq \distop{\nopt}{}(\rank)$ witnesses $\rank \leq \fp{\distop{\nopt}{}}$, hence if $r(s) = \infty$ for a state $s$, then $\propt_s(\reach\target)= 0$.
\end{remark}

\emph{Certificates for non-almost-sure (a.s.) reachability}, i.e., $\propt(\reach\target) \!<\! 1$, are needed in \Cref{sec:ExpRew}. Ranking function-based certificates for this property are---perhaps surprisingly---more involved.
In \Cref{def:distopmod} below we define a \emph{complementary distance operator} that captures (approximately) the distance to the states $Z$ from which $\target$ is avoided surely, i.e., $\propt_s(\reach\target) = 0$ for all $s\in Z$.
By \Cref{lem:distopequiv}, finite distance to $Z$ witnesses positive $\nopt$-reachability of $Z$ and thus non-a.s.\ $\opt$-reachability of $\target$.
A major complication is that $Z$ is not given explicitly.
We address this by (i) considering the \emph{least} fp, and (ii) letting the operator only increment the distance if two successors do not have the same rank. For this, we use \emph{Iverson bracket} notation: $[\varphi] = 1$ if $\varphi$ is true; $[\varphi] = 0$, else.
Together, (i) and (ii) ensure that $s \in \states$ has rank $0$ in the $\lfp{}$ if and only if $s \in Z$.
%
\begin{definition}[Complementary Distance Operator]%
	\label{def:distopmod}%
	We define the \emph{complementary distance operator} 
    $\distopmod{\opt} \colon \exnats^\states \to \exnats^\states$, with
	\begin{align*}
		\distopmod{\opt}(\rank)(s)
        ~=~
		\begin{cases}
			\infty & \hspace{-20mm} \text{ if } s \in \target\\[2pt]
			 \underset{a \in \act(s)}{\opt} \bigg( \underset{s' \in \post(s,a)}{\min}  \rank(s') \,+\, \big[\exists u,v \!\in\! \post(s,a)\colon \rank(u)\!\neq\!\rank(v)\big] \bigg)  & \\[-3pt]
             & \hspace{-20mm} \text{ if } s \in \states\setminus\target
		\end{cases}
	\end{align*}
\end{definition}
Note that unlike the distance operator $\distop{\opt}{}$ from \Cref{def:distop}, $\distopmod{\opt}$ \emph{does not have a unique $\fp{}$}: The constant $\rank = \vec\infty$ is always a trivial fixed point.
\begin{restatable}{lemma}{distopmodEquiv}%
	\label{lem:distopmodequiv}%
	Let $\rank = \lfp{\distopmod{\opt}}$.
	Then for all $s \in \states$, $\rank(s) = \infty \iff \pr_s^{\opt}(\lozenge \target) = 1$.
\end{restatable}
\begin{restatable}[Certificates for $\propt(\reach\target) \!<\! 1$]{proposition}{certQualReachTwo}%
	\label{prop:certQualReachTwo}%
	A function $\rank \!\in\! \exnats^\states$ is called a \emph{valid certificate for non-a.s.\ $\opt$-reachability} if $\distopmod{\opt}(\rank) \!\leq\! \rank$.
    If $\rank$ is valid, then $\forall s \!\in\! \states \colon \rank(s) \!<\! \infty \!\implies\! \pr_s^{\opt}(\lozenge \target) \!<\! 1$.
    \tw{proof: \Cref{app:certQualReachTwo}}
\end{restatable}

\begin{remark}[Certificates for $\propt(\reach\target) \!=\! 1$]%
    Since $\distopmod{\opt}$ does \emph{not} have a unique fp, we cannot use the trick from \Cref{rem:certQualReachZero} to certify $\propt(\reach\target) \!=\! 1$ with ranking functions.
    Sections \ref{sec:certreach:lowermin} and \ref{sec:certreach:lowermax} present certificates for general lower bounds.
\end{remark}

\section{Certificates for Quantitative Reachability}
\label{sec:certificates}

This section presents our certificates for bounds on minimal and maximal reachability probabilities (\Cref{fig:certoverview}).
They are characterized via a \emph{Bellman operator}:

\begin{definition}[Bellman Operator for Reachability]
    \label{def:bellmanReach}%
    We define the Bellman operator for reachability
    $\bellmanopt \colon [0,1]^\states \to [0,1]^\states$ as usual:
    \begin{align*}
        \bellmanopt(\vals)(s)
        ~=~
        \begin{cases}
            1 & \text{ if } s \in \target\\
            \underset{a \in \act(s)}{\opt} \sum\limits_{s' \in \post(s,a)} \prmdp(s, a, s') \cdot \vals(s') & \text{ if } s \in \states \setminus \target 
        \end{cases}
    \end{align*}
\end{definition}
Similar to $\distop{\opt}{}$ from \Cref{sec:qualreachandsafe}, $\bellmanopt$ is a monotone function on the complete lattice $([0,1]^\states, \leq)$.
Thus, $\bellmanopt$ has a least fixed point by \Cref{thm:knastertarski}.

\begin{theorem}[{\cite[Sec. 3.5]{DBLP:conf/spin/ChatterjeeH08}}]
    \label{thm:reachislfp}
    For all $s \in \states$, $(\lfp{\bellman{\opt}})(s) = \propt_s(\lozenge \target)$.
\end{theorem}

We stress that $\bellmanopt$ has multiple fixed points in general.
For instance, $\vals = \vec 1$ is always a trivial fixed point.
\Cref{thm:reachislfp} states that the reachability probabilities are characterized as the \emph{least} fixed point. 
    
\subsection{Upper Bounds on Optimal Reachability Probabilities}
\label{sec:certreach:upper}

Following \Cref{leitmotif:knasterTarski}, we obtain the following by \Cref{thm:reachislfp}:

\begin{proposition}[Certificates for Upper Bounds on $\propt(\reach\target)$]
    \label{thm:certupper}%
    A probability vector $\vals \in [0,1]^\states$ satisfying $\bellmanopt(\vals) \leq \vals$ is a \emph{valid certificate for upper bounds on $\opt$-reachability}.
    If $\vals$ is valid, then $\forall s \in \states \colon \propt_s(\lozenge \target) \leq \vals(s)$.
\end{proposition}

\begin{example}
    \label{ex:certMinUpper}
    We verify that the numbers $\vals$ above the states in \Cref{fig:exampleReach} on \cpageref{fig:exampleReach} are a valid certificate for upper bounds on $\min$-reachability:
    For $s$ we check
    \[
        \bellmanmin(\vals)(s)
        ~=~
        \min \left\{ \tfrac 1 3 \cdot \textcolor{\colMinUpper}{0} + \tfrac 1 3 \cdot \textcolor{\colMinUpper}{\tfrac 1 2} + \tfrac 1 3 \cdot \textcolor{\colMinUpper}{1} ~,~ 1 \cdot \textcolor{\colMinUpper}{1} \right\}
        ~=~
        \min \left\{ \tfrac 1 2 , 1 \right\}
        ~\overset{\checkmark}{\leq}~
        \tfrac 1 2 
        ~=~
        \vals(s)
        ~,
    \]
    and similar for $z$ and $t$.
    Thus \Cref{thm:certupper} yields $\prmin_{s}(\reach\target) \leq \frac 1 2$.
    This particular certificate remains valid when changing $x(s)$ to any probability in $[\frac 1 2, 1]$.
    In general, however, increasing individual values may break inductivity. 
\end{example}

\subsection{Lower Bounds on Minimal Reachability Probabilities}
\label{sec:certreach:lowermin}

With \Cref{thm:knastertarski}, we can only certify lower bounds on \emph{greatest} fixed points.
Lower bounds on reachability probabilities---which constitute the \emph{least} fixed point of $\bellmanopt$---are thus more involved.\ifarxivelse{\footnote{Or, as the authors of \cite{DBLP:journals/pacmpl/HarkKGK20} put it: \enquote{Aiming low is harder.}}}{}
We propose to tackle this situation as follows:

\begin{leitmotif}[Modified Bellman Operators]
    \label{leitmotif:uniqueFixedPoints}
    We often \emph{modify} a basic Bellman-type operator to restrict its set of fixed points and enforce a certain extremal (i.e., least or greatest) fixed point of interest.
\end{leitmotif}
%

We now focus on $\min$-reachability first and modify $\bellmanmin$ as follows:
\begin{align*}
    \modbellmanmin \colon
    [0,1]^\states \to [0,1]^\states,
    ~
    \modbellmanmin(\vals)(s)
    =
    \begin{cases}
        \bellmanmin(\vals)(s) & \text{ if } \prmin_s(\lozenge\target) > 0 \\
        0 & \text{ if } \prmin_s(\lozenge\target) = 0 
    \end{cases}
\end{align*}

\begin{lemma}[{Unique Fixed Point~\cite[Thm.~10.109]{principles}}]
    \label{thm:uniquefp}
    $\modbellmanmin$ has a unique fixed point $\fp{\modbellmanmin} = \lfp{\bellmanmin}$.
\end{lemma}

By \Cref{thm:uniquefp} and \Cref{thm:reachislfp}, any probability vector $\vals \leq \modbellmanmin(\vals)$ witnesses that $\vals(s) \leq \prmin_s(\lozenge \target)$ for all $s \in \states$.
However, evaluating $\modbellmanmin(\vals)$ is not straightforward as it requires determining, for each $s \in \states$ whether $\prmin_{s}(\reach\target) > 0$.
Hence, we include an additional certificate for positive reachability from \Cref{sec:qualreachandsafe}:

\begin{restatable}[Certificates for Lower Bounds on $\prmin(\reach\target)$]{proposition}{certlower}%
    \label{thm:certlower}%
    \newline
    A tuple of probability vector and ranking function $(\vals, \rank) \in [0,1]^\states \times \exnats^\states$ is a \emph{valid certificate for lower bounds on $\min$-reachability} if 
    \begin{align*}
        1)~\distop{\max}{}(\rank) \leq \rank,
        \quad
        2)~\vals \leq \bellmanmin(\vals),
        \quad
        3)~\forall s \in \states\setminus\target \colon \vals(s) > 0 \implies \rank(s) < \infty.
    \end{align*}
    If $(\vals, \rank)$ is valid, then $\forall s \in \states \colon \prmin_s(\lozenge \target) \geq \vals(s)$.\\
\end{restatable}
%
\begin{example}
    We apply \Cref{thm:certlower} to the MDP in \Cref{fig:exampleReach}.
    The pairs $\vals(v) \mid \rank(v)$ below each state $v$ constitute a valid certificate $(\vals, \rank)$ for lower bounds on $\min$-reachability.
    Indeed, we have shown in \Cref{ex:certPosReach} that it satisfies Condition 1) $\distop{\max}{}(\rank) \leq \rank$.
    Condition 2)~$\vals \leq \bellmanmin(\vals)$ holds as well; in fact, we even have $\vals = \bellmanmin(\vals)$, see \Cref{ex:certMinUpper}.
    For the additional Condition 3), notice that $s$ is the only state in $\states \setminus \target$ with $\vals(s) > 0$, and that $\rank(s) < \infty$ holds as required.
    We conclude that $\prmin_{s}(\reach\target) \geq \textcolor{\colMinLower}{\frac{1}{2}}$.
\end{example}

\subsection{Lower Bounds on Maximal Reachability Probabilities}
\label{sec:certreach:lowermax}

Our approach for lower bounds on $\prmin(\reach\target)$ from \Cref{sec:certreach:lowermin} does not immediately extend to $\max$-reachability because $\modbellmanmax$ (a modification of $\bellmanmax$ analogous to $\modbellmanmin$) does \emph{not} have a unique fixed point in general, see \ifarxivelse{\Cref{app:remarkbellmanmax}}{\cite[App.~D.3]{arxivversion}} for a concrete counterexample.
This problem is caused by end components~\cite[Def.~3.13]{dealfaro97}.
Towards a solution, we observe that, essentially by definition,
\begin{align*}
    \forall s\in \states \colon\quad \prmax_s(\reach\target) \geq \vals(s) \iff \exists \text{ Strategy } \strat \colon \prind{\strat}_s(\reach\target) \geq \vals(s)
    ~.
\end{align*}
In words, a lower bound on a $\max$-reachability probability is always witnessed by some strategy.%
\footnote{Dually, an upper bound on a $\min$-reachability probability is also witnessed by a strategy, but our corresponding certificates from \Cref{thm:certupper} do not rely on this.}
Hence we adopt the following:

\begin{leitmotif}[Witness Strategies]
    \label{leitmotif:witnessStrat}
    In some cases, especially when progress towards a target is required, it is helpful to certify a witness strategy.
\end{leitmotif}

\subsubsection{Certificates with an Explicit Witness Strategy.}

Recall from \Cref{sec:prelims} that given a strategy $\strat \colon \states \to \act$ for MDP $\mdp$, we can consider the induced DTMC $\mdp^\strat$.
We write $\bellmanind{\strat}$ for the Bellman operator associated with $\mdp^\strat$ (notice that a DTMC is just a special case of an MDP).
Further, we let $\modbellmanind{\strat}$ be the corresponding modified Bellman operator.
By \Cref{thm:reachislfp} and \Cref{thm:uniquefp}:

\begin{lemma}
    $\modbellmanind{\strat}$ has a unique fixed point $(\fp{\modbellmanind{\strat}})(s) = \prind{\strat}_s(\reach \target)$ for all $s \in \states$.
\end{lemma}

Thus, we can certify lower bounds similar to \Cref{thm:certlower} (we write $\distop{\strat}{}$ for the distance operator $\distop{\opt}{}$ in the DTMC induced by $\strat$):

\begin{restatable}[Certificates for Lower Bounds on $\prmax(\reach\target)$ \!+\! Strategy]{proposition}{certlowerstratwitness}%
    \label{thm:certlowerstratwitness}%
    A triple $(\vals, \rank, \strat) \in [0,1]^\states \times \exnats^\states \times \act^\states$ is a \emph{valid certificate for lower bounds on $\max$-reachability with witness strategy} if
    \begin{align*}
        1)~\distop{\strat}{}(\rank) \leq \rank,
        \quad
        2)~\vals \leq \bellman{\strat}(\vals),
        \quad
        3)~\forall s \in \states\setminus\target \colon \vals(s) > 0 \implies \rank(s) < \infty.
    \end{align*}
    If $(\vals, \rank, \strat)$ is valid, then $\forall s \in \states \colon \prmax_s(\reach\target) \geq \prind{\strat}_s(\reach\target)  \geq \vals(s)$.
\end{restatable}
\tw{proof in \Cref{app:certlowerstratwitness}}

\subsubsection{Certificates without a Witness Strategy.}

Increasing the size of the certificate by including the strategy can be avoided, as it can be \enquote{read off} from the certifying probability vector $\vals \in [0,1]^\states$.
To this end, we define the \emph{$\vals$-increasing} actions of state $s \in \states$:
$\indact{\vals}(s)=\{a\in\act(s) \mid \vals(s) {\leq} \sum_{s' \in \post{(s,a)}} \prmdp(s, a, s') {\cdot} \vals(s') \}$.
If $\vals \leq \bellmanmax(\vals)$, then $\indact{\vals}(s)$ contains at least one action.
Next, we define a variant of the distance operator which only considers $\vals$-increasing actions:
\begin{align*}
    \distopvals{\min}{} \colon
    \exnats^\states \to \exnats^\states,
    ~
    \distopvals{\min}{}(\rank)(s)
    =
    \begin{cases}
        0 & \text{ if } s \in \target \\
        1 ~+~ \min\limits_{a \in \indact{\vals}(s)} ~ \min\limits_{s' \in \post(s,a)}  \rank(s') & \text{ if } s \in \states \setminus \target
    \end{cases}
\end{align*}

\begin{restatable}[Certificates for Lower Bounds on $\prmax(\reach\target)$]{proposition}{certlowernostrat}%
    \label{thm:certlower-nostrat}%
    A tuple $(\vals,\rank) \in [0,1]^\states \times \exnats^\states$ is a \emph{valid certificate for lower bounds on $\max$-reachability}~if
    \begin{align*}
        1)~\distopvals{\min}{}(\rank) \leq \rank,
        \quad
        2)~\vals \leq \bellmanmax(\vals),
        \quad
        3)~\forall s \in \states \setminus \target \colon \vals(s) > 0 \implies \rank(s) < \infty.
    \end{align*}
    If $(\vals, \rank)$ is valid, then $\forall s \in \states \colon \prmax_s(\reach\target) \geq \vals(s)$.
\end{restatable}

\section{Certificates for Expected Rewards}
\label{sec:ExpRew}

We present certificates for bounds on expected rewards (\Cref{fig:certoverviewrewards}) in the  \enquote{$\starinf$} semantics that assigns infinite reward to paths not reaching $\target$
, with the other case in \ifarxivelse{\Cref{app:sec:exprewStarRho}}{\cite[App.~F]{arxivversion}}.
We employ the reward variant of the Bellman operator:

\begin{definition}[Bellman Operator for Expected Rewards]
    \label{defn:bellmanopRewards}%
    We define the Bellman operator for expected rewards
    $\bellmanropt \colon \RgeqZeroInfStates \to \RgeqZeroInfStates$ as follows:
    \begin{align*}
        \bellmanropt(\vals)(s)
        ~=~
        \begin{cases}
        0 & \text{ if } s \in \target \\
        \rew(s) + \underset{a \in \act(s)}{\opt} \sum\limits_{s' \in \post(s,a)} \prmdp(s, a, s') \cdot \vals(s') & \text{ if } s \in \states \setminus \target 
        \end{cases}
    \end{align*}
\end{definition}
The above definition assumes that multiplication by $\infty$ \emph{absorbs} positive numbers, i.e., $p \cdot \infty = \infty$ for all $p > 0$, and $a + \infty = \infty + a = \infty$ for all $a \in \RgeqZeroInf$.

Again, $\bellmanropt$ is a monotone function on the complete lattice $(\RgeqZeroInfStates, \leq)$
and thus has a least and a greatest fixed point by \Cref{thm:knastertarski}.
Unfortunately, as it turns out, the sought-after expected rewards $\Eopt_s(\reach\target)$, $s \in \states$, are \emph{neither of these two fixed points}.
Indeed, $\lfp{\bellmanropt}$ corresponds to the expected rewards in the semantics considered in \ifarxivelse{\Cref{app:sec:exprewStarRho}}{\cite[App.~F]{arxivversion}}, and $\gfp{\bellmanropt}$ is a trivial upper bound assigning $\infty$ to all states, see the example in \Cref{sec:exprew:lowerbounds}.

\begin{remark}[Asymmetry and Duality]
	In \Cref{sec:certificates}, an asymmetry between upper and lower bounds arose as the reachability probabilities are a \emph{least} fixed point. 
	Further, for the case of maximizing reachability, spurious fixed points occurred and we required a witness strategy to \enquote{make progress} towards the targets (the fact that this case requires special treatment of end components is well established in literature, e.g.,~\cite{DBLP:conf/cav/HartmannsK20}).
	For \emph{safety} objectives, where the goal is to avoid a set of bad states, the situation is dual:
    The safety probabilities are a \emph{greatest} fixed point, so the lower bound case is simple, and when minimizing the upper bound, we require a witness strategy.
	The $\starinf$ semantics for expected rewards share some similarities with a safety objective, since the value is maximized (i.e., is infinite) when the target set is avoided.
	This section thus differs from \Cref{sec:certificates} in two ways:
    (i) Everything is dual, as $\starinf$ is \enquote{safety-like}, and
    (ii) additional complications arise from the trivial greatest fixed point $\gfp{\bellmanropt} = \vec\infty$, see below.
\end{remark}

\subsection{Lower Bounds on Optimal Expected Rewards}\label{sec:exprew:lowerbounds}
\begin{wrapfigure}[6]{r}{0.215\textwidth}%
	\vspace{-0.9cm}
    \centering
    \begin{tikzpicture}[every state/.style={minimum size=0},thick,>=stealth]
        \node[state,label=-90:$\infty$] (s) {$s$};
        \node[state,right=of s,accepting,label=-90:$0$] (t) {$t$};
        \draw[->] (s) edge[loop above] node[left] {$\tfrac 1 2$} (s);
        \draw[->] (s) edge node[auto] {$\tfrac 1 2$}(t);
        \draw[->] (t) edge[loop above] node[left] {$1$}(t);
    \end{tikzpicture}
    \caption{A DTMC.}
    \label{fig:smallCounterExample}
\end{wrapfigure}
%
Due to the absorptive property of multiplication by $\infty$, $\gfp{\bellmanropt}$ may assign $\infty$ to states that actually have finite value:
For instance, in the DTMC in \Cref{fig:smallCounterExample}, the gfp assigns $\infty$ to $s$ because $\infty = \rew(s) + \frac 1 2 \cdot \infty + \frac 1 2 \cdot 0$, while in fact $\E_s(\reach\target) = 2 \cdot \rew(s) < \infty$.
To address this, we force the values of states that a.s. reach the target to be finite as follows:

\begin{restatable}{lemma}{infLowerGfp}%
    \label{lem:infLowerGfp}%
    Let $\vals \in \RgeqZeroInfStates$ be such that 1) $\vals \leq \bellmanropt(\vals)$ and 
    2) for all $s \in\states$: $\pr^{\nopt}_s(\lozenge\target) = 1 \implies \vals(s) < \infty$ .
    Then it holds for all $s \in \states$ that $\vals(s) \leq \Eopt_s(\lozenge \target)$.
    \tw{Proof in \Cref{app:infLowerGfp}}
\end{restatable}

Intuitively, \Cref{lem:infLowerGfp} requires that a lower bound on $\Emin_s(\reach\target)$ can only be infinite if $\target$ cannot be reached a.s., i.e.\ $\prmax_s(\reach\target) < 1$ (dually for $\Emax$).
Combining \Cref{lem:infLowerGfp} and a certificate for \emph{non-a.s. reachability} (\Cref{sec:qualreachandsafe}) yields:

\begin{restatable}[Certificates for Lower Bounds on $\Eopt(\reach\target)$]{proposition}{certslowerexprewinf}%
    \label{prop:certsInfLower}%
    A tuple $(\vals, \rank) \in \RgeqZeroInfStates \times \exnats^\states$ is a \emph{valid certificate for lower bounds on $\opt$-exp.\ rewards} if
    \begin{align*}
        1)~\distopmod{\nopt}(\rank) \leq \rank,
        \quad
        2)~\vals \leq \bellmanropt(\vals),
        \quad
        3)~\forall s \in \states \colon \vals(s) = \infty \implies \rank(s) < \infty.
    \end{align*}
    If $(\vals,\rank)$ is valid, then $\forall s\in\states \colon \Eopt_s(\reach\target) \geq \vals(s)$.
    \tw{Proof in \Cref{app:certsInfLower}}
\end{restatable}

\subsection{Upper Bounds on Maximal Expected Rewards}
\label{sec:certrew:uppermax}
%
%
Next we focus on upper bounds on maximal expected rewards.
Using \Cref{leitmotif:uniqueFixedPoints} as for \emph{lower} bounds on \emph{minimal} reachability probabilities (\Cref{sec:certreach:lowermin}), we obtain such certificates via a modified Bellman operator:
\begin{align*}
    \modbellmanrmax \colon
    \RgeqZeroInfStates \to \RgeqZeroInfStates,
    \quad
    \modbellmanrmax(\vals)(s)
    ~=~
    \begin{cases}
        \bellmanrmax(\vals)(s) & \text{ if } \prmin_s(\lozenge\target) > 0 \\
        \infty & \text{ if } \prmin_s(\lozenge\target) = 0 
    \end{cases}
\end{align*}
\begin{restatable}{lemma}{infMaxUpperLfp}
    \label{lem:infMaxUpperLfp}
    For all $s \in \states$, $(\lfp{\modbellmanrmax})(s) = \mathbb{E}^{\max}_s(\lozenge \target)$.
    \tw{Proof in \Cref{app:infMaxUpperLfp}}
\end{restatable}
We stress that unlike $\modbellmanmin$ from \Cref{sec:certreach:lowermin}, $\modbellmanrmax$ does not have a \emph{unique} fixed point, see \Cref{fig:smallCounterExample}.
Nonetheless, with \Cref{lem:infMaxUpperLfp}, \Cref{leitmotif:knasterTarski}, and the certificates for positive reachability from \Cref{prop:certQualReach}, we obtain:

\begin{restatable}[Certificates for Upper Bounds on $\Emax(\reach\target)$]{proposition}{certsInfMaxUpper}%
    \label{prop:certsInfMaxUpper}%
    A tuple $(\vals, \rank) \in \RgeqZeroInfStates \times \exnats^\states$ is a \emph{valid certificate for upper bounds on  $\max$-exp.\ rewards} if 
    \begin{align*}
        1)~\distop{\max}{}(\rank) \leq \rank,
        \quad
        2)~\bellmanrmax(\vals) \leq \vals,
        \quad
        3)~\forall s \in \states \colon \vals(s) < \infty \implies \rank(s) < \infty.
    \end{align*}
    If $(\vals,\rank)$ is valid, then $\forall s\in\states \colon \Emax_s(\reach\target) \leq \vals(s)$.
\end{restatable}
\tw{Proof in \Cref{app:certsInfMaxUpper}}

\subsection{Upper Bounds on Minimal Expected Rewards}
\label{sec:InfUpperMin}
Our approach for this case parallels the one for \emph{lower} bounds on \emph{maximal} reachability probabilities from \Cref{sec:certreach:lowermax}.
The modified Bellman operator $\modbellmanrmin$ (defined analogous to $\modbellmanrmax$ from above) does \emph{not} characterize the minimal expected rewards as its least fixed point.
The problem are, again, end components, see \ifarxivelse{\Cref{app:remarkbellmanminrewards}}{\cite[App.~E.5]{arxivversion}} for a counter-example.
Following \Cref{leitmotif:witnessStrat} and \Cref{sec:certreach:lowermax}, we can, however, certify upper bounds on $\Emin(\reach\target)$ by including a witness strategy (see \ifarxivelse{\Cref{app:EminUpperWithWitness}}{\cite[App.~E.6]{arxivversion}}).

As with lower bounds on $\max$-reachability, it is also possible to avoid this explicit witness strategy:
%
%
%
%
%
We define the \emph{$\vals$-decreasing actions} of $s$ as $\coindactrew{\vals}(s) = \{a\in\act(s) \mid \vals(s) \geq \rew(s) + \sum_{s' \in \post(s,a)} \prmdp(s, a, s') \cdot \vals(s') \}$.
%
%
If $\bellmanrmin(\vals) \leq \vals$, then $\coindactrew{\vals}(s) \neq \emptyset$.
We define a distance operator with $\distoprvals{\min}{}$
that only considers $\vals$-decreasing actions completely analogous to $\distopvals{\min}{}$ from \Cref{sec:certreach:lowermax}.
%
\begin{restatable}[Certificates for Upper Bounds on $\Emin(\lozenge\target)$]{proposition}{certsInfMinUpperNoWitness}%
    \label{prop:certsInfMinUpperNoWitness}%
    A tuple $(\vals, \rank) \in \RgeqZeroInfStates \times \exnats^\states$ is a \emph{valid certificate for upper bounds on $\min$-exp.\ rewards} if 
    \begin{align*}
        1)~\distoprvals{\min}{}(\rank) \leq \rank,
        \quad
        2)~\bellmanrmin(\vals) \leq \vals,
        \quad
        3)~\forall s \in \states \colon \vals(s) < \infty \implies \rank(s) < \infty.
    \end{align*}
    If $(\vals, \rank)$ is valid, then $\forall s \in \states \colon \Emin_s(\reach\target) \leq \vals(s)$.
\end{restatable}
\tw{Proof in \Cref{app:certsInfMinUpperNoWitness}}


\section{Computing Certificates}
\label{sec:computing}

In \Cref{sec:qualreachandsafe,sec:certificates,sec:ExpRew} we described \emph{what} certificates are and discussed their verification conditions.
We now elaborate on \emph{how to compute} certificates.
To this end, we first discuss computation of (co-)inductive value vectors $\vals$ and then focus on the ranking functions $\rank$ required by some certificates (see \Cref{fig:certoverview}).
We stress that a sound certificate checker detects any wrong results produced by buggy implementations of the methods discussed in this section.
Indeed, during implementation of the certificate computation algorithms in \storm, checking the certificates helped finding and resolving implementation bugs.

As we enter the realm of numeric computation, some remarks are in order.
For computational purposes we assume that the transitions probabilities are \emph{rational numbers}, i.e., fractions of integers.
Moreover:
\begin{tightcenter}
	\emph{Our goal is to compute a certificate with a \underline{rational} value vector $\vals$ and to check it with \underline{exact}, \underline{arbitrar}y\underline{ }p\underline{recision rational number arithmetic}.}
\end{tightcenter}

\paragraph{Certificates via Exact Algorithms.}
The conceptually easiest certifying MDP model checking algorithm is to compute the rational reachability probabilities or expected rewards \emph{exactly}.
The resulting value vector is both inductive and co-inductive.
Thus, exact algorithms yield a certificate essentially as a by-product.
We refer to~\cite{TACAS23,practitionersJournalPreprint} for an in-depth comparison of exact algorithms based on \emph{Policy Iteration} (PI), \emph{Rational Search} (RS), and \emph{Linear Programming} (LP).
The practically most efficient algorithm is PI with exact LU decomposition as linear equation solver; see~\cite[Secs.~2.2~and~4.2]{practitionersJournalPreprint} for a description of the algorithm.

\paragraph{Certificates via Approximate Algorithms.}
In practice, most probabilistic model checkers use algorithms that are not exact but approximate: 
They employ \emph{approximate, fixed-precision floating point arithmetic} and use a variant of VI that only returns an approximate result, namely for each state an interval $[\ell,u]$ containing the exact value, such that $|\ell - u| \leq \varepsilon$ for a given $\varepsilon$ (typically $10^{-6}$).
They do this because (i)~when using exact arithmetic, fractions can grow very large, hindering scalability, (ii)~VI-based algorithms often outperform PI, albeit not as dramatically as folklore claimed~\cite{TACAS23,practitionersJournalPreprint}, and (iii)~approximate results usually suffice.
We now exemplify with the VI-variant \emph{Interval Iteration} (II)~\cite{DBLP:conf/cav/Baier0L0W17,HM18} how to make an approximate, floating point-based algorithm certifying, leaving other variants such as \emph{optimistic VI}~\cite{DBLP:conf/cav/HartmannsK20} and \emph{Sound VI}~\cite{DBLP:conf/cav/QuatmannK18} for future work.

II for reachability\footnote{For expected rewards, II additionally requires computing an upper bound, see~\cite{DBLP:conf/cav/Baier0L0W17}.} works by first \emph{collapsing end components}~\cite{atva14,HM18} of the MDP to ensure that $\bellmanopt$ has a \emph{unique} fixed point.
II then runs two instances of VI in parallel, starting from $\vals^{(0)} = \vec 0$ and $\valsb^{(0)} = \vec 1$:
\begin{align*}
	\vec 0 = \vals^{(0)} \leq \bellmanopt(\vals^{(0)}) = \vals^{(1)} \leq 
	~\ldots~
	\fp{\bellmanopt}
	~\ldots~
	\leq \valsb^{(1)} = \bellmanopt(\valsb^{(0)}) \leq \valsb^{(0)} = \vec 1
\end{align*}
Both sequences contain (co-)inductive vectors only and converge to the fixed point.
The iteration can be stopped when the difference is as small as desired.

However, as we demonstrate experimentally (\Cref{sec:expeval}), \emph{inexact floating point arithmetic usually breaks \mbox{(co-)}inductivity of the elements in the II sequences}, as was already reported in~\cite{DBLP:conf/tacas/WinklerK23} in a similar setting.
More precisely, let $\bellmanopt_\floats$ be a \enquote{floating point variant} of $\bellmanopt$, i.e., 
	the (exact) result of each operation is rounded \emph{to a nearest float}.
    This the default rounding mode in IEEE 754.
Let $\vals_{\floats}^{(i)}$ be the $i$-th element, $i > 0$, in the lower VI sequence of $\bellmanopt_\floats$ starting from $\vec 0$.
Then, due to rounding errors, $\vals_{\floats}^{(i)} \leq \bellmanopt(\vals_{\floats}^{(i)})$ does \emph{not} hold in general, i.e., $\vals_{\floats}^{(i)}$ might not be co-inductive.
%
We propose two ways to mitigate this problem: \emph{Safe rounding}~\cite{arndFP} and \emph{Smooth~II}.

First, safe rounding amounts to configuring the IEEE754 rounding mode so that results of floating point computations are always rounded towards $0$ when iterating from below, and towards $\infty$ when iterating from above.
While safe rounding provably yields sound bounds~\cite{arndFP}, it may slow down or even prevent convergence of II.
Nonetheless, in practice,
II with safe rounding finds significantly more certificates than II with default rounding (\Cref{sec:expeval}).

Second, for Smooth II we define the \emph{$\gamma$-smooth Bellman~operator} ($\gamma \in [0,1)$)
\[
\bellmanopt_\gamma(\vals) ~=~ \gamma \cdot \vals + (1-\gamma) \cdot \bellmanopt(\vals)
~,
\]
where scalar multiplication and addition are component-wise.
$\bellmanopt_\gamma$ and $\bellmanopt$ have the same fixed points, and every \mbox{(co-)}inductive value vector w.r.t.\ $\bellmanopt_\gamma$ is also \mbox{(co-)}inductive w.r.t.\ $\bellmanopt$\ifarxivelse{ (see \Cref{app:smoothII})}{~\cite[App.~G.1]{arxivversion}}.
The key property of $\bellmanopt_\gamma$ compared to $\bellmanopt$ is that the former enforces \emph{ultimately strictly monotonic VI sequences}.
This mitigates the floating point rounding issues.
Notice, however, that smoothing slows down convergence.
Smoothing and safe rounding may be combined.

\paragraph{Computing Ranking Functions.}

We briefly outline how to obtain the unique and least fixed points of $\distop{\opt}{}$ and $\distopmod{\opt}$, respectively (see \Cref{def:distop,def:distopmod}).\ifarxivelse{\footnote{The restricted variants $\distopvals{\min}{}$ and $\distoprvals{\min}{}$ from \Cref{sec:certreach:lowermax,sec:InfUpperMin} are the same as $\distop{\min}{}$ in a sub-MDP restricted to increasing/decreasing actions.}}{}

First, $\fp{\distop{\opt}{}}$ can be computed via VI from $\rank^{(0)} = \vec\infty$.
This iteration converges in finitely many steps.
Second, to compute $\lfp{\distopmod{\opt}{}}$ we propose to perform VI from $\rank^{(0)}$ with $\rank^{(0)}(s) = [\pr^{\nopt}_s(\reach\target) = 1] \cdot \infty$ for all $s \in \states$.
The condition in the Iverson bracket can be evaluated using standard graph analysis~\cite[Section~10.6.1]{principles}.
This iteration converges in finitely many steps as well, see \ifarxivelse{\Cref{app:rankingFunctionVIDetails,app:algsForComputing}}{\cite[Apps.~G.2~and~G.3]{arxivversion}} for details and a practically more efficient algorithm.

\section{Experimental Evaluation}
\label{sec:expeval}

\newcommand{\pix}{PI$^X$\xspace}
\newcommand{\iiraw}{II}
\newcommand{\ii}{\iiraw\xspace}
\newcommand{\sii}[1]{\iiraw$^{\circlearrowright #1}$\xspace}
\newcommand{\rii}[1][]{\iiraw$^{\ifstrempty{#1}{}{\circlearrowright #1}}_{\mathrm{rnd}}$\xspace}
\newcommand{\datatriple}[3]{~\textcolor{gray}{[\textcolor{plotdarkgreen}{#1}|\textcolor{plotdarkred}{#2}|\textcolor{black}{#3}]}}
\newcommand{\numbenchmarks}{447\xspace}
\newcommand{\numSuccessIIAndPITogether}{396\xspace}

\newcommand{\pixdata}{\pix\datatriple{351}{0}{96}} 
\newcommand{\iidata}{\ii\datatriple{171}{225}{51}} 
\newcommand{\siionedata}{\sii{0.1}\datatriple{243}{138}{66}} 
\newcommand{\siifivedata}{\sii{0.5}\datatriple{235}{146}{66}} 
\newcommand{\siieightdata}{\sii{0.8}\datatriple{303}{77}{67}} 
\newcommand{\siininedata}{\sii{0.9}\datatriple{340}{35}{72}} 
\newcommand{\siininefivedata}{\sii{0.95}\datatriple{335}{33}{79}} 
\newcommand{\siinineninedata}{\sii{0.99}\datatriple{323}{30}{94}} 
\newcommand{\riidata}{\rii\datatriple{249}{147}{51}} 
\newcommand{\riizerofivedata}{\rii[0.05]\datatriple{338}{44}{65}} 
\newcommand{\riionedata}{\rii[0.1]\datatriple{337}{45}{65}} 
\newcommand{\riitwodata}{\rii[0.2]\datatriple{333}{49}{65}} 
\newcommand{\riifivedata}{\rii[0.5]\datatriple{334}{47}{66}} 
\newcommand{\riieightdata}{\rii[0.8]\datatriple{342}{38}{67}} 

\paragraph{Implementation.}
We implemented certificate computation as discussed in \Cref{sec:computing} in \storm~\cite{storm}.
Given a higher-level model description (\prism~\cite{prism} or \tool{Jani}~\cite{DBLP:conf/tacas/BuddeDHHJT17}), and a reachability probability or expected reward query, our implementation proceeds in three steps:
First, \storm builds an explicit MDP from the description. 
Second, it computes a certificate for both lower \emph{and} upper bounds, such that the relative difference between the two values is at most $\varepsilon = 10^{-6}$ for each MDP state.
Finally, \storm checks the validity of the certificate.

Following the discussion in \Cref{sec:computing}, we consider the following algorithms:
Regarding exact computation, we use PI with \emph{exact} LU decomposition, called \pix.
For approximate computation with floating point arithmetic, we employ \ii.
Further, to investigate the impact of the rounding error mitigation techniques from \Cref{sec:computing}, we complement \ii with either safe rounding (denoted \rii), smoothing with parameter $\gamma$ (denoted \sii{\gamma}; we consider $\gamma \in \{0.05,0.8,0.9,0.95\}$), or a combination of both (denoted \rii[\gamma]).
Overall, we compare \pix and seven variants of \ii.
We employ additional standard modifications of the algorithms, namely:
We eliminate end components whenever possible, apply topological optimizations for \pix and \ii, and apply Gauß-Seidl Bellman updates for \ii~\cite{DBLP:conf/cav/Baier0L0W17,TACAS23}.

In all three steps of the implementation, we represent numbers as arbitrary precision rationals implemented in GMP~\cite{gmp}---except when running \ii in the second step (in which case we convert rationals to their nearest floats, potentially yielding invalid certificates).
We thus certify reachability probabilities and expected rewards with respect to the \emph{exact} MDP without rounding errors.

The MDP and the certificate computed with \storm can be exported and checked by an independent \emph{formally verified certificate checker}.
To construct the latter, we verified the correctness of the certificate checking algorithms in the interactive proof assistant \isabelle~\cite{DBLP:books/sp/NipkowPW02}, extending previous work on MDPs~\cite{DBLP:journals/jar/Holzl17,DBLP:conf/aaai/SchaffelerA23} by total rewards and qualitative reachability properties.
Based on this library, we proved correct the soundness of the certificates described in \Cref{sec:qualreachandsafe,sec:certificates,sec:ExpRew}.
We used \isabelle{}'s code export mechanism~\cite{DBLP:conf/itp/HaftmannKKN13} to obtain a verified, executable Standard ML implementation that employs exact rational arithmetic.
The construction of the MDP from a \prism or \tool{Jani} model as well as export and parsing of MDPs and certificates are currently not verified.

\paragraph{Benchmarks and Setup.}
We use all 366 benchmark instances from the quantitative verification benchmark set (QVBS)~\cite{QVBS} that (i) consider an MDP with a reachability or reward objective and (ii) for which \storm can build an explicit representation within 5 minutes.
Additionally, since the QVBS contains no models exhibiting non-trivial ECs, we include 71 structurally diverse models from various sources detailed in~\cite[Sec.~5.3]{practitionersJournalPreprint}.
Overall, we consider the complete \emph{alljani} set from \cite{practitionersJournalPreprint}.
%
We invoke \storm for each combination of benchmark instance and certificate algorithm and report the overall runtime (walltime).
All experiments ran on Intel Xeon 8468 Sapphire 2.1 Ghz systems.
We used \tool{Slurm}
to limit the individual executions to 4 CPU cores and \SI{16}{\giga\byte} of RAM, with a time limit of \SI{900}{\second}.
Next, we discuss our findings by answering three research questions (RQs).


\begin{figure}[t]
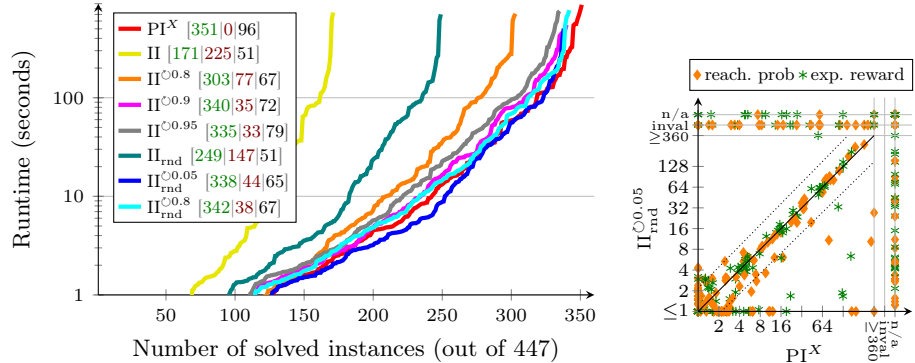

	\setlength{\quantileplotwidth}{0.67\linewidth}
	\centering
	\quantileplot{plotdata/quantile.csv}{
		logs.Storm.topoexpi/plotred,
		logs.Storm.topofpii/plotyellow,
		logs.Storm.topofpsmoothii80/plotorange,
		logs.Storm.topofpsmoothii90/plotpink,
		logs.Storm.topofpsmoothii95/plotdarkgray,
		logs.Storm.topofproundii00/plotteal,
		logs.Storm.topofproundii05/plotblue,
		logs.Storm.topofproundii80/plotcyan
	}{
		\pixdata,
		\iidata,
		\siieightdata,
		\siininedata,
		\siininefivedata,
		\riidata,
		\riizerofivedata,
		\riieightdata
	}{1}{360}{1}{900}{north west}
    \hfill
	\scatterplotAlg{plotdata/scatter.csv}{\thisrow{logs.Storm.topoexpi} < 1 ? 1 : \thisrow{logs.Storm.topoexpi}}{\pix}{\thisrow{logs.Storm.topofproundii05} < 1 ? 1 : \thisrow{logs.Storm.topofproundii05}}{\rii[0.05]}{true}
	\caption{RQ1: Runtime for computing certificates of \pix and several combinations of mitigation techniques with \ii (left); and detailed comparison of \pix and \rii[0.05] (right).}
	\label{fig:RQ1}
\end{figure}

\paragraph{RQ1: Best algorithm for certificate generation?}
\Cref{fig:RQ1} (left) compares the runtimes of \pix and our seven \ii variants.
A point $(x,y)$ for algorithm $A$ indicates that there are $x$ instances
for which $A$ computes a \emph{valid} certificate within $y$ seconds (including time for model construction but excluding time for exporting the certificate files).
The triples \!\!\datatriple{$v$}{$w$}{$u$} in the legend indicate that the algorithm produced a total of $v$ valid and $w$ invalid certificates (with invalidity likely due to floating point issues), while for the remaining $u$ instances no result was found within the resource limits.
As expected, all certificates produced by the exact \pix{} are valid, while standard \ii{} produces many invalid certificates.
Safe rounding and smoothing improve the number of valid certificates.
Notably, \sii{\gamma} (only smoothing) performs best for $\gamma$ values close to 1, while the performance of \rii[\gamma] (smoothing \emph{and} safe rounding) is less sensitive towards $\gamma$; see \ifarxivelse{\Cref{app:experiments}}{\cite[App.~H]{arxivversion}} for more details.
Among all II variants, \rii[0.05] shows the best overall performance.

The scatter plot in \Cref{fig:RQ1} (right) further compares \pix and \rii[0.05].
A data point $(x,y)$ corresponds to one benchmark instance, where $x$ and $y$ are runtimes of \pix and \rii[0.05]. 
A point at $\ge 300$ indicates a runtime between 300 and 900 seconds, \emph{inval} means an invalid certificate, and \emph{n/a} denotes an aborted computation due to time/memory limits.
Many instances that \pix cannot solve are solved by \rii[0.05] and vice versa.
This is already the case without computing certificates, as the structure of a benchmark affects the performance of the algorithms differently~\cite{TACAS23,practitionersJournalPreprint}.
Thus, as in the case without computing certificates, there is no \enquote{best algorithm}, and both \pix and variants of \ii can be considered. 
Overall, \numSuccessIIAndPITogether out of \numbenchmarks instances are correctly solved and certified by \pix or \rii[0.05] (or both).
We highlight that \pix is not only a complete certifying algorithm, but also practically efficient, even though it uses exact arithmetic.

\begin{figure}[t]
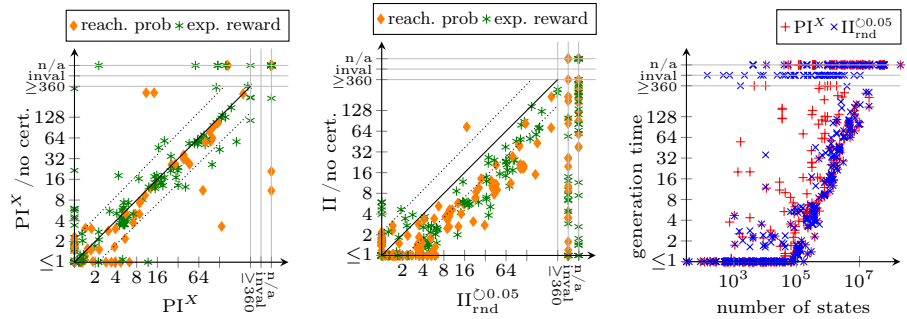

	\scatterplotAlg{plotdata/scatter.csv}{%
		\thisrow{logs.Storm.exprt-topoexpi} < 1 ? 1 : \thisrow{logs.Storm.exprt-topoexpi}
	}{\pix}{\thisrow{logs.Storm.expinocert} < 1 ? 1 :  \thisrow{logs.Storm.expinocert}}{\pix\!/no cert.}{true}
    \hfill
	\scatterplotAlg{plotdata/scatter.csv}{%
		\thisrow{logs.Storm.exprt-topofproundii05} < 1 ? 1 : \thisrow{logs.Storm.exprt-topofproundii05}
	}{\rii[0.05]}{\thisrow{logs.Storm.fpiinocert} < 1 ? 1 :  \thisrow{logs.Storm.fpiinocert}}{\ii\!/no cert.}{true}
    \hfill
	\scatterplotStatesStorm{exprt-topoexpi}{\pix}{exprt-topofproundii05}{\rii[0.05]}
	\caption{
		RQ2: Runtime overhead of certified MDP model checking for \pix (left) and \ii (middle), and scalability of both with respect to the number of states (right).
	}
	\label{fig:RQ2}
\end{figure}

\begin{figure}[t]
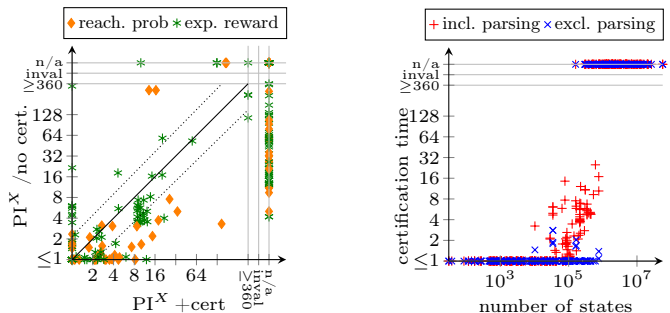

	\centering
	\scatterplotAlg{plotdata/scatter.csv}{%
		\thisrow{logs.Storm.exprt-topoexpi} < 360 ? (\thisrow{logs.cert.topoexpi}< 360 ? (\thisrow{logs.Storm.exprt-topoexpi} + \thisrow{logs.cert.topoexpi} < 360 ? (\thisrow{logs.Storm.exprt-topoexpi} + \thisrow{logs.cert.topoexpi} < 1 ? 1 : \thisrow{logs.Storm.exprt-topoexpi} + \thisrow{logs.cert.topoexpi}) : 360) :\thisrow{logs.cert.topoexpi}) : \thisrow{logs.Storm.exprt-topoexpi}
	}{\pix\!+cert}{\thisrow{logs.Storm.expinocert} < 1 ? 1 :  \thisrow{logs.Storm.expinocert}}{\pix\!/no cert.}{true}
    \hspace{1cm}
	\scatterplotStatesCert{topoexpi}
	\caption{
		RQ3: Runtime overhead of the certified pipeline/Runtime of certificate checking
	}
	\label{fig:RQ3}
\end{figure}

\paragraph{RQ2: Runtime overhead of certificate generation?}
\Cref{fig:RQ2} (left/middle) reports the runtime overhead of generating a certificate for \pix and \rii[0.05].
For \pix, the overhead is typically within a factor of 2, often significantly less. It is sometimes faster due to implementation differences in the certifying variant of \pix.
For \rii[0.05], the overhead is slightly larger, typically around 1.5 to 4.
This is partly due to the slower convergence caused by smoothing.
\Cref{fig:RQ2} (right) investigates the scalability of certificate generation with respect to the number of states.
For MDPs with up to $10^5$ states, certificate generation usually completes within a minute (often much less); for more than $10^7$ states, it usually times out.

\paragraph{RQ3: Scalability of the formally verified certificate checker?}
\Cref{fig:RQ3} (left) compares the runtime of the full pipeline including certificate generation and verification using our formally verified checker (\pix\!+cert) with plain, uncertified MDP model checking based on \pix.
Compared to \Cref{fig:RQ2} (left), the added verification of the certificates causes additional time/memory outs, and roughly doubles the runtime of the other instances.
\Cref{fig:RQ3} (right) reveals that parsing is currently a major bottleneck in the verified checker.
Nonetheless, the checker completes within a few seconds on MDPs with up to $\approx 10^5$ states, and usually within \SI{30}{\second} for instances with up to $\approx 10^6$ states.

\section{Conclusion and Future Work}

We proposed \emph{fixed point certificates} as a new standard for certified model checking of reachability and expected reward properties in MDPs.
The soundness of these certificates was formalized in \isabelle, increasing their trustworthiness and enabling us to generate a formally verified certificate checker, applicable to non-trivial practically relevant instances.
Our certificates can be generated with moderate overhead via minor, yet careful, modifications of established algorithms like II or PI.
This allows tool developers and competitions~\cite{qcomp20}---for which our certificates provide formally verified reference results---to adopt our proposal with relatively low effort.
%
Future work is to develop a more efficient certificate format.
Further, we plan to extend our theory to other quantitative verification settings~\cite{qcomp23}, e.g., stochastic games and $\omega$-regular properties, and make it amenable to techniques such as \emph{symbolic model checking} and \emph{partial exploration}.


\paragraph{Data availability statement.}
The models, tools, and scripts to reproduce our experimental evaluation are available at DOI \href{https://doi.org/10.5281/zenodo.14626585}{10.5281/zenodo.14626585}~\cite{artifact}.

\bibliographystyle{splncs04}
\bibliography{references}

\ifarxivelse{%
\newpage
\appendix
\section*{\centering Appendix}

\crefalias{section}{appendix} 
\crefalias{subsection}{appendix} 

\section{Certificates for Expected Rewards: Overview}
\begin{table}[h]
    \caption{
        Our certificates for \emph{expected rewards} with $\starinf$ semantics.
    }
    \label{fig:certoverviewrewards}
    \begin{adjustbox}{max width=\linewidth}
        \setlength{\tabcolsep}{4pt}%
        \renewcommand{\arraystretch}{1.05}%
        \begin{tabular}{l l l}        
            \toprule 
            \rowcolor{black!4}\quad\textbf{Certificate} & \textbf{Conditions} & \textit{Explanation} \\ \midrule
            \rowcolor{\colMinUpper!4}\multicolumn{3}{l}{\textcolor{black}{\textbf{Upper bounds} on \textbf{minimal} expected rewards:} ~ $\forall s \in \states \colon \Emin_{s}(\reach\target) \leq \vals(s)$ \hspace*{1cm} \textcolor{gray}{[\Cref{prop:certsInfMinUpperNoWitness}]}} \\
            \rowcolor{\colMinUpper!4}\quad$\vals \in \RgeqZeroInfStates$ & $\bellmanrmin(\vals) \leq \vals$ & \textit{\textbf{min}-Bellman operator \textbf{decreases} value of {all} states}\\
            \rowcolor{\colMinUpper!4}\quad$\rank \in \exnats^\states$ & $\distoprvals{\min}{}(\rank) \leq \rank$ & \textit{$\rank$ upper bounds \textbf{minimal} distances to $\target$ via \textbf{$x$-decr.\ actions}}\\ 
            \rowcolor{\colMinUpper!4}& $\vals(s) {<} \infty \implies \rank(s) {<} \infty$ & \textit{\textbf{finite} expected reward necessitates finite distance} ($\forall s \in \states$)\\[3pt] 
            \rowcolor{\colMaxUpper!4}\multicolumn{3}{l}{\textcolor{black}{\textbf{Upper bounds} on \textbf{maximal} expected rewards:} ~  $\forall s \in \states \colon \Emax_{s}(\reach\target) \leq \vals(s)$ \hspace*{1cm} \textcolor{gray}{[\Cref{prop:certsInfMaxUpper}]}} \\
            \rowcolor{\colMaxUpper!4}\quad$\vals \in \RgeqZeroInfStates$ & $\bellmanrmax(\vals) \leq \vals$ & \textit{\textbf{max}-Bellman operator \textbf{decreases} value of {all} states}\\
            \rowcolor{\colMaxUpper!4}\quad$\rank \in \exnats^\states$ & $\distop{\max}{}(\rank) \leq \rank$ & \textit{$\rank$ upper bounds \textbf{maximal} distances to $\target$} \\ 
            \rowcolor{\colMaxUpper!4}& $\vals(s) {<} \infty \implies \rank(s) {<} \infty$ & \textit{\textbf{finite} expected reward necessitates finite distance} ($\forall s \in \states$)\\[3pt]
            \rowcolor{\colMinLower!4}\multicolumn{3}{l}{\textcolor{black}{\textbf{Lower bounds} on \textbf{minimal} expected rewards:} ~ $\forall s \in \states \colon \Emin_{s}(\reach\target) \geq \vals(s)$ \hspace*{1cm} \textcolor{gray}{[\Cref{prop:certsInfLower}]} } \\
            \rowcolor{\colMinLower!4}\quad$\vals \in \RgeqZeroInfStates$ & $\bellmanrmin(\vals) \geq \vals$ & \textit{\textbf{min}-Bellman operator \textbf{increases} value of {all} states}\\
            \rowcolor{\colMinLower!4}\quad$\rank \in \exnats^\states$ & $\distopmod{\max}(\rank) \leq \rank$ & \textit{$\rank$ upper bounds \textbf{maximal} distances to \textbf{end components}}\\
            \rowcolor{\colMinLower!4}\quad & $\vals(s) {=} \infty \implies \rank(s) {<} \infty$ & \textit{\textbf{infinite} expected reward necessitates finite distance} ($\forall s \in \states$)\\[3pt] 
            \rowcolor{\colMaxLower!4}\multicolumn{3}{l}{\textcolor{black}{\textbf{Lower bounds} on \textbf{maximal} expected rewards:} ~ $\forall s \in \states \colon \Emax_{s}(\reach\target) \geq \vals(s)$ \hspace*{1cm} \textcolor{gray}{[\Cref{prop:certsInfLower}]}} \\
            \rowcolor{\colMaxLower!4}\quad$\vals \in \RgeqZeroInfStates$ & $\bellmanrmax(\vals) \geq \vals$ & \textit{\textbf{max}-Bellman operator \textbf{increases} value of {all} states}\\
            \rowcolor{\colMaxLower!4}\quad$\rank \in \exnats^\states$ & $\distopmod{\min}(\rank) \leq \rank$ & \textit{$\rank$ upper bounds \textbf{minimal} distances to \textbf{end components}}\\
            \rowcolor{\colMaxLower!4}\quad & $\vals(s) {=} \infty \implies \rank(s) {<} \infty$ & \textit{\textbf{infinite} expected reward necessitates finite distance} ($\forall s \in \states$) \\ 
            \bottomrule
        \end{tabular}
    \end{adjustbox}
\end{table}

\section{More Details on Preliminaries}


So far, (in \Cref{sec:prelims}) we only defined reachability probabilities and expected rewards w.r.t. probability measures on infinite paths. For some proofs however, we also use reachability probabilities w.r.t. all paths of length $\leq n$ for a fixed $n \in \nats$ and expected rewards w.r.t. all paths of length exactly $n$.
To be precise, these so called step-bounded reachability probabilities are used in the proof of \Cref{lem:distopequiv} as well as in the proof of \Cref{lem:infMaxUpperLfp}, and the step-bounded cumulative rewards are used in the proof of \Cref{lem:infMaxUpperLfp} as well as in the proof of \Cref{thm:rhoUpperLfp}.


Fix a DTMC $(\states,\act,\prmdp)$ and a \emph{target} (or goal) set $\target \in\states$. Let $n \in \nats$ and consider a finite path $s_0s_1 \dots s_{k}$ of length $k \leq n$. We define the random variable
\[
    \lozenge^{\leq n} \target(s_0s_1 \dots s_{k}) =
    \begin{cases}
        1 & \text{if } \exists i \colon s_i \in \target\\
        0 & \text{else.}
    \end{cases}
\]
For a state $s \in \states$, we define the step-bounded reachability probabilities $\mathbb{P}_s(\lozenge^{\leq n}\target)$ from $s$ towards $\target$ as the expected value (Lebesgue integral) of $\lozenge^{\leq n} \target$.

Now fix a \emph{reward function} $\rew \colon \states \to \RgeqZero$ and a finite path $s_0s_1 \dots s_{n-1} \in \states^{n}$. We define the random variable
\[
    \rew^{\lozenge^{= n}\target}(s_0s_1 \dots s_{n-1}) =
    \begin{cases}
        \sum_{i = 0}^{\min\{j\mid s_j\in \target\}} \rew(s_i) & \text{if } \exists j \colon s_j \in \target\\
        \sum_{i = 0}^{n-1} \rew(s_i) & \text{else,}
    \end{cases}
\]
for both, the \enquote{$\starinf$} and the \enquote{$\starrho$} semantics of expected rewards. For a state $s \in \states$, we define the step-bounded cumulative rewards $\E_s(\rew^{\lozenge^{= n}\target})$ from $s$ towards $\target$ as the expected value of $\rew^{\lozenge^{= n}\target}$. We often write $\E_s(\lozenge^{= n}\target)$ instead of $\E_s(\rew^{\lozenge^{= n}\target})$ when $\rew$ is clear from the context.

Finally, given an MDP $\mdp = (\states, \act, \prmdp)$, a state $s \in \states$, a target set $\target\subseteq\states$, a reward function $\rew \colon \states \to \RgeqZero$, and $\opt \in \{\min, \max\}$ we define the \emph{optimal step-bounded reachability probability} $\propt_{s}(\reach^{\leq n} \target) = \opt_{\strat} \prind{\strat}_{s}(\reach^{\leq n} \target)$ and the \emph{optimal step-bounded cumulative expected reward} $\Eopt_{s}(\rew^{\reach^{= n} \target}) = \opt_{\strat} \Estrat_{s}(\rew^{\reach^{= n} \target})$, where $\prind{\strat}_s(\reach^{\leq n} \target)$ and $\Estrat_{s}(\rew^{\reach^{= n} \target})$ are the step-bounded reachability probabilities and step-bounded cumulative expected rewards in the induced DTMC $\mdp^\strat$.

With the \enquote{$\starrho$} semantics of expected rewards, we have for all strategies $\sigma$ that $\Estrat_s(\lozenge \target) = \sup_{n \in \nats}\Estrat_s(\lozenge^{= n} \target)$. With the \enquote{$\starinf$} semantics, the same holds only if $\prind{\strat}_s(\lozenge\target) = 1$ holds for state $s$, i.e. only if all infinite paths starting from $s$ reach the target in the induced DTMC $\mdp^\strat$.
\section{Proofs of \Cref{sec:qualreachandsafe}}

All proofs from \Cref{sec:qualreachandsafe} required for soundness were formally verified in \isabelle.
We include additional conventional pen-and-paper proofs in this appendix.
The only proof not yet formally verified is the completeness-proof of the certificates for almost-sure reachability (\Cref{lem:distopmod-completeness}).
This does not affect the correctness of the certificate checker we propose.

\subsection{Proof of \Cref{lem:distopuniquefp}}
\label{app:distopuniquefp}

\distopuniquefp*
\begin{proof}
    Towards a contradiction, assume that $\rank_1 = \lfp{\distop{\opt}{}} \neq \gfp{\distop{\opt}{}} = \rank_2$.
    Pick $s \in \states$ such that $\rank_1(s) < \rank_2(s)$ and $\rank_1(s)$ is minimal (a minimum exists because $\states$ is finite by assumption), i.e.,
    \begin{align}
        \forall t \in \states\colon~\rank_1(t) < \rank_1(s) \text{ implies } \rank_1(t) = \rank_2(t). \label{eq:distopdemonuniquefpminimalrank}
    \end{align}
    It follows that $s \in \states \setminus \target$ as otherwise $\rank_1(s) = \rank_2(s) = 0$.
    We proceed by a case distinction for choosing $\opt \in \{\min, \max\}$.
    
    \begin{itemize}
        \item Case $\opt = \max$:\\
        Let $b \in \act(s)$ and $t \in \post(s,b)$ with 
        \[
            \rank_2(s) ~=~ 1 + \min\limits_{s' \in \post(s,b)} \rank_2(s')
            \quad\text{and}\quad
            \rank_1(t) ~=~ \min\limits_{s' \in \post(s,b)} \rank_1(s').
        \]
        In particular, such $b$ exists since $r_2(s)$ is a fixed point of $\distop{\max}{}$.
        Our assumption that $r_1(s) < r_2(s)$ is contradicted since
        \begin{align*}
            \rank_1(s) 
            ~\geq~ 1 + \min\limits_{s' \in \post(s,b)} \rank_1(s')
            ~=~ 1 + \rank_1(t)
            & ~\overset{(*)}{=}~  1 + \rank_2(t)\\
            & ~\ge~ 1 + \min\limits_{s' \in \post(s,b)} \rank_2(s')\\
            & ~=~  \rank_2(s).
        \end{align*}
        The equality $(*)$ follows by applying $\eqref{eq:distopdemonuniquefpminimalrank}$ and the fact that $\rank_1(s) \ge 1 + \rank_1(t)$ implies $\rank_1(t) < \rank_1(s)$.
    
        \item Case $\opt = \min$:\\
        Let $b \in \act(s)$ and $t \in \post(s,b)$ with 
        \[
            \rank_1(s) ~=~ 1 + \min\limits_{s' \in \post(s,b)} \rank_1(s')
            \quad\text{and}\quad
            \rank_1(t) ~=~ \min\limits_{s' \in \post(s,b)} \rank_1(s').
        \]
        In particular, such $b$ exists since $r_1(s)$ is a fixed point of $\distop{\min}{}$.
        Our assumption that $r_1(s) < r_2(s)$ is contradicted since
        \begin{align*}
            \rank_1(s) 
            ~=~ 1 + \min\limits_{s' \in \post(s,b)} \rank_1(s')
            ~=~ 1 + \rank_1(t)
            & ~\overset{(*)}{=}~  1 + \rank_2(t)\\
            & ~\geq~ 1 + \min\limits_{s' \in \post(s,b)} \rank_2(s')\\
            & ~\geq~  \rank_2(s).
        \end{align*}
        The equality $(*)$ follows by applying $\eqref{eq:distopdemonuniquefpminimalrank}$ and the fact that $\rank_1(s) = 1 + \rank_1(t)$ implies $\rank_1(t) < \rank_1(s)$.
    \end{itemize}\qed
\end{proof}

\subsection{Proof of \Cref{lem:distopequiv}}
\label{app:rankandqualreach}

\distopequiv*
\begin{proof}
We show both directions seperately.
\begin{itemize}
    \item ``$\Rightarrow$'': We first prove an auxiliary statement. Let $\strat$ be a strategy, $s \in \states$ and let $\rank'$ be the unique fixed point of $\distop{\strat}{}$ (the distance operator $\distop{\min}{} = \distop{\max}{}$ in the induced DTMC $\mdp^\strat$). Show by induction that $\forall n \in \mathbb{N}$:
    \[
        \prind{\strat}_s(\lozenge^{\leq n}\target) > 0 \implies \rank'(s) \leq n
    \]
    If $n = 0$ then $s \in \target$ and thus $\rank'(s) = 0$.
    If $n > 0$ then $\prind{\strat}_s(\lozenge^{\leq n}\target) > 0$ implies that there is an $s' \in \post(s, \strat(s))$ with $\prind{\strat}_{s'}(\lozenge^{\leq n-1}\target) > 0$ and by induction hypothesis $\rank'(s') \leq n-1$. This means however that
    \[
        \rank'(s) \leq 1 + \rank'(s') \leq n.
    \]

    Now, we show the direction ``$\Rightarrow$'' by contraposition. If $\opt = \max$, then, since $\prmin_s(\lozenge \target) > 0$ and $\mdp$ is finite, there exists an $n \in \mathbb{N}$ such that for all strategies $\strat$, it holds that $\prind{\strat}_s(\lozenge^{\leq n}\target) > 0$. The previous induction implies that for all strategies $\strat$, it holds that $\fp{\distop{\strat}{}} \leq n$ which implies that $\rank(s) \leq n < \infty$.

    If $\opt = \min$, then, since $\prmax_s(\lozenge \target) > 0$ and $\mdp$ is finite, there exists an $n \in \mathbb{N}$ and a strategy $\strat$ such that $\prind{\strat}_s(\lozenge^{\leq n}\target) > 0$. The previous induction implies that $\fp{\distop{\strat}{}} \leq n$ which implies that  $\rank(s) \leq n < \infty$.
    
    \item ``$\Rightarrow$'': If $\opt = \max$, the proof is similar to \cite[Lemma 10.110]{principles}.
    We show by induction on $n \in \nats$ that for all $s \in \states$, $\rank(s) = n$ implies $\prmin_s(\reach\target) > 0$.
    
    For $n=0$ note that $\rank(s) = n = 0$ implies $s \in \target$, hence $\prmin_s(\reach\target) = 1 > 0$.
    
    Now let $n \geq 0$ be arbitrary.
    We show the claim for $n+1$.
    Let $s \in \states$ be such that $\rank(s) = n+1$ (this implies $s \notin \target$).
    Consider an arbitrary MD strategy $\strat$.
    Since $\rank$ is a fixed point of $\distop{\max}{}$, 
    \[
        \exists s' \in \states \colon~
        \prmdp(s,\strat(s),s') > 0
        ~\land~
        \rank(s) \geq 1 + \rank(s')
        ~.
    \]
    Thus $\rank(s') \leq n$ and by the induction hypothesis, $\prmin_{s'}(\reach\target) > 0$.
    But
    \begin{align*}
        \prind{\strat}_s(\reach\target)
        =
        \sum_{s''} P(s,\strat(s),s'') \cdot  \prind{\strat}_{s''}(\reach\target)
        & \geq
        \prmdp(s,\strat(s),s') \cdot \prind{\strat}_{s'}(\reach\target)\\
        & \geq
        \prmdp(s,\strat(s),s') \cdot \prmin_{s'}(\reach\target) > 0
        ~.
    \end{align*}
    It follows that $\prmin_s(\reach\target) > 0$ since $\strat$ was arbitrary.

    If $\opt = \min$, we show by induction on $n \in \nats$ that for all $s \in \states$, $\rank(s) = n$ implies $\prmax_s(\reach\target) > 0$.
    
    For $n=0$ note that $\rank(s) = n = 0$ implies $s \in \target$, hence $\prmax_s(\reach\target) = 1 > 0$.
    
    Now let $n \geq 0$ be arbitrary.
    We show the claim for $n+1$.
    Let $s \in \states$ be such that $\rank(s) = n+1$ (this implies $s \notin \target$).
    Since $\rank$ is a fixed point of $\distop{\min}{}$, there exists an MD strat $\strat$, such that
    \[
        \exists s' \in \states \colon~
        \prmdp(s,\strat(s),s') > 0
        ~\land~
        \rank(s) = 1 + \rank(s')
        ~.
    \]
    Thus $\rank(s') = n$ and by the induction hypothesis, $\prmax_{s'}(\reach\target) > 0$.
    But
    \begin{align*}
        \prind{\strat}_s(\reach\target)
        =
        \sum_{s''} P(s,\strat(s),s'') \cdot  \prind{\strat}_{s''}(\reach\target)
        & \geq
        \prmdp(s,\strat(s),s') \cdot \prind{\strat}_{s'}(\reach\target)\\
        & \geq
        \prmdp(s,\strat(s),s') \cdot \prmax_{s'}(\reach\target) > 0
        ~.
    \end{align*}
    It follows that $\prmax_s(\reach\target) > 0$ since we found at least one strategy $\strat$ such that $\prind{\strat}_s(\reach\target) > 0$.
    \end{itemize}\qed

\end{proof}

\subsection{Proof of \Cref{prop:certQualReach}}
\label{app:certQualReach}

\certQualReach*
\begin{proof}
    Suppose that $\rank$ is valid and $\rank(s) < \infty$.
	Then by \Cref{thm:knastertarski} (Knaster-Tarski) we have $(\fp{\distop{\nopt}{}})(s) \leq \rank(s) < \infty$.
	We conclude that $\propt_s(\reach\target) > 0$ by \Cref{lem:distopequiv}.\qed
\end{proof}

\subsection{Proof of \Cref{lem:distopmodequiv}}

Before we prove the claim, we need to show that the least fixed point of the operator is well-defined.
To this end, we show that the operator is monotonous, which together with \Cref{thm:knastertarski} implies that the least fixed point exists.

\begin{lemma}[Monotonicity of Complementary Distance Operator]\label{lem:distopmodMonotone}
	Let $\rank_1, \rank_2 \colon 	~ \exnats^\states \to \exnats^\states$ be two ranking functions with $\rank_1 \leq \rank_2$.
	Then $\distopmod{\opt}(\rank_1) \leq \distopmod{\opt}(\rank_2)$.
\end{lemma}
\begin{proof}
	This proof is a mostly straightforward application of definitions, mainly \Cref{def:distopmod}.	
	Recall that the comparison $\rank_1 \leq \rank_2$ is point-wise, so we have to prove for all $s\in\states$ that $\distopmod{\opt}(\rank_1)(s) \leq \distopmod{\opt}(\rank_2)(s)$.
	For $s\in\target$, we have $\distopmod{\opt}(\rank_1) = \distopmod{\opt}(\rank_2) = \infty$ by \Cref{def:distopmod}.
	
	For $s\in\states \setminus \target$, the proof is more involved.
	We prove below that for all $a\in\act(s)$ we have 
	\begin{align}
		&\underset{s' \in \post(s,a)}{\min}  \rank_1(s')
		~+~ [\exists s',s'' \in \post(s,a): \rank_1(s')\neq\rank_1(s'')] \notag\\
		\leq &\underset{s' \in \post(s,a)}{\min}  \rank_2(s')
		~+~ [\exists s',s'' \in \post(s,a): \rank_2(s')\neq\rank_2(s'')] \label{eq:distopmodproof}
	\end{align}
	
	Using this, we prove the claim of monotonicity as follows (only giving the case of $\opt=\max$, as the other case is analogous):
	\begin{align*}
		\distopmod{\opt}(\rank_1) &= 
		\underset{a \in \act(s)}{\max} ~ \underset{s' \in \post(s,a)}{\min}  \rank_1(s')
		~+~ [\exists s',s'' \in \post(s,a): \rank_1(s')\neq\rank_1(s'')] \tag{By \Cref{def:distopmod}.}\\
		&= 
		\underset{s' \in \post(s,a_1)}{\min}  \rank_1(s')
		~+~ [\exists s',s'' \in \post(s,a_1): \rank_1(s')\neq\rank_1(s'')] 
		\tag{By picking an arbitrary optimal action $a_1$}\\
		&\leq
		\underset{s' \in \post(s,a_1)}{\min}  \rank_2(s')
		~+~ [\exists s',s'' \in \post(s,a_1): \rank_2(s')\neq\rank_2(s'')]
		\tag{By \Cref{eq:distopmodproof}}\\
		&\leq
		\underset{a \in \act(s)}{\max} ~ \underset{s' \in \post(s,a)}{\min}  \rank_2(s')
		~+~ [\exists s',s'' \in \post(s,a): \rank_2(s')\neq\rank_2(s'')]
		\tag{By definition of $\max$}\\
		&=
		\distopmod{\opt}(\rank_2) \tag{By \Cref{def:distopmod}}			
	\end{align*}
	
	It remains to prove \Cref{eq:distopmodproof}.
	For this, we let $a\in\act(s)$ be an arbitrary action and proceed by case distinctions about the Iverson brackets.
	\begin{itemize}
		\item \textbf{Case 1} ($\rank_2$ Iverson evaluates to 1): $\exists s',s'' \in \post(s,a): \rank_2(s')\neq\rank_2(s'')$.
		The following chain of equations proves our goal:
		\begin{align*}
			&\phantom{\leq} \underset{s' \in \post(s,a)}{\min}  \rank_1(s')
			~+~ [\exists s',s'' \in \post(s,a): \rank_1(s')\neq\rank_1(s'')]\\
			&\leq	
			\underset{s' \in \post(s,a)}{\min}  \rank_1(s')
			~+~ 1 
			\tag{Since Iverson bracket is at most 1}\\
			&\leq 
			\underset{s' \in \post(s,a)}{\min}  \rank_2(s')
			~+~ 1 
			\tag{By assumption: $\rank_1(s')\leq\rank_2(s')$}\\
			&=
			\underset{s' \in \post(s,a)}{\min}  \rank_2(s')
			~+~ [\exists s',s'' \in \post(s,a): \rank_2(s')\neq\rank_2(s'')]
			\tag{By Case 1 assumption: Iverson bracket evaluates to 1}\\
		\end{align*}
		\item \textbf{Case 2} ($\rank_2$ Iverson evaluates to 0): 
		$\forall s',s'' \in \post(s,a): \rank_2(s')=\rank_2(s'')$.
		\begin{itemize}
			\item \textbf{Case 2a} ($\rank_1$ Iverson evaluates to 0):
			$\forall s',s'' \in \post(s,a): \rank_1(s')=\rank_1(s'')$.
			\begin{align*}
				&\phantom{\leq} \underset{s' \in \post(s,a)}{\min}  \rank_1(s')
				~+~ [\exists s',s'' \in \post(s,a): \rank_1(s')\neq\rank_1(s'')]\\
				&=	
				\underset{s' \in \post(s,a)}{\min}  \rank_1(s') ~+~ 0
				\tag{By Case 2a assumption: Iverson bracket evaluates to 0}\\
				&\leq 
				\underset{s' \in \post(s,a)}{\min}  \rank_2(s') ~+~ 0
				\tag{By assumption: $\rank_1(s')\leq\rank_2(s')$}\\
				&=
				\underset{s' \in \post(s,a)}{\min}  \rank_2(s')
				~+~ [\exists s',s'' \in \post(s,a): \rank_2(s')\neq\rank_2(s'')]
				\tag{By Case 2 assumption: Iverson bracket evaluates to 0}\\
			\end{align*}
			\item \textbf{Case 2b} ($\rank_2$ Iverson evaluates to 1): 
			$\exists s',s'' \in \post(s,a_1): \rank_1(s')\neq\rank_1(s'')$.
			
			Let $t \in \argmin_{s' \in \post(s,a)} r_1(t)$. 
			By Case 2b assumption, there exists a $t'$ with $\rank_1(t)\neq\rank_1(t')$; since ranks are natural numbers and $t$ was picked as the $\argmin$, we have $\rank_1(t) +1 \leq \rank_1(t')$.
			\begin{align*}
				&\phantom{\leq} \underset{s' \in \post(s,a)}{\min}  \rank_1(s')
				~+~ [\exists s',s'' \in \post(s,a): \rank_1(s')\neq\rank_1(s'')]\\
				&=	
				\rank_1(t) +1 
				\tag{By choice of $t$ and Case 2b assumption}\\
				&\leq
				\rank_1(t')
				\tag{By argument above}\\
				&\leq 
				\rank_2(t') 
				\tag{By assumption: $\rank_1(s')\leq\rank_2(s')$}\\
				&=
				\underset{s' \in \post(s,a)}{\min}  \rank_2(s')
				\tag{By Case 2 assumption, all successors have the same $\rank_2$}\\
				&=
				\underset{s' \in \post(s,a)}{\min}  \rank_2(s')
				~+~ [\exists s',s'' \in \post(s,a): \rank_2(s')\neq\rank_2(s'')]
				\tag{By Case 2 assumption: Iverson bracket evaluates to 0}\\
			\end{align*}
		\end{itemize}
	\end{itemize}
\qed
\end{proof}

Moreover, the following claim is useful and instructive.
Define  the set of sink states that cannot reach the target set under optimal strategies.
\[Z := \{s \mid \pr_s^{\opt}(\lozenge \target) = 0\}\] 
(Note that when $\opt=\min$, these states might have a path to the target set, but also a way to avoid reaching it.)
Then we have:
\begin{lemma}\label{lem:ZiffRankZero}
	Let $\opt \in \{\min, \max\}$ and let $\rank = \lfp{\distopmod{\opt}}$.
	It holds for all $s \in \states$ that $\rank(s) = 0 \iff s\in Z$.
\end{lemma}
\begin{proof}
	\noindent\textbf{Forward Direction $\implies$:}
	In words, we want to prove that if a state has rank 0, it is a sink state.
	We first prove the claim for $\opt=\max$ and then mention the differences for the $\opt=\min$ case.
	
	\noindent\textbf{Case $\opt=\max$:}
	
	Assume $\rank(s) = 0$.
	Since $\rank$ is a fixed point, we know that $\distopmod{\max}(\rank)(s) = \rank(s) = 0$. Thus, by \Cref{def:distopmod}, we know that for all $a\in\act(s), s'\in\post(s,a)$, we have $\rank(s')=0$ (all actions since we take the action with maximum rank, and all successors since if there was one with a larger rank, the Iverson bracket would evaluate to 1).
	This is not only true for $s$, but for all states with rank~0.
	
	Consider the set $X^*$, defined as the fixpoint of the following recursion:
	$X_0 = \{s\}$ and $X_{i+1} = X_i \cup \{p \in \post(q,a) \mid \exists q \in X_i, a\in\act(q)\}$.
	$X^*$ is well defined, since there are finitely many states, and every iteration adds a state or terminates.
	
	By the above argument, all states in $X^*$ have rank 0.
	Moreover, $X^*$ is exactly the set of states reachable from $s$ under any strategy.
	Since target states have rank infinity, there can be no target state in $X^*$, and thus no target state reachable from $s$. 
	Thus, $\pr_s^{\max}(\lozenge \target) = 0$, and $s\in Z$
	\smallskip
	%
	%
	
	\noindent\textbf{Case $\opt=\min$:}
	
	In the case of $\opt=\min$, the argument changes slightly: 
	The set $X^*$ is not the set of all reachable states, but the set of states reachable under some particular strategy that always picks the action minimizing the rank.
	Still, all states in $X^*$ have the property that they have rank 0 and an available action that keeps the path in $X^*$.
	Thus, there exists a strategy that keeps the play in $X^*$ which cannot contain a target state. Thus,  $\pr_s^{\min}(\lozenge \target) = 0$, and $s \in Z$.
	\medskip

	\noindent\textbf{Backward Direction $\impliedby$:}
	Let $\rank'(q) = 0$ for all $q\in Z$, and be arbitrary otherwise.
	We show that for all $q\in Z$, $\distopmod{\opt}(\rank')(q)= \rank'(q)$, and hence the least fixed point of $\distopmod{\opt}$ assigns 0 to all states in $Z$ (independent of the ranks of other states).
	Since $Z\cap T=\emptyset$, we always use the second case of \Cref{def:distopmod}.
	
	\noindent\textbf{Case $\opt=\max$:}
	
	Since $q\in Z$, $\pr_q^{\max}(\lozenge \target) = 0$.
	Hence, for all actions $a \in \act(q)$, all successors $q'\in \post(q,a)$ are also in $Z$, since otherwise $q$ could reach a state that has positive reachability probability, and hence could reach the target set with positive probability.
	Thus, for all actions, the Iverson bracket evaluates to 0, and thus, $\distopmod{\max}(\rank')(q) = 0$, concluding the case. 
	\smallskip
	
	\noindent\textbf{Case $\opt=\min$:}
	 
	Since $q\in Z$, $\pr_q^{\min}(\lozenge \target) = 0$.
	Hence, there exists an action $a \in \act(q)$ where all $q'\in \post(q,a)$ are also in $Z$, since otherwise $q$ would be forced to reach a state with positive reachability probability, and reach the target set with positive probability.
	For this action, the Iverson bracket evaluates to 0, and thus, $\distopmod{\min}(\rank')(q) = 0$, concluding the case.
	\qed
	
\end{proof}

\newcommand{\UZ}{\states \setminus \target\mathbin\mathsf{U}Z}
\newcommand{\pZ}{z(s)}

\begin{lemma}[Soundness]
	\label{lem:distopmod-soundness}%
	Let $\opt \in \{\min, \max\}$ and let $\rank = \lfp{\distopmod{\opt}}$.
	Let $\sigma$ be an optimal strategy for reaching $\target$ according to $\opt$.
	Let $\pZ = \pr_s^{\sigma}(\UZ)$ for $s \in \states$,
	where $\pr_s^{\sigma}(\UZ)$ denotes the probability that the play reaches $Z$ before reaching $\target$ under strategy $\sigma$ (constrained reachability problem).
	It holds for all $s \in \states$ that $\rank(s) < \infty \implies \pZ > 0$.
\end{lemma}
\begin{proof}	
	We show by induction the following statement:
	$
	\rank(s) \leq n \implies \pZ > 0.
	$
	Using this, for every natural number $n$, we have $\rank(s) = n \implies \pZ > 0$.
	Since $\exnats = \nats \cup \infty$, this means that $\rank(s) < \infty \implies \pr_s^{\sigma}(\UZ) > 0$.
	
	\noindent \textbf{Base Case:} $\rank(s) = 0 \implies \pZ > 0$
	
		From \Cref{lem:ZiffRankZero}, we know that if $\rank(s)=0$, then $s\in Z$.
		Thus, trivially, for all $\sigma$, $\pZ =1 >0$.
		
	\noindent \textbf{Induction Hypothesis:} $\rank(s) \leq n \implies \pZ > 0$
		
		Intuitively, if the rank of a state is finite, under optimal play it reaches the set of sink states $Z$ (before $\target$) with positive probability.
		Note that because of the Iverson bracket, the counting in $\rank$ is not perfect as in \Cref{lem:distopequiv}, i.e.\ it does not give the length of the shortest path that reaches the set. 
		
	\noindent \textbf{Induction Step:} $\rank(s) \leq n + 1 \implies \pZ > 0$		
		
		\begin{itemize}
		\item \textbf{Case $\opt=\min$:}
		
		Assume $\rank(s) = n+1$.
		Since $\rank$ is a fixed point, we know that also $\distopmod{\max}(\rank)(s) = \rank(s) = n+1$.
		Now there are two possibilities:
		\begin{itemize}
			\item Case 1 (Iverson bracket 1): 
			For an action $a$ minimizing the rank, there exists an $s'\in\post(s,a)\setminus \target$ with $\rank(s') = n$ where the Iverson bracket evaluates to 1.
			Intuitively, the player can decrease the rank, hence making progress towards the sink states $Z$. 
			Formally, using action $a$, the probability to reach in one step a state with $\rank \leq n$ is positive.
			Applying the Induction Hypothesis yields that the probability to reach $Z$ is positive under the strategy that first plays $a$ and then optimally from $s'$. 
			Note that we only showed existence of a strategy that reaches $Z$ (before $\target$) with positive probability, but did not yet consider that we talk about an optimal strategy. 
			However, for $\opt=\min$, an optimal strategy will never decrease the chance of reaching $Z$, but would only increase it, so proving this existence suffices.
			\item Case 2 (Iverson bracket 0): 
			For all actions $a\in\act(s)$ minimizing the rank, we have that all $s'\in \post(s,a)$ have $\rank(s')=n+1$.
			We use a construction similar to the one in the base case:
			Let $X^*$ be the fixpoint of the following recursion:
			$X_0 = \{s\}$ and $X_{i+1} = X_i \cup \{p \in \post(q,a) \mid \exists q \in X_i, \forall a\in\act(q), \forall p' \in \post(q,a): \rank(p') = n+1\}$.
			Again, since the number of states is finite, $X^*$ is well-defined.
			Also, $X^* \cap \target = \emptyset$, since the rank of all states in $X^*$ is finite.
			If for some $i$, we reach a state $q$ that does not satisfy $\forall a\in\act(q), \forall p' \in \post(q,a): \rank(p') = n+1$, then by the argument in Case 1, this $q$ reaches $Z$ with positive probability.
			By following the strategy that has positive probability to reach $q$ from $s$ and then applying the Induction Hypothesis, we can conclude that $q$ can reach $Z$ with positive probability.
			
			The only remaining case is that for all $q \in X^*$, we have $\forall a\in\act(q), \forall p' \in \post(q,a): \rank(p') = n+1$.
			That means that there exists a strategy that from $s$ keeps the play inside $X^*$; in fact, every strategy does.
			From $X^* \cap \target = \emptyset$, we obtain $\pr_s^{\min}(\lozenge \target) = 0$, and $s\in Z$.
		\end{itemize}
		Overall, we have shown that for $\opt=\min$, if a rank is finite, then there exists a strategy that reaches $Z$ before $\target$ with positive probability.
		So in particular, the optimal strategy for minimizing the probability to reach $\target$ will also do so.
			\item \textbf{Case $\opt=\max$:}
		
		This case is mostly analogous to the previous one, with the differences arising from two sources: The existential and all quantification of actions is inverted, and the argument in the case that we found a set $X^*$ where all states have rank $n+1$ is different.
		
		Assume $\rank(s) = n+1$.
		Since $\rank$ is a fixed point, we know that also $\distopmod{\max}(\rank)(s) = \rank(s) = n+1$.
		Now there are two possibilities:
		\begin{itemize}
			\item Case 1 (Iverson bracket 1): 
			For all actions $a$ maximizing the rank, there exists an $s'\in\post(s,a)$ with $\rank(s') = n$ where the Iverson bracket evaluates to 1.
			Since this is the maximizing action, all other actions $b$ have that there exists an $s'\in\post(s,a) \setminus \target$ with $\rank(s') \leq n$.
			
			Intuitively, for all actions the player cannot avoid decreasing the rank, hence making progress towards the sink states $Z$. 
			Formally, for every strategy, the probability to reach in one step a state with $\rank \leq n$ is positive.
			Applying the Induction Hypothesis yields that the probability to reach $Z$ before $\target$ is positive.
			
			\item Case 2 (Iverson bracket 0): 
			There exists an action $a\in\act(s)$ maximizing the rank where all $s'\in \post(s,a)$ have $\rank(s')=n+1$.
			Let $X^*$ be the fixed point of the following recursion:
			$X_0 = \{s\}$ and $X_{i+1} = X_i \cup \{p \in \post(q,a) \mid \exists q \in X_i, a\in\act(q): \forall p' \in \post(q,a): \rank(p') = n+1\}$.
			
			
			We briefly interrupt the proof to explain in detail the motivation behind including the Iverson bracket, which is visible exactly in this case (although it is only necessary for the proof of the backward direction below):
			The local goal of the operator $\distopmod{\max}$ is to maximize the distance to $Z$, so remaining in $X^*$ seems like the best choice. 
			In fact, without the Iverson bracket, all states in $X^*$ would have infinite rank, as they can count up depending on each other.
			Since infinite rank should denote almost sure reachability, this is why we use the Iverson bracket to keep the rank finite at $n+1$.
			
			Returning to the proof, we now show that for all $q\in X^*$, we have $z(q) > 0$ (thus in particular proving that $s$ also reaches $Z$ with positive probability under an optimal strategy).
			
			First, consider the case that there is no state-action pair that leaves $X^*$, formally for all $q\in X^*$ and $a\in\act(q)$, we have $\post(q,a)\subseteq X^*$.
			Then, since $X^*$ cannot contain any target states (as their rank would be infinite, not $n+1$), and since all actions remain in $X^*$, we have $X^*\subseteq Z$.
			Thus $\pZ = 1 > 0$.
			
			Second, consider the case that there exists a state-action pair that leaves $X^*$, i.e.\ there exists a $q\in X^*$ and $a\in\act(q)$ such that $\post(q,a) \nsubseteq X^*$.
			For every such leaving state-action pair, we have that there exists a $q'\in \post(q,a)$ with $\rank(q') \leq n$: The rank $\rank(q')$ cannot be greater than $n+1$, as otherwise $\rank(q)$ would be greater than $n+1$, and it cannot be equal to $n+1$, because otherwise $q'$ would have been added to $X^*$.
			Thus, every strategy using a leaving state-action pair reaches a state with $\rank$ less than or equal to $n$, and by Induction Hypothesis, then this strategy has positive probability to reach $Z$.
			
			Finally, note that every strategy remaining in $X^*$ forever gets value 0, as $X^*$ does not contain any target state.
			Consequently, an optimal strategy has to use a leaving state-action pair, and thus decrease the rank.
		\end{itemize}
		
		Overall, we have shown that for $\opt=\max$, if a rank is finite, then an optimal strategy for reaching $\target$ will reach $Z$ (constrained to $\states \setminus \target$) with positive probability (avoiding to be stuck in a cycle and not reaching the target set at all).
		\end{itemize}
	
	Together, we have shown that for $\opt$ being $\min$ or $\max$, if the rank is finite, then the sink states are reached with positive probability under strategies optimal for reaching $\target$.\qed
\end{proof}

\begin{lemma}[Completeness]
	\label{lem:distopmod-completeness}%
	Let $\opt \in \{\min, \max\}$ and let $\rank = \lfp{\distopmod{\opt}}$.
	Let $\sigma$ be an optimal strategy for reaching $\target$ according to $\opt$.
	Let $\pZ = \pr_s^{\sigma}(\UZ)$ for $s \in \states$,
	where $\pr_s^{\sigma}(\UZ)$ denotes the probability that the play reaches $Z$ before reaching $\target$ under strategy $\sigma$ (constrained reachability problem).
	It holds for all $s \in \states$ that $\rank(s) < \infty \impliedby \pZ > 0$.
\end{lemma}
\begin{proof}
	We start with the assumption that $\pZ > 0$ and assume for contradiction $\rank(s) = \infty$.
	We make a case distinction on whether $\opt$ used by $\distopmod{\opt}$ is $\max$ or $\min$.
	\smallskip
	
	\noindent\textbf{Case $\opt=\min$:}
	
	In this case, since $\rank$ is a fixpoint of $\distopmod{\min}$, using \Cref{def:distopmod} we get that for all states $q\in\states$: $\rank(q)=\infty$ implies that for all $a\in\act(q), s'\in\post(q,a)$ we have $\rank(s')=\infty$.
	
	As in \Cref{lem:distopmod-soundness}, we inductively construct a set $X^*$, defined as the least fixpoint of:
	$X_0 = \{s\}$ and $X_{i+1} = X_i \cup \{p \in \post(q,a) \mid \exists q \in X_i \setminus \target , a\in\act(q)\}$.
	The set $X^*$ is well defined, since there are finitely many states, and every iteration adds a state or terminates.
	Moreover, by the above argument, all successors of states not in $\target$ have rank infinity and hence we have for all $q\in X^*$ that $\rank(q)=\infty$.
	
	However, since $\pZ > 0$, we also have $Z \cap X^* \neq \emptyset$.
	\Cref{lem:ZiffRankZero} shows that $\rank(q)=0$ for all $q\in Z$ when $\rank$ is the least fixed point of $\distopmod{\opt}$.
	This is a contradiction to the rank being infinite for all states in $X^*$, concluding the case.
	\smallskip
	
	%

	\noindent\textbf{Case $\opt=\max$:}
	
	In this case, since $\rank$ is a fixpoint of $\distopmod{\max}$, using \Cref{def:distopmod} we get that for all states $q\in\states$,  $\rank(q)=\infty$ implies that there exists $a\in\act(q)$ such that for all $s'\in\post(q,a)$ we have $\rank(s')=\infty$.
	We construct $X^*$ as the least fixpoint of:
	$X_0 = \{s\}$ and $X_{i+1} = X_i \cup \{p \in \post(q,a) \mid \exists q \in X_i \setminus \target, a\in\act(q): \forall q'\in\post(q,a): \rank(q')=\infty\}$.
	
	Note that this $X^*$ contains all states that are reachable from $s$ when following actions that maximize the rank (until we visit $\target$).
	These actions need not be those that are used by the optimal strategy $\sigma$.
	Thus, we cannot simply conclude that $X^*\cap Z \neq \emptyset$ as in the $\opt=\min$ case.
	We now show that there have to be some states in $X^*$ that should have finite rank, thus contradicting the assumption $\rank(s)=\infty$ and proving our goal.
	Intuitively, we show that if $s$ has infinite rank, then there has to be some set of states reachable from $s$ that wrongly has infinite rank.
	
	As the proof is quite involved, we adopt a proof style of providing the major claims, and then indenting the necessary definitions and reasoning.
	This way, the overall idea can be obtained by reading only the major items.
	
	\begin{enumerate}
		\item In the following, let $Y$ be a \emph{maximal end component} that is \emph{bottom} in $X^*$ with $Y\cap T = \emptyset$.
		
		We provide definitions and the proof that such a $Y$ exists:
		\begin{itemize}
			\item \textbf{Definition of end component:}
			An end component (EC)~\cite[Def. 3.13]{dealfaro97} is a set of states $Y$ such that there exists a set of actions $B \subseteq \bigcup_{q\in Y} \act(q)$ (with $\act(q)\cap B \neq \emptyset$ for all $q\in Y$) and the following two conditions are satisfied:
			(i) for all $q\in Y$ and $a\in \act(q)\cap B$, we have $\post(q,a)\subseteq Y$, and 
			(ii) the graph induced by $(Y,B)$ (see~\cite[Def. 3.12]{dealfaro97}) is strongly connected.
			A \emph{maximal end component (MEC)} is an end component $Y$ such that there exists no $Y'\supset Y$ that is also an end component.
			\item We identify $X^*$ with the sub-MDP induced by $X^*$ with self-loops at all states in $X^* \cap \target$. Therefore, every state in $\target$ forms it's own MEC. 
			\item \textbf{Definition of bottom:}
			A MEC is \emph{bottom} in $X^*$ if for all $q\in Y$ and $a \in \act(q)$ with $\post(q,a) \nsubseteq Y$ we also have $\post(q,a) \nsubseteq X^*$.
			
			\item For all $q \in X^* \setminus \target$ there exists $a\in \act(q)$ such that all $q\in \post(q,a)$ have $\rank(q')=\infty$.
			
			\textbf{Reason:} Assume there was a $q \in X^* \setminus \target$ where all actions $a \in \act(q)$ had a successor $q'\in\post(q,a)$ with $\rank(q')<\infty$.
			Then, $\distopmod{\max}(\rank)(q) < \infty$, which is a contradiction to $\rank$ being a fixed point.
			\item There exist ECs in $X^*$.
			
			\textbf{Reason:} 
			We know that there is a strategy that keeps the play inside the sub-MDP $X^*$, namely the one playing the staying actions whose existence we proved in the previous step.
			By~\cite[Thm. 3.2]{dealfaro97}, under this strategy an EC in the sub-MDP is reached with probability 1.
			Thus, there have to be ECs contained in $X^*$.
			
			\item There exist bottom MECs in $X^*$.
			
			Such a MEC must exist, since we can topologically order the MECs and pick one at the end of a chain. The actions exiting $Y$ must lead outside of $X^*$, because if they lead to another state $p\in X^*$, either all actions of this state would lead out of $X^*$ which leads to a contradiction as in the first step, or it would have to form an EC, which contradicts $Y$ being bottom.
			\ms{Maybe we should just shorten this to every (sub)MDP contains a bottom MEC + citation}
			
			\item For every MEC $Y'$ that is not bottom in $X^*$, there exists a strategy that almost surely reaches a state $q'\in X^*\setminus Y'$.
			
			\textbf{Reason:} Since $Y'$ is not bottom in $X^*$, there exists an exiting state $e\in Y'$ and $a\in \act(q)$ such that $\post(q,a) \nsubseteq Y$, but $\post(q,a) \subseteq X^*$.
			Since $Y'$ is an EC, its underlying graph is strongly connected and from every $q\in Y'$ we can play a strategy that almost surely reaches $e$.
			By playing $a$ in $e$, we almost surely reach states outside of the non-bottom MEC $Y$ that are still in $X^*$.
			
			\item There exists a bottom MEC $Y \subseteq X^*$ such that $Y\cap \target = \emptyset$.
			
			\textbf{Reason:} Assume for contradiction that all bottom MECs in $X^*$ contain a target state. 			
			By~\cite[Thm. 3.2]{dealfaro97}, we almost surely reach an EC under all strategies.
			Consider the strategy that in every non-bottom EC $Y'$ leaves $Y'$, which exists by the previous step.
			Under this strategy, we almost surely reach a bottom MEC in $X^*$. 
			Note that upon reaching an EC, we can almost surely reach all states inside this EC.
			\ms{again, we could cite the fact that a.s. a bottom MEC is reached}
			
			Hence, if all bottom MECs contain a target state, $s$ can almost surely reach a target state by following the strategy that reaches bottom MECs. This is a contradiction to the initial assumption of the lemma, namely $\pZ > 0$.
			Consequently, the assumption that all bottom MECs contain a target state must be wrong, and thus there exists a bottom MEC that does not contain a target state.
		\end{itemize}
		
		\item Define a ranking function $\rank_0$ that is finite on $Y$.
		\begin{itemize}
			\item Let $E := \{(q,a) \mid q\in Y, a\in\act(q): \post(q,a) \nsubseteq Y\}$ be the set of exits from~$Y$.
		
			\item For all $(q,a) \in E$, we have $\exists q' \in \post(q,a)$ with $\rank(q') < \infty$.
		
			\textbf{Reason:} Since $Y$ is bottom in $X^*$, we have $\post(q,a) \nsubseteq X^*$. If all $q' \in \post(q,a)$ had $\rank(q') = \infty$, we would have $q'\in X^*$.
		
			\item Let $n := \max_{(q,a)\in E} \min_{q' \in \post(q,a)} \rank(q')$ be the maximum rank of any exiting action from $Y$.
			By the previous step, $n<\infty$.
		
			\item Let $\rank_0(q) := \begin{cases} n+1 &\text{ if } q\in Y\\ \rank(q) &\text{ otherwise}\end{cases}$ be a modified ranking function.	
		\end{itemize}
	
		\item From $\rank_0$, we can construct a ranking function $\rank'$ that is a fixed point of $\distopmod{\max}$.
		
		\begin{itemize}
			\item Let $\rank_i := (\distopmod{\max})^i(\rank_0)$.
			\item There exists an $i$ such that $\rank_i = \distopmod{\max}(\rank_i)$.
			\begin{itemize}
				\item  For all $q\in Y$ and all $i\in\nats$, we have $\rank_{i}(q) =n+1$.
				
				\textbf{Reason:} This is the crucial step showing that even in ECs, the complementary distance operator correctly keeps the rank finite.
				Here, we use the fact that \Cref{def:distopmod} utilizes Iverson brackets (as without them, in ECs, the rank would count to infinity, and the Backward Direction would not hold, since all states in ECs would have infinite rank independent of their reachability probability).
				
				Let $q\in Y$. We perform an induction on $i$, the number of times the complementary distance operator was applied.
				Trivially, for $i=0$ we have $\rank_0(q) = n+1$.
				
				For the induction step when we apply the distance operator the $(i+1)$-st time, we separately consider actions staying in $Y$ or exiting it.
				
				For all staying actions $a \in \act(q) \setminus \{b \mid (q,b) \in E\}$, we have that $\post(q,a) \subseteq Y$ by the fact that $Y$ is an EC. 
				Thus, by Induction Hypothesis, all $q'\in\post(q,a)$ have $\rank_i(q') = n+1$, and the complementary distance operator evaluates this action to $n+1$, using the fact that Iverson bracket evaluates to 0.		
				
				Recall that we chose $n := \max_{(q,a)\in E} \min_{q' \in \post(q,a)} \rank(q')$.
				Thus, for all exiting actions $a \in \act(q) \cap \{b \mid (q,b) \in E\}$, there exists a $q'\in\post(q,a)$ with $\rank_i(q') \leq \rank(q') \leq n$.
				Thus, the complementary distance operator evaluates all exiting actions to at most $n+1$.
				
				Overall, when applying the complementary distance operator the $(i+1)$-st time, the staying actions maximize the rank and set it to $n+1$.
				Thus, $\rank_{i+1}(q) = n+1$ and we completed the induction.
				
				\item We have $\rank_{1} \leq \rank_0$.
				
				\textbf{Reason:} 
				For $q\in Y$, the previous step proves this. For $q\in\target$, we trivially have infinite rank in both $\rank_1$ and $\rank_0$.
				To prove it for the remaining states, fix $q\in \states \setminus (Y\cup\target)$.
				
				\begin{itemize}
					\item Define $\rank(q,a) = \underset{q' \in \post(q,a)}{\min}  \rank_0(q') \,+\, [\exists s',s'' \in \post(q,a): \rank_0(q')\neq\rank_0(q'')]$, and analogously for $\rank_0$.
					\item For all $a\in\act(q)$, we have $\rank_0(q,a) \leq \rank(q,a)$.
					
					\textbf{Reason:}
					If $\rank(q,a) = \infty$, the statement trivially holds.
					If $\post(q,a)\cap Y =\emptyset$, then $\rank_0(q') = \rank(q')$ for all $q'\in \post(q,a)$, and the statement holds.
					
					If $\post(q,a)\cap Y \neq\emptyset$ and $\rank(q,a) = m < \infty$ for some $m\in\nats$, then $\underset{q' \in \post(q,a)}{\min}  \rank(q') = m-1$. This is because the Iverson bracket evaluates to 1, since there exists the $q'$ with $\rank(q')=m$ and a $q'' \in Y$ with $\rank(q'')=\infty$.
					
					Now, if $n+1 \geq m-1$, we have $\rank_0(q,a) = m-1 = \rank(q,a)$; and if $n+1 < m-1$, we have $\rank_0(q,a) = n+1 < \rank(q,a)$.
					
					\item We conclude using \Cref{def:distopmod}, the previous step, the fact that $\rank$ is a fixed point and that $\rank_0(q)=\rank(q)$ for $q\notin Y$:
					
					$\rank_1(q) = \distopmod{\max}(\rank_0)(q) = \underset{a \in \act(q)}{\max} \rank_0(q,a) \leq \underset{a \in \act(q)}{\max} \rank(q,a) =$\\ $\distopmod{\max}(\rank)(q) = \rank(q) = \rank_0(q)$.
				\end{itemize}
				
				\item We have $\rank_{i+1} \leq \rank_i$.
				
				\textbf{Reason:}
				We use monotonicity (\Cref{lem:distopmodMonotone}) and the previous step to show:
				$\rank_{i+1} = (\distopmod{\max})^i(\rank_1) \leq (\distopmod{\max})^i(\rank_0) = \rank_i$.
				
				\item By the previous step, the sequence $(\rank_i)_{i\in\nats}$ is monotonically decreasing on a bounded domain (since the minimum rank is 0). Thus, in finitely many iterations, we get $\rank_{i+1}=\rank_i$.
			\end{itemize}
		
	\end{itemize}
		
		\item Let $\rank' := \lim_{i\to\infty} (\distopmod{\max})^i (\rank_0)$ be the fixed point induced by $\rank_0$.
		
		\textbf{Reason:} By the previous step, we know that such a fixed point exists.
		
		\item $\rank' < \rank$.
		
		\textbf{Reason:} We have $\rank_0 \leq \rank$ (since we only decreased the rank of the states in $Y$ from $\infty$ to $n+1$), so also for all $i\in\nats$ we have $\rank_i = (\distopmod{\max})^i(\rank_0) \leq (\distopmod{\max})^i(\rank) = \rank$ by monotonicity (\Cref{lem:distopmodMonotone}) and $\rank$ being a fixed point.
		
		Moreover, above we showed that for all $q\in Y$ and all $i\in\nats$, we have $\rank_{i}(q) =n+1 < \infty = \rank(q)$.
		Thus, $\rank'$ is strictly smaller than $\rank$ on $Y$, and less than or equal on all other states.
		
		\item We have proven our goal: $\rank$ is not the least fixed point of $\distopmod{\max}$, which is a contradiction.
		Thus the initial assumption that $\rank(s)=\infty$ even though $\pZ > 0$ was wrong, and we have
		\[
		\pZ > 0 \implies \rank(s)<\infty
		\]
	\end{enumerate}
	\qed
\end{proof}

\distopmodEquiv*
\begin{proof}	
	Let $s\in\states$ be an arbitrary state.
	We rephrase the goal by negating both sides of the equivalence, stating that the rank is finite if and only if the reachability probability is strictly less than 1, formally
	$\rank(s) < \infty \iff \pr_s^{\opt}(\lozenge \target) < 1$.
	We now reformulate the goal of almost surely reaching the target set in terms of reaching $Z$ before $\target$ with positive probability.
	Let $\sigma$ be an optimal strategy for reaching the target set. 
	Let $\pZ = \pr_s^{\sigma}(\UZ)$ for $s \in \states$,
	where $\pr_s^{\sigma}(\UZ)$ denotes the probability that the play reaches $Z$ before reaching $\target$ under strategy $\sigma$.
	We have $\pr_s^{\opt}(\lozenge \target) < 1 \iff \pZ > 0$ by standard arguments (consult the \isabelle proof for details).
	Thus, we can rephrase the goal as 
	\[
	\rank(s) < \infty \iff \pZ > 0.
	\]
	
	Now, applying \Cref{lem:distopmod-soundness,lem:distopmod-completeness} yields the claim.
	\qed
\end{proof}

\subsection{Proof of \Cref{prop:certQualReachTwo}}
\label{app:certQualReachTwo}

\certQualReachTwo*
\begin{proof}
    Suppose that $\rank$ is valid and let $\rank(s) < \infty$ for a state $s \in \states$.
    By \Cref{thm:knastertarski} (Knaster-Tarski) we have $(\lfp{\distopmod{\opt}})(s) \leq \rank(s) < \infty$.
    Hence by \Cref{lem:distopmodequiv}, $\propt_s(\reach\target) < 1$.
\end{proof}

\section{Proofs of \Cref{sec:certificates}}
All proofs from \Cref{sec:certificates} were formally verified in \isabelle.
We include additional conventional pen-and-paper proofs in this appendix.

\subsection{Proof of \Cref{thm:reachislfp}}

We remark that, while this result is standard in the literature, formalizing its proof in Isabelle was surprisingly complicated: It requires proving continuity of the Bellman operator, and we did not find any source that provided a fully formal proof of this claim.

\subsection{Proof of \Cref{thm:certlower}}
\label{app:certlower}

\certlower*
\begin{proof}
    Let $(\vals, \rank)$ be valid.
    By \Cref{thm:knastertarski} (Knaster-Tarski) and \Cref{thm:uniquefp} it suffices to show that $\vals \leq \modbellmanmin(\vals)$.
    Let $s \in \states$ be arbitrary.
    If $\vals(s) = 0$, then the inequality $\vals(s) \leq \modbellmanmin(\vals)(s)$ holds trivially.
    Hence assume that $\vals(s) > 0$.
    We make a further case distinction:
    \begin{itemize}
        \item $s \in \target$:
        Then $\modbellmanmin(\vals)(s) = 1$, so again we trivially have $\vals(s) \leq \modbellmanmin(\vals)(s)$.
        \item $s \notin \target$:
        By 2) we have $\vals(s) \leq \min_{a \in \act(s)} \sum_{s' \in \post(s,a)} \prmdp(s, a, s') \cdot \vals(s')$.
        Further, since $x(s) > 0$ we know by 3) that $\rank(s) < \infty$.
        By 1) and \Cref{prop:certQualReach}, $\prmin_s(\reach\target) > 0$.
        Hence, by definition of $\modbellmanmin$, $\vals(s) \leq \modbellmanmin(\vals)(s)$ is equivalent to $\vals(s) \leq \min_{a \in \act(s)} \sum_{s' \in \post(s,a)} \prmdp(s, a, s') \cdot \vals(s')$, but this holds as noted above.
    \end{itemize}\qed
\end{proof}

\subsection{$\modbellmanmax$ does not have a Unique Fixed Point}
\label{app:remarkbellmanmax}

The following modified Bellman operator does \emph{not} have a unique fixed point:
\begin{align*}
    \modbellmanmax \colon
    [0,1]^\states \to [0,1]^\states,
    ~
    \modbellmanmax(\vals)(s)
    =
    \begin{cases}
        \bellmanmax(\vals)(s) & \text{ if } \prmax_s(\lozenge\target) > 0 \\
        0 & \text{ if } \prmax_s(\lozenge\target) = 0 
    \end{cases}
\end{align*}

As a counter example, consider the following MDP:
\begin{center}
    \begin{tikzpicture}[initial text=, node distance=8mm and 30mm, on grid, thick]
    \node[state, initial] (s0) {$s_0$};
    \node[state,right=of s0] (s1) {$s_1$};
    \node[state,accepting,left=of s0] (s2) {$s_2$};
    \draw[->] (s0) edge[loop above] node[auto] {$\alpha$} (s0);
    \draw[->] (s1) edge[loop right] node[auto] {$\alpha$} (s1);
    \draw[->] (s2) edge[loop left] node[auto] {$\alpha$} (s2);
    \draw[->] (s0) edge node[auto] {$\beta, \frac 1 2$} (s1);
    \draw[->] (s0) edge node[auto] {$\beta, \frac 1 2$} (s2);
    \end{tikzpicture}
\end{center}
Every probability vector $\vals$ with $\vals(s_0) \in [\frac 1 2, 1]$, $\vals(s_1) = 0$, $\vals(s_2) = 1$ is a fixed point of $\modbellmanmax$.

\subsection{Proof of \Cref{thm:certlowerstratwitness}}
\label{app:certlowerstratwitness}

\certlowerstratwitness*
\begin{proof}
    Follows by applying \Cref{thm:certlower} to the induced DTMC $\mdp^\strat$ and using the trivial fact that $\prind{\strat}_s(\reach\target) \leq \prmax_s(\reach\target)$ for all $s \in \states$.\qed
\end{proof}

\subsection{Proof of \Cref{thm:certlower-nostrat}}
\label{app:certlower-nostrat}

\certlowernostrat*
\begin{proof}
    Let $\strat$ be any MD strategy satisfying the following:
    For all $s\in\states$,
    \begin{align}
        \label{eq:maxstrat}
        \strat(s)
        ~\in~
        \argmin_{a\in \indact{\vals}(s)} ~ \min_{s' \in \post(s,a)} \rank(s')
        \tag{$\dagger$}
        ~.
    \end{align}
    Such a $\strat$ exists as $\indact{\vals}(s) \neq \emptyset$ for all $s \in \states$ due to the assumption that $\vals \leq \bellmanmax(\vals)$.
    Note that $\strat$ depends on $\vals$.
    Intuitively, $\strat$ picks an inductive action that brings it closer to the target $\target$ with positive probability (if this is possible).
    We now show that under this strategy, the conditions of \Cref{thm:certlowerstratwitness} are satisfied.
    
    The third condition of \Cref{thm:certlowerstratwitness} holds by assumption.
    Thus, it remains to show:
    \begin{enumerate}
        \item $\distop{\strat}{}(\rank) \leq \rank$, \quad and  
        \item $\vals \leq \bellman{\strat}(\vals)$.
    \end{enumerate}
    
    First condition:
    Trivial for target states $s \in \target$. 
    For all other states $s\in\states\setminus \target$ we have 
    \begin{align*}
        \distop{\strat}{}(\rank)(s) &~=~ 1 ~+~ \min_{s' \in \post(s, \strat(s))}  \rank(s') \tag{definition of $\distop{\strat}{}$ in case $s \in \states\setminus \target$}
        \\
        &~=~ 1 ~+~ 
        \min_{a \in \indact{\vals}(s)} ~ \min_{s' \in \post(s,a)} \rank(s') \tag{definition of $\strat$, see \eqref{eq:maxstrat}}
        \\
        &~=~ \distopvals{\min}{}(\rank)(s) \tag{definition of $\distopvals{\min}{}$.}
        \\
        &~\leq~ \rank(s) \tag{by assumption}
        ~.
    \end{align*}
    
    Second condition:
    Trivial for target states $s \in \target$. 
    For all other states $s\in\states\setminus \target$ we have 
    \begin{align*}
        \vals(s) &~\leq~ \sum\limits_{s' \in \states} \prmdp(s, \strat(s), s') \cdot \vals(s') \tag{$\strat(s) \in \indact{\vals}(s)$ by \eqref{eq:maxstrat}}
        \\
        &~=~ \bellman{\strat}(\vals)(s) \tag{definition of $\bellman{\strat}$}
        ~.
    \end{align*}\qed
\end{proof}

%
\section{Proofs of \Cref{sec:ExpRew}}
All proofs from \Cref{sec:ExpRew} were formally verified in \isabelle.
We include additional conventional pen-and-paper proofs in this appendix.

\subsection{Proof of \Cref{lem:infLowerGfp}}
\infLowerGfp*
\begin{proof}
\label{app:infLowerGfp}
    If for $s \in \states$ we have that $\pr_s^{\nopt}(\lozenge \target) < 1$ then $\Eopt_s(\lozenge \target) = \infty$ and $\vals(s)$ is trivially a lower bound.
    Similarly, for all states $s \in \target$, property 1) yields $\vals(s) = 0$ and thus $\vals(s)$ is a lower bound again.
    
    For all states $s \in \states\setminus\target$ with $\pr_s^{\nopt}(\lozenge \target) = 1$, we have that $\pr_s^{\strat}(\lozenge \target) = 1$ for some optimal strategy $\strat$ (optimal w.r.t.\ expected reward).
    In the induced DTMC $\mdp^{\strat}$, for each such $s$ there exists a shortest path $\rho_s$ to $\target$ of length $n_s$.
    For all states $s \in \states$, let $z_s = \max\{0, \vals(s)-\Eopt_s(\lozenge \target)\}$.
    Due to property 2), $z_s$ is finite for $s \in \states\setminus\target$ with $\pr_s^{\nopt}(\lozenge \target) = 1$.
    Note that $\vals(s) > \Eopt_s(\lozenge \target)$ if and only if $z_s > 0$.
    
    Next we show by induction on $n \in \nats_{\geq 1}$ that for all states $s$ with $n_s = n$, if $z_s > 0$, then there is a state $u_s$ reachable from $s$ with $z_{u_s} > z_s$.
    We only show the case of $\opt=\min$ and then describe the necessary modification for the $\opt=\max$ case below.
    \begin{itemize}
        \item If $n = 1$, then $s$ has a direct successor $t \in \target$. If $z_s > 0$, then $\Emin_s(\lozenge\target) < \vals(s) \leq \bellmanrmin(\vals)(s)$ yields:
        \begin{align*}
            &\rew(s) + \sum\limits_{s' \in \states} \prmdp(s, \strat(s), s') \cdot \Emin_{s'}(\lozenge \target) \\
            ~<~ &\vals(s) \\
            ~\leq~ &\rew(s) + \sum\limits_{s' \in \states} \prmdp(s, \strat(s), s') \cdot \vals(s') ~.
            \tag{$\dagger$}
        \end{align*}
        By definition, we have that $\Emin_s(\lozenge\target) + z_s = \vals(s)$ and $\vals(s') \leq \Emin_{s'}(\lozenge\target) + z_{s'}$. With this we can rewrite inequality $(\dagger)$ to
        \[
            z_s
            ~\leq
            ~\sum_{s' \in \states} \prmdp(s, \strat(s), s') \cdot z_{s'} ~.
        \]
        Since $z_t = 0$ and $\prmdp(s, \strat(s), t) > 0$, we find that
        \[
            z_s
            ~\leq~
            \prmdp(s, \strat(s), t) \cdot 0 + \sum\limits_{s' \in \post(s,a)\setminus\{t\}} \prmdp(s, \strat(s), s') \cdot z_{s'}
            ~.
        \]
        Assume $z_{s'} \leq z_s$ for all successors $s'$. Since $\sum\limits_{s' \in \post(s,a)\setminus\{t\}} \prmdp(s, \strat(s), s') < 1$ this yields $z_s \leq r \cdot z_s$ for some $r < 1$ which is a contradiction. Thus, there exists a from $s$ reachable state $u_s$ (here even a direct successor) with $z_{s'} > z_s$
        \item If $n > 1$ and if $z_s > 0$, then, similar as before, $\Emin_s(\lozenge\target) < \vals(s) \leq \bellmanrmin(\vals)$ yields:
        \[
            z_s \leq \sum\limits_{s' \in \states} \prmdp(s, \strat(s), s') \cdot z_{s'}.
        \]
        There is a successor $s'$ of $s$ that is on the path $\rho_s$, from which by induction hypothesis a state $u_{s'}$ is reachable with $z_{u_{s'}} > z_{s'}$. Since $u_{s'}$ is also reachable from $s$, we are done if $z_{s'} \geq z_s$. Otherwise, $z_{s'} < z_s$.
        Then, however, assuming $z_{s'} \leq z_s$ for all successors $s'$ yields $z_s \leq r \cdot z_s$ for some $r < 1$ again, which is a contradiciton.
        Thus, there exists a successor $s'$ of $s$ with $z_{s'} > z_s$.  
    \end{itemize}

	For $\opt=\max$, this induction works analogously by replacing $\min$ with $\max$, except for one detail that requires quite some technical work:
	Inequation $\dagger$ does not hold in general, since an optimal strategy for the expected reward need not be optimal for the estimates $x$.
	We noted and fixed this problem when formalizing the proof in \isabelle.
	The solution is to prove that \cref{lem:infLowerGfp} works for every strategy $\sigma$, so in particular also for an optimal strategy for the expected reward.
	However, this technical modification is not very interesting, so in the interest of space, here we do not write it out formally and thereby avoid essentially duplicating the induction proof. We refer the curious reader to the \isabelle theory file, \texttt{Lemma Certificates\_Rewards.lb\_R\_supI} (see the Data Availability Statement at the beginning of the paper for a link to the artifact containing the theory file).
	
	Finally, towards a contradiction, assume that for some state $s \in \states\setminus\target$ with $\pr_s^{\nopt}(\lozenge \target) = 1$, it holds that $x(s) > \Eopt_s(\lozenge\target)$, i.e., $z_s > 0$.
    We can assume w.l.o.g.\ that $s$ is a state with this property where $z_s$ is maximal.
    But we have just shown above that there is a state $u_s$ reachable from $s$ with $z_{u_s} > z_s$, contradiction.\qed
\end{proof}

\subsection{Proof of \Cref{prop:certsInfLower}}
\certslowerexprewinf*
\begin{proof}
    \label{app:certsInfLower}
    Let $(\vals, \rank)$ be valid.
    By condition 2) and \Cref{lem:infLowerGfp} it suffices to show that $\vals$ has the following property:
    For all $s \in \states$ with $\pr_s^{\nopt}(\lozenge \target) = 1$, we have $\vals(s) < \infty$.
    Hence let $s \in \states$ be such that $\pr_s^{\nopt}(\lozenge \target) = 1$ and assume towards contradiction that $\vals(s) = \infty$.
    Then, by condition 3), it follows that $\rank(s) < \infty$.
    By condition 1) and Knaster-Tarski (\Cref{thm:knastertarski}), $\lfp{\distopmod{\nopt}} \leq \rank$.
    Hence $(\lfp{\distopmod{\nopt}})(s) < \infty$ and thus, by \Cref{lem:distopmodequiv}, $\pr_s^{\nopt}(\lozenge \target) < 1$, contradiction.\qed
\end{proof}

\subsection{Proof of \Cref{lem:infMaxUpperLfp}}
\infMaxUpperLfp*
\begin{proof}
    \label{app:infMaxUpperLfp}
    For all states $s \in \states$ with $\prmin_s(\lozenge\target) < 1$ there exists a strategy $\strat$ with $\prind{\strat}_s(\square\neg\target) > 0$ and thus the maximized expected reward of $s$ is $\infty$. States $s$ with $\prmin_s(\lozenge\target) = 0$ are immediately set to the correct value by the modified Bellman operator.
    For all states $s$ with $0 < \prmin_s(\lozenge\target) < 1$, there exists a strategy $\strat$ and an $n \in \nats$ such that $\prind{\strat}_s(\lozenge^{\leq n}\{s'\}) > 0$ for some $s'$ with $\prmin_{s'}(\lozenge\target) = 0$. By an induction, it follows that $(\lfp{\modbellmanrmax})(s) = (\tilde{\mathcal{E}}^{\max})^{n}(0)(s) = \infty$ is correct too.
    
    For all states $s$ with $\prmin_s(\lozenge\target) = 1$, we show by induction on $n \in \nats$ that $(\tilde{\mathcal{E}}^{\max})^{n}(0)(s) = \Emax_s(\lozenge^{= n} \target)$.
    \begin{itemize}
        \item $n = 0$. $(\tilde{\mathcal{E}}^{\max})^{n}(0)(s) = 0 = \Emax_s(\lozenge^{= n} \target)$
        \item $n > 0$. For all states $s \in \target$, it holds that $(\tilde{\mathcal{E}}^{\max})^{n}(0)(s) = 0 = \Emax_s(\lozenge^{= n} \target)$.
        For all states $s \in \states\setminus\target$, we have that
        \begin{align*}
            (\tilde{\mathcal{E}}^{\max})^{n}(0)(s) & = \rew(s) + \max\limits_{a \in \act(s)} \sum\limits_{s' \in \states} \prmdp(s, a, s') \cdot (\tilde{\mathcal{E}}^{\max})^{n-1}(0)(s')\\
             & \overset{I.H.}{=} \rew(s) + \max\limits_{a \in \act(s)} \sum\limits_{s' \in \states} \prmdp(s, a, s') \cdot \Emax_{s'}(\lozenge^{= n-1} \target)\\
             & = \Emax_s(\lozenge^{= n} \target)
        \end{align*}
    \end{itemize}
    Due to $\prmin_s(\lozenge\target) = 1$, it follows that $\Emax_s(\lozenge \target) = \sup_{n \in \nats}\Emax_s(\lozenge^{= n} \target)$, which concludes the theorem.\qed
    
\end{proof}

\subsection{Proof of \Cref{prop:certsInfMaxUpper}}
\certsInfMaxUpper*
\begin{proof}
    \label{app:certsInfMaxUpper}
    Let $(\vals, \rank)$ be valid.
    With Knaster-Tarski (\Cref{thm:knastertarski}) and \Cref{lem:infMaxUpperLfp}, it suffices to show that $\modbellmanrmax(\vals) \leq \vals$.
    To this end let $s \in \states$ be arbitrary.
    We make a case distinction:
    \begin{itemize}
        \item Case $\prmin_s(\lozenge\target) > 0$: Then by definition of $\modbellmanrmax$ and property 2), $\modbellmanrmax(\vals)(s) = \bellmanrmax(\vals)(s) \leq \vals(s)$.
        \item Case $\prmin_s(\lozenge\target) = 0$: In this case $\modbellmanrmax(\vals)(s) = \infty$ by definition, i.e., we have to show that $\vals(s) = \infty$.
        Assume towards contradiction that $\vals(s) < \infty$.
        Then by property 3), $\rank(s) < \infty$.
        By property 1) and \Cref{prop:certQualReach}, we conclude that $\prmin_s(\lozenge\target) > 0$, contradiction.
    \end{itemize}\qed
\end{proof}

\subsection{Remark on $\modbellmanrmin$}
\label{app:remarkbellmanminrewards}
A modified Bellman operator analogously defined as in the $\max$-expected rewards case is:
\begin{align*}
    \modbellmanrmin \colon
    ~
    \RgeqZeroInfStates \to \RgeqZeroInfStates,
    ~
    \modbellmanrmin(\vals)(s) \mapsto
    \begin{cases}
    \bellmanrmin(\vals)(s) & \text{ if }\prmax_s(\lozenge\target) > 0 \\
    \infty & \text{ if } \prmax_s(\lozenge\target) = 0
    \end{cases}
\end{align*}

However, $\modbellmanrmin$ does not work for minimizing, i.e., the least fixed point of it is \emph{not} equal to the expected reward. A counter example is the following MDP with rewards $\rew(s_0) = 0, \rew(s_1) = 100$ and $\rew(s_2) = 0$. For the right-most state $s_0$ it holds that $\prmax_{s_0}(\lozenge\target) = 1$, which means that the least fixed point of the modified Bellman operator assigns to $s_0$ the value $0$. However, it holds that $\Emin_{s_0}(\lozenge\target) = 100$ (notice that $\rew((s_0)^\omega) = \infty$ due to the \enquote{$\starinf$} semantics).
    \begin{center}
        \begin{tikzpicture}[initial text=, node distance=8mm and 30mm, on grid, thick]
        \node[state] (s1) {$s_1$};
        \node[state,right=of s1] (s0) {$s_0$};
        \node[state,accepting,left=of s1] (s2) {$s_2$};
        \draw[->] (s2) edge[loop left] node[auto] {$\alpha$} (s2);
        \draw[->] (s0) edge[loop right] node[auto] {$\alpha$} (s0);
        \draw[->] (s0) edge node[auto] {$\beta$} (s1);
        \draw[->] (s1) edge node[auto] {$\alpha$} (s2);
        \end{tikzpicture}
    \end{center}

\subsection{Certificates for Upper Bounds on $\Emin(\reach\target)$ with a Witness Strategy}
\label{app:EminUpperWithWitness}

We write $\bellmanrind{\strat}$ for the Bellman operator $\bellmanrmin = \bellmanrmax$ in the DTMC induced by strategy $\strat$.
Similarly, $\modbellmanrind{\strat}$ denotes the modified operator $\modbellmanrmax$ defined in \Cref{sec:certrew:uppermax}. The following is \Cref{lem:infMaxUpperLfp} applied to the induced DTMC $\mdp^\strat$:

\begin{restatable}{lemma}{infMinUpperLfp}
    \label{lem:infMinUpperLfp}
    For all $s \in \states$, $(\lfp{\modbellmanrind{\strat}})(s) = \mathbb{E}^{\strat}_s(\lozenge \target)$.
\end{restatable}


\begin{restatable}[Certificates for Upper Bounds on $\Emin(\lozenge\target)$ \!+\! Strategy]{proposition}{certsInfMinUpper}%
    \label{prop:certsInfMinUpper}%
    Let $\mdptuple$ be an MDP and
    $\target \subseteq \states$ a target set.
    A triple $(\vals, \rank, \strat) \in \RgeqZeroInfStates \times \exnats^\states \times \act^\states$ is called a \emph{valid certificate for upper bounds on minimal expected rewards with witness strategy} if it satisfies
    \begin{enumerate}[1)]
        \item $\distop{\strat}{}(\rank) \leq \rank$,  
        \item $\bellmanr{\strat}(\vals) \leq \vals$, 
        \item $\forall s \in \states \colon \vals(s) < \infty \implies \rank(s) < \infty$.
    \end{enumerate}
    If $(\vals, \rank, \strat)$ is valid, then $\forall s \in \states \colon \Emin_s(\reach\target) \leq \E^{\strat}_s(\reach\target) \leq \vals(s)$.
\end{restatable}
\begin{proof}
    Follows from applying \Cref{prop:certsInfMaxUpper} to the induced DTMC $\mdp^\strat$ and noticing that $\Emin_s(\lozenge \target) \leq \Estrat_s(\lozenge \target)$ for all $s \in \states$.\qed
\end{proof}


\subsection{Proof of \Cref{prop:certsInfMinUpperNoWitness}}
\certsInfMinUpperNoWitness*

\begin{proof}
    \label{app:certsInfMinUpperNoWitness}
    Let $(\vals,\rank)$ be valid and let $\strat$ be a strategy satisfying the following:
    For all $s\in\states$,
    \begin{align}
        \label{eq:maxstratrewinf}
        \strat(s)
        ~\in~
        \argmin_{a\in \coindactrew{\vals}(s)} ~ \min_{s' \in \post(s,a)} \rank(s')
        \tag{$\dagger$}
        ~.
    \end{align}
    Such a $\strat$ exists as $\coindactrew{\vals}(s) \neq \emptyset$ for all $s \in \states$ since $\bellmanrmin(\vals) \leq \vals$ (condition 2) in in \Cref{prop:certsInfMinUpperNoWitness}).
    Note that $\strat$ depends on $\vals$.
    We now show that $(\vals,\rank,\strat)$ is a valid certificate in the sense of \Cref{prop:certsInfMinUpper}.
    
    Condition 3) from \Cref{prop:certsInfMinUpper} holds by assumption.
    Thus, it remains to show:
    \begin{enumerate}[1)]
        \item $\distop{\strat}{}(\rank) \leq \rank$:
        For $s \in \target$ we have $\distop{\strat}{}(\rank)(s) = 0 \leq \rank(s)$.
        For $s\in\states\setminus \target$:
        \begin{align*}
            \distop{\strat}{}(\rank)(s) &~=~ 1 ~+~ \min_{s' \in \post(s, \strat(s))}  \rank(s') \tag{definition of $\distop{\strat}{}$}
            \\
            &~=~ 1 ~+~ 
            \min_{a \in \coindactrew{\vals}(s)} ~ \min_{s' \in \post(s,a)} \rank(s') \tag{definition of $\strat$, see \eqref{eq:maxstratrewinf}}
            \\
            &~=~ \distoprvals{\min}{}(\rank)(s) \tag{definition of $\distoprvals{\min}{}$.}
            \\
            &~\leq~ \rank(s) \tag{by condition 1) in \Cref{prop:certsInfMinUpperNoWitness}}
            ~.
        \end{align*}
        \item $\bellmanr{\strat}(\vals) \leq \vals$:
        For $s \in \target$ we have $\bellmanr{\strat}(\vals)(s) = 0 \leq \vals(s)$.
        For $s\in\states\setminus \target$:
        \begin{align*}
            \vals(s)& ~\geq~ \rew(s) + \sum\limits_{s' \in \states} \prmdp(s, \strat(s), s') \cdot \vals(s') \tag{$\strat(s) \in \coindactrew{\vals}(s)$}
            \\
            &~=~ \bellmanr{\strat}(\vals)(s) \tag{definition of $\bellman{\strat}$}
            ~.
        \end{align*}
    \end{enumerate}\qed
\end{proof}

\section{Certificates for Expected Rewards with Alternative Semantics}
\label{app:sec:exprewStarRho}
We now present certificates for bounds on expected rewards with the other semantics mentioned in \Cref{sec:prelims}, called \enquote{$\starrho$} semantics in the following.
This semantics assigns the accumulated reward to paths not reaching $\target$.
The accumulated reward can be infinite or finite depending on whether the path reaches a cycle with a positive reward or where all states have a reward of $0$. We can characterize expected rewards with the \enquote{$\starrho$} semantics using the same Bellman operator $\bellmanropt$ as the one associated to the \enquote{$\starinf$} semantics (see \Cref{defn:bellmanopRewards}). The certificates show similarities to the certificates for reachability (\Cref{sec:certificates}) in the sense that computing expected rewards with the \enquote{$\starrho$} semantics is also a least fixed point objective:
\begin{theorem}
\label{thm:rhoUpperLfp}
For all $s \in \states$, $(\lfp{\bellmanropt})(s) = \Eopt_s(\lozenge \target).$

\end{theorem}
\mw{citation. Ensuring reliability paper?}
\begin{proof}
    We show by induction on $n \in \nats$ that $(\bellmanropt)^n(0)(s) = \Eopt_s(\lozenge^{= n} \target)$.
    \begin{itemize}
        \item $n = 0$. $(\bellmanropt)^n(0)(s) = 0 = \Eopt_s(\lozenge^{= n} \target)$
        \item $n > 0$. For all states $s \in \target$, it holds that $(\bellmanropt)^n(0)(s) = 0 = \Eopt_s(\lozenge^{= n} \target)$.
        For all states $s \in \states\setminus\target$, we have that
        \begin{align*}
            (\bellmanropt)^n(0)(s) & = \rew(s) + \underset{a \in \act(s)}{\opt} \sum\limits_{s' \in \states} \prmdp(s, a, s') \cdot (\bellmanropt)^{n-1}(0)(s')\\
             & \overset{I.H.}{=} \rew(s) + \underset{a \in \act(s)}{\opt} \sum\limits_{s' \in \states} \prmdp(s, a, s') \cdot \Eopt_{s'}(\lozenge^{= n-1} \target)\\
             & = \Eopt_s(\lozenge^{= n} \target)
        \end{align*}
    \end{itemize}
    Since $\Eopt_s(\lozenge \target) = \sup_{n \in \nats}\Eopt_s(\lozenge^{= n} \target)$, the theorem follows.\qed
\end{proof}
This leads to a simple mechanism to certify upper bounds.

\subsection{Upper Bounds on Optimal Expected Rewards}
The following is an immediate consequence of \Cref{thm:knastertarski} (Knaster-Tarski) and \Cref{thm:rhoUpperLfp}:
\begin{proposition}[Certificates for Upper Bounds on $\Eopt(\lozenge\target)$]
\label{prop:certsRhoUpper}
A value vector $\vals \in \RgeqZeroInfStates$ satisfying $\bellmanropt(\vals) \leq \vals$ is called a \emph{valid certificate for upper bounds on $\opt$-expected rewards with \enquote{$\starrho$} semantics}. If $\vals$ is valid, then $\forall s \in \states \colon \Eopt_s(\lozenge \target) \leq \vals(s).$
\end{proposition}

\subsection{Lower Bounds on Minimal Expected Rewards}
\label{sec:rhoLowerBoundsMin}
For the certification of lower bounds on the minimal expected rewards, our goal is to modify the Bellman operator such that its greatest fixed point is equal to the minimal expected rewards of the MDP.
The function $\bellmanrmintarget{}$ has two problems in this regard.
The first problem is that Value Iteration from above yields the wrong results e.g. for states in end components in which all states have a reward of $0$. When a path reaches such a state, it will never collect any additional reward in the future. This property makes such states behave exactly like target states, which motivates to redefine the set of target states to
\[
    \newtargetmin \coloneqq \{s ~\mid~ \Emin_s(\lozenge\target) = 0\}.
\]
Note, that $\target \subseteq \newtargetmin$ holds as well as $\Emin_s(\lozenge \target) = \Emin_s(\lozenge \newtargetmin)$ for all $s \in \states$. We equip our Bellman and distance operators $\bellmanropt$ and $\distopmod{\opt}$ with a subscript, $\bellmanropttarget{\target}$ and $\distopmodtarget{\opt}{\target}$ respectively, that indicates which target set is considered. An important property of the Bellman operators regarding multiple target sets is that the output can only be larger, if a smaller target set is considered, i.e. for $\newtarget \subseteq \target$ and for all $\vals \in \RgeqZeroInfStates$, we have that $\bellmanropttarget{\target}(\vals) \leq \bellmanropttarget{\newtarget}(\vals)$ and $\bellmanrindtarget{\strat}{\target}(\vals) \leq \bellmanrindtarget{\strat}{\newtarget}(x)$. Now we can consider the Bellman operator $\modbellmanrrhomin$ that takes the redefined target set into account and by that the first problem is already solved. The second problem that still remains is that the spurious greatest fixed point (assigning $\infty$ to all non-target states and $0$ to all target states) appears again (as it already did for \emph{lower} bounds in the \enquote{$\starinf$} case). We deal with it in a similar manner as in the previous case while respecting the new target set: by forcing all values of states reaching the target with probability one to be finite.

\begin{lemma}
\label{lem:rhoMinLowerGfp}
    Let $\vals \in \RgeqZeroInfStates$ be such that 1) $\vals \leq \bellmanrmintarget{\newtargetmin}(\vals)$ and 2) $\prmax_s(\lozenge\newtargetmin) = 1 \implies \vals(s) < \infty$ for all $s \in \states$. Then it holds for all $s \in \states$ that $\vals(s) \leq \Emin_s(\lozenge \target)$.
\end{lemma}
\begin{proof}
    By definition, we have that $\Emin_s(\lozenge \target) = \Emin_s(\lozenge \newtargetmin)$.
    Then, \Cref{lem:infLowerGfp} is an even more general statement. Here, we have that $\opt = \min$. Note that even though \Cref{lem:infLowerGfp} considers the $\starinf$ case, the proof is still applicable because our redefined target set ensures that for all states $s$ with $\prmax_s(\lozenge\newtargetmin) < 1$ it holds that $\Emin_s(\lozenge \newtargetmin) = \infty$.\qed
\end{proof}

We reuse our certificates from \Cref{sec:qualreachandsafe} for \emph{non-almost sure reachability} again to find all states $s$ satisfying $\prmax_s(\lozenge\newtargetmin) < 1$. However, the redefined target set brings up another complication: The certificate checker needs to be provided with our redefined target set $\newtargetmin$ since computing it from the scratch would significantly increase the runtime. Then however, the certificate checker also needs to certify that the provided redefined target set is correct.


The certificate checker is provided with the input target set $\newtargetmininput$ as part of the certificate. Certifying that $\newtargetmininput = \newtargetmin$ has a similar expense as computing $\newtargetmin$ from scratch, but we make the following observation: for the certification of lower bounds, it suffices to have that $\newtargetmin \subseteq \newtargetmininput$. This can be checked efficiently, as we know from \Cref{sec:qualreachandsafe} how to certify qualitative reachability and we have the following equivalence: for all states $s \in \states$, we have that $s \in \newtargetmin$ is equivalent to $\prmin_s(\lozenge \pos) = 0$, where $\pos = \{s ~\mid~ \rew(s) > 0\}$. Thus, the task of certifying the redefined target set became a task of certifying some qualitative reachability property, which is done by bullets 4 and 5 in the following proposition.

\begin{proposition}[Certificates for Lower Bounds on $\Emin(\lozenge\target)$]
\label{prop:certsRhoMinLower}
    Let $\mdptuple$ be an MDP with redefined target set $\newtargetmin \subseteq \states$ and $\pos = \{s ~\mid~ \rew(s) > 0\}$. 
A quadruple $(\vals, \rank_1, \rank_2, \newtargetmininput) \in \RgeqZeroInfStates \times \exnats^\states \times \exnats^\states \times 2^{\states}$ is called a \emph{valid certificate for lower bounds on minimal expected rewards with \enquote{$\starrho$} semantics} if it satisfies
    \begin{enumerate}
        \item $\distopmodtarget{\max}{\newtargetmininput}(\rank_1) \leq \rank_1$,
        \item $\vals \leq \bellmanrmintarget{\newtargetmininput}(\vals)$,
        \item $\forall s \in \states \colon \rank_1(s) = \infty \implies \vals(s) < \infty$,
        \item $\distop{\max}{\pos}(\rank_2) \leq \rank_2$,
        \item $\forall s \in \states \colon \rank_2(s) = \infty \implies s \in \newtargetmininput$.
    \end{enumerate}
    If $(\vals, \rank_1, \rank_2, \newtargetmininput)$ is valid, then $\forall s \in \states \colon \Emin_s(\lozenge \target) \geq \vals(s)$.
\end{proposition}
\begin{proof}
    Bullets 4 and 5 together with Knaster-Tarski verify that $\newtargetmin \subseteq \newtargetmininput$ because:
    \begin{align*}
        s \in \newtargetmin & \implies \prmin_s(\lozenge \pos) = 0\\
         & \overset{\Cref{lem:distopequiv}}{\implies} \fp{\distop{\max}{\pos}}(s) = \infty\\
         & \overset{bullet~ 4}{\implies} \rank_2(s) = \infty\\
         & \overset{bullet~ 5}{\implies} s \in \newtargetmininput.
    \end{align*}
    
    Let $s \in \states$. If $\prmax_s(\lozenge \newtargetmin) < 1$, then $\Emin_s(\lozenge \target) = \infty$ and $\vals(s)$ surely is a lower bound. Otherwise, bullet 1 and Knaster-Tarski yield that $\fp{\distopmodtarget{\max}{\newtargetmininput}} \leq \rank_1$. Since $\newtargetmin \subseteq \newtargetmininput$, we also have that $\fp{\distopmodtarget{\max}{\newtargetmin}} \leq \rank_1$. Bullets 1 and 3 further yield for all $s \in \states$:
    \begin{align*}
    \prmax_s(\lozenge \newtargetmin) = 1 & \overset{\Cref{lem:distopmodequiv}}{\implies} \fp{\distopmodtarget{\max}{\newtargetmin}} = \infty\\
     & \overset{bullet~ 1}{\implies} \rank_1(s) = \infty\\
     & \overset{bullet~ 3}{\implies} \vals(s) < \infty.
    \end{align*}

    Bullet 2 implies that $\vals \leq \bellmanrmintarget{\newtargetmin}(\vals)$. Then, \Cref{lem:rhoMinLowerGfp} is applied and yields the proposition.\qed
\end{proof}


\subsection{Lower Bounds on Maximal Expected Rewards}
For lower bounds on maximal expected rewards, we again redefine the target set. First, consider
\[
    \newtargetmax \coloneqq \{s ~\mid~ \Emax_s(\lozenge\target) = 0\}.
\]
Similar as in the previous section, we have that $\target \subseteq \newtargetmax$ and $\Emax_s(\lozenge \target) = \Emax_s(\lozenge \newtargetmax)$. Examples for states that are in $\newtargetmax$, but that might not have been in $\target$, are states in bottom strongly connected components, in which all states have a reward of $0$. However, the Bellman operator $\bellmanrmaxtarget{\newtargetmax}$ that takes the new target set into account is not useful for the certification. The greatest fixed point of $\bellmanrmaxtarget{\newtargetmax}$, even taken aside the possible spurious greatest fixed point, does not yield the maximized expected reward values. This is different from the minimizing case, but analogue to our certificates for reachability, where a modification of the Bellman operator also did not work (with our techniques) for maximizing. We deal with the situation by including an explicit witness strategy in the certificate. Then, we again encode the strategy into the value- and rank vector.

\subsubsection{Certificates including an Explicit Witness Strategy}
Reconsider the ``Bellman operator'' $\bellmanrind{\strat}$ of the Markov chain induced by an MD strategy $\strat$. The target states of the induced Markov chain are redefined to
\[
    \newtargetstrat \coloneqq \{s ~\mid~ \Estrat_s(\lozenge\target) = 0\}.
\]
The Bellman operator $\bellmanrindtarget{\strat}{\newtargetstratinput}$ that takes the new target set into account will be useful for the certification, after adressing the issue with the spurious greatest fixed point again.
\begin{lemma}
\label{lem:rhoMaxLowerGfp}
    Let $\vals \in \RgeqZeroInfStates$ be such that 1) $\vals \leq \bellmanrindtarget{\strat}{\newtargetstrat}(\vals)$ and 2) $\prind{\strat}_s(\lozenge\newtargetstrat) = 1 \implies \vals(s) < \infty$ for all $s \in \states$. Then it holds for all $s \in \states$ that $\vals(s) \leq \Estrat_s(\lozenge \target)$.
\end{lemma}
\begin{proof}
    By definition, we have that $\Estrat_s(\lozenge \target) = \Estrat_s(\lozenge \newtargetstrat)$.
    Then, \Cref{lem:infLowerGfp} is an even more general proof, because we can take the already given strategy $\strat$ as the strategy that is considered in the proof. Note that even though \Cref{lem:infLowerGfp} considers the $\starinf$ case, the proof is still applicable because our redefined target set ensures that for all states $s$ with $\prind{\strat}_s(\lozenge\newtargetstrat) < 1$ it holds that $\Estrat_s(\lozenge \newtargetstrat) = \infty$.\qed
\end{proof}

We reuse our certificates from $\Cref{sec:qualreachandsafe}$ for \emph{non-almost sure reachability} again to find all states $s$ satisfying $\prind{\strat}_s(\lozenge\newtargetstrat) < 1$. Additionally, we have the same need for certifying the redefined target set as in the previous section.


\begin{proposition}[Certificates for Lower Bounds on $\Emax(\lozenge\target)$ + Witness Strategy]
\label{prop:certsRhoMaxLower}
    Let $\mdptuple$ be an MDP with redefined target set $\newtargetmin \subseteq \states$ and $\pos = \{s ~\mid~ \rew(s) > 0\}$. 
A quintuple $(\vals, \rank_1, \rank_2, \strat, \newtargetstratinput) \in \RgeqZeroInfStates \times \exnats^\states \times \exnats^\states \times \act^{\states} \times 2^{\states}$ is called a \emph{valid certificate for lower bounds on minimal expected rewards with \enquote{$\starrho$} semantics including a witness strategy} if it satisfies
    \begin{enumerate}
        \item $\distopmodtarget{\strat}{\newtargetstratinput}(\rank_1) \leq \rank_1$,
        \item $\vals \leq \bellmanrindtarget{\strat}{\newtargetstratinput}(\vals)$,
        \item $\forall s \in \states \colon \rank_1(s) = \infty \implies \vals(s) < \infty$,
        \item $\distop{\strat}{\pos}(\rank_2) \leq \rank_2$,
        \item $\forall s \in \states \colon \rank_2(s) = \infty \implies s \in \newtargetstratinput$.
    \end{enumerate}
    If $(\vals, \rank_1, \rank_2, \strat, \newtargetstratinput)$ is valid, then $\forall s \in \states \colon \Emax_s(\lozenge\target) \geq \Estrat_s(\lozenge\target) \geq \vals(s)$.
\end{proposition}
\begin{proof}
    Bullets 4 and 5 together with Knaster-Tarski verify that $\newtargetstrat \subseteq \newtargetstratinput$ because:
    \begin{align*}
        s \in \newtargetstrat & \implies \prind{\strat}_s(\lozenge \pos) = 0\\
         & \overset{\Cref{lem:distopequiv}}{\implies} \fp{\distop{\strat}{\pos}}(s) = \infty\\
         & \overset{bullet~ 4}{\implies} \rank_2(s) = \infty\\
         & \overset{bullet~ 5}{\implies} s \in \newtargetstratinput.
    \end{align*}
    
    Let $s \in \states$. If $\prind{\strat}_s(\lozenge \newtargetstrat) < 1$ then $\Estrat_s(\lozenge \target) = \Emax_s(\lozenge \target) = \infty$ and $\vals(s)$ surely is a lower bound. Otherwise, bullet 1 and Knaster-Tarski yield that $\fp{\distopmodtarget{\strat}{\newtargetstratinput}} \leq \rank_1$. Since $\newtargetstrat \subseteq \newtargetstratinput$, we also have that $\fp{\distopmodtarget{\strat}{\newtargetstrat}} \leq \rank_1$. Bullets 1 and 3 further yield for all $s \in \states$:
    \begin{align*}
    \prind{\strat}_s(\lozenge \newtargetstrat) = 1 & \overset{\Cref{lem:distopmodequiv}}{\implies} \fp{\distopmodtarget{\strat}{\newtargetstrat}}(s) = \infty\\
     & \overset{bullet~ 1}{\implies} 
    \rank_1(s) = \infty\\
     & \overset{bullet~ 3}{\implies} \vals(s) < \infty.
    \end{align*}
    Bullet 2 implies that $\vals \leq \bellmanrindtarget{\strat}{\newtargetstrat}(\vals)$. Then, \Cref{lem:rhoMaxLowerGfp} is applied and since for all $s \in \states \colon \Estrat_s(\lozenge \target) \leq \Emax_s(\lozenge \target)$ the proposition follows.\qed
\end{proof}


\subsubsection{Certificates without Witness Strategies}
In this section, we encode the strategy of the certificate in $\vals$ as well as the two ranking functions. This is practical since we anyway need those for the certificate, and then the strategy is no longer part of the certificate. For every $s \in \states$ and value vector $\vals$, we define the \emph{$\vals$-increasing} actions with respect to the reward available at $s$ as follows:
\begin{align*}
    \indactrew{\vals}(s)
    ~=~
    \{a\in\act(s) \mid \vals(s) \leq \rew(s) + \sum\limits_{s' \in \states} \prmdp(s, a, s') \cdot \vals(s') \}
    ~.
\end{align*}
Note that $\indactrew{\vals}(s)$ may be empty in general.
However, if $\vals \leq \bellmanrmaxtarget{\newtargetmax}(\vals)$, then $\indactrew{\vals}(s)$ contains at least one action. Next, we define a variant of the distance operator $\distoprrhovals{\opt}{\target}$ that only considers $\vals$-increasing actions.

\begin{align*}
    \distoprrhovals{\opt}{\target} \colon
    ~
    \exnats^\states \to \exnats^\states,
    ~
    \rank \mapsto \lambda s.
    \begin{cases}
        0 & \text{ if } s \in \target\\
        1 ~+~ \underset{a \in \indactrew{\vals}(s)}{\opt} ~ \min\limits_{s' \in \post(s,a)}  \rank(s') & \text{ if } s \in \states\setminus\target
    \end{cases}
\end{align*}

As opposed to the two other times where we encoded a strategy into the value- and rank vector (for reachability and for expected rewards with the \enquote{$\starinf$} semantics), this time there is not only one ranking function with respect to which the strategy has to be consistent. Thus, this time the encoding is only possible if there exists a strategy that is consistent with both ranking functions. Quadruples $(\vals, \rank_1, \rank_2, \newtargetmaxinput)$ for which there does not exist such a strategy, are immediately rejected as invalid certificates. 


\begin{proposition}[Certificates for Lower Bounds on $\Emax(\lozenge\target)$ - without Witness Strategy]
\label{prop:certsRhoMaxLowerNoWitness}
    Let $\mdptuple$ be an MDP with redefined target set $\newtargetmax \subseteq \states$ and $\pos = \{s ~\mid~ \rew(s) > 0\}$. A quadruple $(\vals, \rank_1, \rank_2, \newtargetmaxinput) \in \RgeqZeroInfStates \times \exnats^\states \times \exnats^\states \times 2^{\states}$ is called a \emph{valid certificate for lower bounds on maximal expected rewards with \enquote{$\starrho$} semantics} if there exists a strategy $\sigma$ such that
    \begin{enumerate}
        \item $\forall s \in \states\colon \sigma(s) \in         \underset{a \in \indactrew{\vals}(s)}{\argmax} ~ \underset{s' \in \post(s,a)}{\min} \rank_1(s')$,
        \item $\forall s \in \states\colon \sigma(s) \in         \underset{a \in \indactrew{\vals}(s)}{\argmax} ~ \underset{s' \in \post(s,a)}{\min} \rank_2(s')$,
    \end{enumerate}
    and if it satisfies
    \begin{enumerate}
        \item $\distopmodtarget{\max}{\newtargetmaxinput}(\rank_1) \leq \rank_1$,
        \item $\vals \leq \bellmanrmaxtarget{\newtargetmaxinput}(\vals)$,
        \item $\forall s \in \states \colon \rank_1(s) = \infty \implies \vals(s) < \infty$,
        \item $\distoprrhovals{\max}{\pos}(\rank_2) \leq \rank_2$,
        \item $\forall s \in \states \colon \rank_2(s) = \infty \implies s \in \newtargetmaxinput$.
    \end{enumerate}
    If $(\vals, \rank_1, \rank_2, \newtargetmaxinput)$ is valid, then $\forall s \in \states \colon \Emax_s(\lozenge \target) \geq \vals(s)$.
\end{proposition}
\begin{proof}
    By assumption there exists a strategy $\strat$ satisfying
    \begin{enumerate}
        \item $\forall s \in \states\colon \sigma(s) \in         \underset{a \in \indactrew{\vals}(s)}{\argmax} ~ \underset{s' \in \post(s,a)}{\min} \rank_1(s')$
        \item $\forall s \in \states\colon \sigma(s) \in         \underset{a \in \indactrew{\vals}(s)}{\argmax} ~ \underset{s' \in \post(s,a)}{\min} \rank_2(s')$
    \end{enumerate}

    We now show that under this strategy, the conditions of \Cref{prop:certsRhoMaxLower} are satisfied. The third and fifth condition of \Cref{prop:certsRhoMaxLower} hold by assumption.
    Thus, it remains to show:
    \begin{enumerate}
        \item $\distopmodtarget{\strat}{\newtargetstratinput}(\rank_1) \leq \rank_1$
        \item $\vals \leq \bellmanrindtarget{\strat}{\newtargetstratinput}(\vals)$, \quad and
        \item $\distop{\strat}{\pos}(\rank_2) \leq \rank_2$
    \end{enumerate}

    First condition:
    Trivial for states $s \in \newtargetstratinput$. 
    For all other states $s\in\states\setminus \newtargetstratinput$ we have 
    \begin{align*}
        \distop{\strat}{\newtargetstratinput}(\rank_1)(s) &~=~ 1 ~+~ \min_{s' \in \post(s, \strat(s))}  \rank_1(s') \tag{definition of $\distop{\strat}{\newtargetstratinput}$}
        \\
        &~=~ 1 ~+~ 
        \max_{a \in \indactrew{\vals}(s)} ~ \min_{s' \in \post(s,a)} \rank_1(s') \tag{definition of $\strat$}
        \\
        &~=~ \distopmodtarget{\max}{\newtargetmaxinput}(\rank_1)(s) \tag{definition of $\distopmodtarget{\max}{\newtargetmaxinput}$.}
        \\
        &~\leq~ \rank_1(s) \tag{by assumption}
        ~.
    \end{align*}
    
    Second condition:
    Trivial for target states $s \in \newtargetstratinput$. 
    For all other states $s\in\states\setminus \newtargetstratinput$ we have 
    \begin{align*}
        \vals(s) & ~\leq~ \rew(s) + \sum\limits_{s' \in \states} \prmdp(s, \strat(s), s') \cdot \vals(s') \tag{$\strat(s) \in \indactrew{\vals}(s)$}
        \\
        &~=~ \bellmanrindtarget{\strat}{\newtargetstratinput}(\vals)(s) \tag{definition of $\bellmanrindtarget{\strat}{\newtargetstratinput}$}
        ~.
    \end{align*}  

    Third condition:
    Trivial for states $s \in \pos$. 
    For all other states $s\in\states\setminus \pos$ we have 
    \begin{align*}
        \distop{\strat}{\pos}(\rank)(s) &~=~ 1 ~+~ \min_{s' \in \post(s, \strat(s))}  \rank_2(s') \tag{definition of $\distop{\strat}{\pos}$}
        \\
        &~=~ 1 ~+~ 
        \max_{a \in \indactrew{\vals}(s)} ~ \min_{s' \in \post(s,a)} \rank_2(s') \tag{definition of $\strat$}
        \\
        &~=~ \distoprrhovals{\max}{\pos}(\rank_2)(s) \tag{definition of $\distoprrhovals{\max}{\pos}$.}
        \\
        &~\leq~ \rank_2(s) \tag{by assumption}
        ~.
    \end{align*} \qed
\end{proof}

\section{Computing Certificates}
\subsection{Further Details on Smooth Interval Iteration}\label{app:smoothII}
In \Cref{sec:computing} we mentioned that every (co-)inductive value vector w.r.t. $\bellman{\opt}_{\gamma}$ is also (co-)inductive w.r.t. $\bellman{\opt}$. This follows almost directly from the definitions of the operators since for every $\vals \in [0,1]^{\states}$:
\begin{align*}
    \vals \leq \bellman{\opt}_{\gamma}(\vals) & \Longleftrightarrow \vals \leq \gamma \cdot \vals + (1 - \gamma) \cdot \bellman{\opt}(\vals)\\
     & \Longleftrightarrow (1 - \gamma) \cdot \vals \leq (1 - \gamma) \cdot \bellman{\opt}(\vals)\\
     & \Longleftrightarrow \vals \leq \bellman{\opt}(\vals)
\end{align*}
For Value Iteration from below this means that the sequence the $\gamma$-smooth Bellman~operator creates constantly lies below the sequence that the normal Bellman~operator creates, i.e., if you let $\vals_{\gamma}^{(i)} = (\bellman{\opt}_{\gamma})^i(0)$ and $\vals^{(i)} = (\bellman{\opt})^i(0)$, then it holds for all $i \in \nats$ that $\vals_{\gamma}^{(i)} \leq \vals^{(i)}$. We prove this by an induction:
\begin{itemize}
    \item $i=0$: It holds that $0 \leq 0$.
    \item $i > 0$: For states $s \in \target$ it surely holds that
    \[
        \vals_{\gamma}^{(i)} \leq \vals^{(i)} = 1.
    \]
    For states $s \in \states \setminus \target$ the induction hypothesis yields $\vals_{\gamma}^{(i-1)} \leq \vals^{(i-1)}$ and by that
    \begin{align*}
        \vals_{\gamma}^{(i)} & = \gamma \cdot \vals_{\gamma}^{(i-1)} + (1-\gamma) \cdot \bellman{\opt}(\vals_{\gamma}^{(i-1)})\\
         & \leq \gamma \cdot \vals^{(i-1)} + (1-\gamma) \cdot \bellman{\opt}(\vals^{(i-1)}) \leq \bellman{\opt}(\vals^{(i-1)}) = \vals^{(i)}.
    \end{align*}
\end{itemize}

By the principle of induction, we conclude that $\vals_{\gamma}^{(i)} \leq \vals^{(i)}$ for all $i \in \nats$. Consequently, this also implies that $\lfp{\bellman{\opt}_{\gamma}} \leq \lfp{\bellman{\opt}}$. However, we show that the least fixed points of $\bellman{\opt}$ and $\bellman{\opt}_{\gamma}$ even coincide. To do so, it remains to prove that $\lfp{\bellman{\opt}} \leq \lfp{\bellman{\opt}_{\gamma}}$. To do so, we show in the following that $\lfp{\bellman{\opt}} \leq \bellman{\opt}_{\gamma}(\lfp{\bellman{\opt}})$ and then Knaster-Tarski yields the desired. Since
\[
	\bellman{\opt}_{\gamma}(\lfp{\bellman{\opt}}) = \gamma \cdot \lfp{\bellman{\opt}} + (1-\gamma) \cdot \bellman{\opt}(\lfp{\bellman{\opt}}) = \lfp{\bellman{\opt}},
\]
we obtain the desired.

\paragraph{A $\gamma$-Smooth Operator for Expected Rewards.}
For expected rewards, we obtain a $\gamma$-smooth operator with similar properties as
\[
	\bellmanropt_\gamma(\vals) = \rew(s) + \gamma \cdot \vals + (1-\gamma) \cdot \underset{a \in \act(s)}{\opt} \sum\limits_{s' \in \post(s,a)} \prmdp(s, a, s') \cdot \vals(s').
\]
Note that the reward is not multiplied by the factor $\gamma$.


\subsection{More Efficient Algorithms for Computing the Ranking Functions}
\label{app:rankingFunctionVIDetails}
\label{app:algsForComputing}
We start by giving a high-level description of how a algorithms computing ranking functions with the properties desired for our certificates work. Then, we give two concrete algorithms in pseudocode and prove their termination and correctness.

\paragraph{Ranking Functions for $\distop{\opt}{}$.}
Several certificates discussed in \Cref{sec:qualreachandsafe,sec:certificates,sec:ExpRew} need a ranking function $\rank$ that satisfies $\distop{\opt}{}(\rank) \leq \rank$ and is sufficiently small to witness finite distance to $\target$ from the states where this is required (see \Cref{fig:certoverview,fig:certoverviewrewards} for the exact conditions).\footnote{Notice that the restricted variants $\distopvals{\min}{}$ and $\distoprvals{\min}{}$ from \Cref{sec:certreach:lowermax,sec:InfUpperMin} are the same as $\distop{\min}{}$ in a sub-MDP.}
To ensure a sufficiently small $\rank$, one can simply compute the exact (unique) fixed point $\rank = \fp{\distop{\opt}{}}$ (see \Cref{lem:distopuniquefp}).
This can be achieved as follows:
We start VI from $\rank^{(0)} = \vec\infty$.
Since $\rank^{(0)}$ is inductive, the VI sequence converges monotonically to $\rank$ from above.
Since $\exnats^\states$ is a \emph{well partial order} (it contains no infinite strictly decreasing sequences), the sequence converges in finitely many steps.
See \Cref{app:algsForComputing} for a more efficient algorithm.

\paragraph{Ranking Functions for $\distopmod{\opt}{}$.}
For the \emph{complementary distance operator} $\distopmod{\opt}{}$ (\Cref{def:distopmod})
which is employed only in \Cref{prop:certsInfLower}, to compute $\rank = \lfp{\distopmod{\opt}{}}$ we propose to perform VI from $\rank^{(0)}$ with $\rank^{(0)}(s) = [\pr^{\nopt}_s(\reach\target) = 1] \cdot \infty$ for all $s \in \states$.
The condition in the Iverson bracket can be evaluated using standard graph analysis~\cite[Section~10.6.1]{principles}.
Notice that $\rank^{(0)} \leq \rank$ by \Cref{lem:distopmodequiv}.
It is moreover easy to see that $\rank^{(0)}$ is co-inductive.
The VI sequence thus converges monotonically to $\rank$ from below.
Since $\rank(s) \neq \infty$ for all $s \in \states$ with $\rank^{(0)}(s) \neq \infty$, the iteration will again converge in finitely many steps.\\

We now give two concrete algorithms for which we prove (1): termination and (2): that they compute the exact fixed point of the distance operator (least fixed point respectively in the case of the complementary distance operator). In the following, we write $\pre(s) = \{s' \in \states \mid \exists a \in \act(s') \colon \prmdp(s',a,s) > 0\}$ for state $s \in \states$ of an MDP $\mdptuple$.

\begin{algorithm}[t]
    \Input{Finite MDP $\mdp = \mdptuple$, target states $\target \subseteq S$}
    \Output{$\rank \in \exnats^\states$ such that $\rank = \fp{\distop{\opt}{}}$}
    \lForEach(\tcp*[h]{$Q$ is a FIFO queue}){$s \in T$}{$\rank(s) \gets 0$\,;~$Q.\mathtt{push}(s)$}
    \lForEach{$s \in \states \setminus \target$}{$\rank(s) \gets \infty$}
    \While{$Q$ is not empty}{
        $\hat{s} \gets Q.\mathtt{pop}()$\;
        \ForEach{$s \in \pre(\hat{s})$}{%
            \lIf{$\rank(s) < \infty$}{\Continue\label{alg:rankfpcompute:line:onlysetifinf}}
            $\textnormal{tmp} \gets 1 + \underset{a \in \act(s)}{\opt} ~ \underset{s' \in \post(s,a)}{\min}  \rank(s') $\label{alg:rankfpcompute:line:rs}\;
            \lIf{$\textnormal{tmp} = \rank(\hat{s}) + 1$}{%
                $\rank(s) \gets \textnormal{tmp} $\,;~$Q.\mathtt{push}(s)$\label{alg:rankfpcompute:line:enqueue}
            }%
        }%
    }
    \caption{Computation of the fixed point of the distance operator}
    \label{alg:rankfpcompute}
\end{algorithm}

\begin{restatable}[Computing the Fixed Point of the Distance Operator]{proposition}{algrankfpcompute}
    \label{prop:algrankfpcompute}
    \Cref{alg:rankfpcompute} terminates and computes $\rank \in \exnats^\states$ with $\rank = \fp{\distop{\opt}{}}$.
\end{restatable}
\begin{proof}
\label{app:algrankfpcompute}
    For $\opt = \max$. For $\opt = \min$, only the wording `for all $a \in \act(s)$' has to be adjusted to `there exists an $a \in \act(s)$'.
    
    Due to \Cref{alg:rankfpcompute:line:onlysetifinf}, a rank of a state can only change from $\infty$ to a finite value and remains constant from that point on. Furthermore, it is an invariant that for each state $s$ that is pushed into the queue, $\rank(s)$ is set to a finite value right before. Thus, every state $s \in S$ is pushed into the queue $Q$ at most once and termination follows.
    
    Let $\rank \in \exnats^\states$ hold the computed values upon termination of \Cref{alg:rankfpcompute}. We show by induction on $i \in \nats$, that the following is an invariant: if a state $s$ with $\rank(s) = i$ is popped from the queue, then all states $s'$ with $\rank(s') < i$ have already been popped from the queue. ($\dagger$)
    \begin{itemize}
        \item $i=0$. Clear, since there are no states $s'$ with $\rank(s') < 0$.
        \item $i>0$. Since $s$ is popped from the queue, $s$ was pushed into the queue when, due to \Cref{alg:rankfpcompute:line:enqueue}, there has been a state $\hat{s}$ popped from the queue with $\rank(\hat{s}) = i-1$. By induction hypothesis, all states $s''$ with $\rank(s'') < i-1$ have already been popped from the queue. Take an arbitrary state $s'$ with $\rank(s') = i-1$. Since all its successors $s''$ with $\rank(s'') < i-1$ have already been popped from the queue, $s'$ got pushed into the queue before $s$ and thus also popped before.
    \end{itemize}

    Now we show, based on the invariant $(\dagger)$, that we have for all $s \in S$ that $\rank(s) = \fp{\distop{\max}{}}(s)$. We show by induction on $i \in \nats$ that
    \[
        \rank(s) = i \quad \Longleftrightarrow \quad \fp{\distop{\max}{}}(s) = i.
    \]
    \begin{itemize}
        \item $i = 0$. Then, $\rank(s)=0$ and $\fp{\distop{\max}{}}(s)=0$ are both equivalent to $s \in \target$.
        \item $i > 0$. We show both directions. If $\fp{\distop{\max,\min}{\target}}(s)=i$, then for all $a \in \act(s)$ there is an $s' \in \post(s,a)$ such that $\fp{\distop{\max}{}}(s')=i-1$. By induction hypothesis, $\rank(s') = i-1$ for all such $s'$. Take the first such $s'$ that is popped from the queue. Certainly, $s$ is a predecessor of $s'$. Let $\hat{\rank} \in \exnats^\states$ hold the values of $\rank$ right before \Cref{alg:rankfpcompute:line:onlysetifinf} is executed. Due to $(\dagger)$, we have that still $\hat{\rank}(s) = \infty$ and thus \Cref{alg:rankfpcompute:line:rs} is executed. Also due to $(\dagger)$, we have that $\hat{\rank}(s')=i-1$ is already set for all states $s'$ with $\fp{\distop{\max,\min}{\target}}(s')=i-1$. Thus, $\max\limits_{a \in \act(s)} ~ \min\limits_{s' \in \post(s,a)} \hat{\rank}(s') = i$, which means that $\rank(s)$ is set to $i$ by the algorithm.

        Now, let $\rank(s) = i$. This means that at some point during the execution of the algorithm, the current vector $\hat{\rank} \in \exnats^\states$ holds values such that $1 + \max\limits_{a \in \act(s)} ~ \min\limits_{s' \in \post(s,a)} \hat{\rank}(s') = i$. This means that for all $a \in \act(s)$ there is an $s' \in \post(s,a)$ such that $\hat{\rank}(s')=i-1$. Since finite values are never overwritten by the algorithm, also $\rank(s')=i-1$ holds for all such $s'$ and by induction hypothesis $\fp{\distop{\max}{}}(s') = i-1$. By definition it follows that $\fp{\distop{\max}{}}(s) = i$.
    \end{itemize}
    Note that from the induction it also follows that 
    \[
        \rank(s) = \infty \Longleftrightarrow \fp{\distop{\max}{}}(s) = \infty
    \]
    which concludes the proof.\qed
\end{proof}

\begin{algorithm}[t]
    \Input{Finite MDP $\mdp = \mdptuple$, target states $\target \subseteq S$}
    \Output{$\rank \in \exnats^\states$ such that $\rank = \lfp{\distopmod{\opt}{}}$}
    \lForEach(\tcp*[h]{$Q$ is FIFO}){$s$ with $\propt_s(\lozenge\target) = 1$}{$\rank(s) \gets \infty$\,;~$Q.\mathtt{push}(s)$}
    \lForEach{$s$ with $\propt_s(\lozenge\target) < 1$}{$\rank(s) \gets 0$}
    \While{$Q$ is not empty}{
        $\hat{s} \gets Q.\mathtt{pop}()$\;
        \ForEach{$s \in \pre(\hat{s})$}{%
            $\textnormal{tmp} \gets \underset{a \in \act(s)}{\opt} ~ \Big( ~ \underset{s' \in \post(s,a)}{\min}  \rank(s') ~+~$\;
             $~~~~~~~~~~~~~~~~~~~~~~~[\exists s',s'' \in \post(s,a): \rank(s')\neq\rank(s'')]\Big)$\;
            \If{$\rank(s) \ne \textnormal{tmp}$}{$\rank(s) \gets \textnormal{tmp}$\,;~$Q.\mathtt{push}(s)$}
        }%
    }
    \caption{Computation of the least fixed point of the modified distance operator}
    \label{alg:rankfpdistopmodcompute}
\end{algorithm}

\begin{restatable}[Computing the Least Fixed Point of the Complementary Distance Operator]{proposition}{algrankfpdistopmodcompute}
    \label{prop:algrankfpdistopmodcompute}
    \Cref{alg:rankfpdistopmodcompute} terminates and computes $\rank \in \exnats^\states$ with $\rank = \lfp{\distopmod{\opt}{}}$.
\end{restatable}
\begin{proof}
\label{app:algrankfpdistopmodcompute}
    Since the algorithm only updates $\rank$ values by applying the $\distopmod{\opt}{}$ function and $\lfp{\distopmod{\opt}{}} = \sup{\{{\distopmod{\opt}{}}^n(0) ~|~ n \in \nats\}}$, it is an invariant that $\hat{\rank}(s) \leq \rank(s)$ holds for all vectors $\hat{\rank}$ during the execution of the algorithm and all $s \in \states$. ($\dagger$)

    According to \Cref{lem:distopmodequiv}, we have that $\rank(s) = \infty$, if and only if $\propt_s(\lozenge\target) = 1$. Because of ($\dagger$) and the fact that the algorithm initializes the rank of exactly all $s$ with $\propt_s(\lozenge\target) = 1$ by $\infty$, we get that the algorithm correctly outputs $\infty$ for exactly those $s$ with $\rank(s) = \infty.$
    Now, we show by induction on $n \in \nats$, that also the ranks of all states $s$ with $\rank(s) = n$ are correctly computed by the algorithm.
    \begin{itemize}
        \item Induction start: $n = 0$. Follows immediately from ($\dagger$) and the fact that the algorithm initializes the ranks of all $s$ with $\propt_s(\lozenge\target) < 1$ by $0$.
        
        \item Induction Step: $n > 0$. Let $s \in \states$ with $\rank(s) = n > 0$. Since $\rank$ is a fixed point, it holds that 
        \[
            \underset{a \in \act(s)}{\opt} ~\Big(~ \underset{s' \in \post(s,a)}{\min} \rank(s') + [\exists s',s'' \in \post(s,a): \rank(s')\neq\rank(s'')]\Big) = n.
        \]
        
        If $\underset{a \in \act(s)}{\opt} ~ \underset{s' \in \post(s,a)}{\min} \rank(s') = n$, then the algorithm considers $s$ at a time when the current vector $\hat{\rank}$ satisfies
        \[
            \underset{a \in \act(s)}{\opt} ~\Big(~ \underset{s' \in \post(s,a)}{\min} \hat{\rank}(s') + [\exists s',s'' \in \post(s,a): \hat{\rank}(s')\neq\hat{\rank}(s'')]\Big) \geq n.
        \]
        Because of ($\dagger$), the algorithm correctly sets $\rank(s) = n$.
        
        If $\underset{a \in \act(s)}{\opt} ~ \underset{s' \in \post(s,a)}{\min} \rank(s') = n-1$, then there exists a successor $s'$ of $s$ such that $\rank(s') = n-1$ and by induction hypothesis the algorithm correctly outputs $\rank(s') = n-1$ too. Since $s$ is a predecessor of $s'$, $s$ is surely considered by the algorithm at a time when $\hat{\rank(s')} = n-1$ is already correctly set. Then, either $[\exists s',s'' \in \post(s,a): \hat{\rank}(s')\neq\hat{\rank}(s'')] = 1$ and $\rank(s) = n$ is correctly set by the algorithm or $[\exists s',s'' \in \post(s,a): \hat{\rank}(s')\neq\hat{\rank}(s'')] = 0$, but then some successor $s''$ will get a rank update at a later time by the algorithm, and thereafter $[\exists s',s'' \in \post(s,a): \hat{\rank}(s')\neq\hat{\rank}(s'')] = 1$ holds and $\rank(s) = n$ is correctly set by the algorithm.
    \end{itemize}
    Termination follows from the fact that states are inserted into the queue only a finite amount of times. This is because each state only gets inserted into the queue after its rank increased, and the rank of all states $s$ with $\rank(s) = \infty$ is already correctly set on initialization.\qed
\end{proof}

\section{Additional Experiments}\label{app:experiments}
\begin{figure}[t]
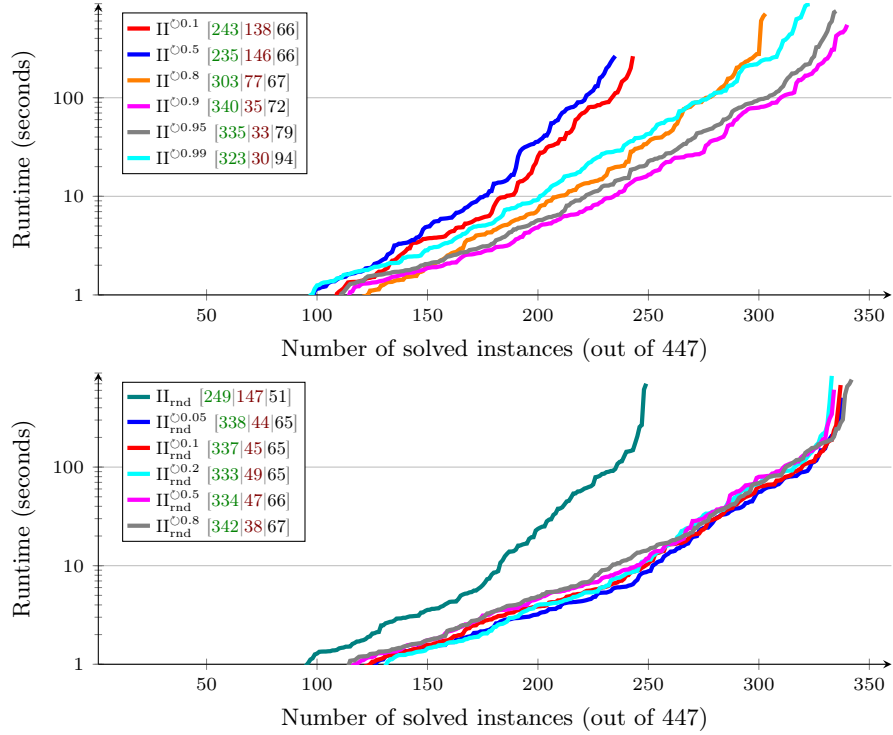

    \centering
    \quantileplot{plotdata/quantile.csv}{
        logs.Storm.topofpsmoothii10/plotred,
        logs.Storm.topofpsmoothii50/plotblue,
        logs.Storm.topofpsmoothii80/plotorange,
        logs.Storm.topofpsmoothii90/plotpink,
        logs.Storm.topofpsmoothii95/plotdarkgray,
        logs.Storm.topofpsmoothii99/plotcyan
    }{
        \siionedata,
        \siifivedata,
        \siieightdata,
        \siininedata,
        \siininefivedata,
        \siinineninedata
    }{1}{360}{1}{900}{north west}\\
    \quantileplot{plotdata/quantile.csv}{
        logs.Storm.topofproundii00/plotteal,
        logs.Storm.topofproundii05/plotblue,
        logs.Storm.topofproundii10/plotred,
        logs.Storm.topofproundii20/plotcyan,
        logs.Storm.topofproundii50/plotpink,
        logs.Storm.topofproundii80/plotdarkgray
    }{
        \riidata,
        \riizerofivedata,
        \riionedata,
        \riitwodata,
        \riifivedata,
        \riieightdata
    }{1}{360}{1}{900}{north west}
    \caption{More Runtime comparison of algorithms for computing certificates.}
    \label{appendix:fig:quantile}
\end{figure}

\Cref{appendix:fig:quantile} shows more plots similar to the one in \Cref{fig:RQ1} (left).
We see that for \sii{\gamma}, larger $\gamma$ values yield better results, where $\gamma = 0.9$ performs best among the ones we have considered.
For \rii[\gamma], the $\gamma$-smooth Bellman operator yields significant improvements over the standard bellman operator. Compared to \sii{\gamma}, however, the performance of \rii[\gamma] is less sensitive to the hyper-parameter $\gamma$.

}{} 

\end{document}